\PassOptionsToPackage{unicode}{hyperref}
\PassOptionsToPackage{hyphens}{url}
\PassOptionsToPackage{dvipsnames,svgnames,x11names}{xcolor}
\documentclass[12pt]{article}

\usepackage{amsmath,amssymb}
\usepackage{iftex}
\ifPDFTeX
  \usepackage[T1]{fontenc}
  \usepackage[utf8]{inputenc}
  \usepackage{textcomp}
\else
  \usepackage{unicode-math}
  \defaultfontfeatures{Scale=MatchLowercase}
  \defaultfontfeatures[\rmfamily]{Ligatures=TeX,Scale=1}
\fi
\usepackage{lmodern}
\IfFileExists{upquote.sty}{\usepackage{upquote}}{}
\IfFileExists{microtype.sty}{%
  \usepackage[]{microtype}
  \UseMicrotypeSet[protrusion]{basicmath}
}{}

\usepackage{graphicx}
\usepackage{subcaption}
\usepackage{booktabs}
\usepackage{longtable,array}
\usepackage{multirow}
\usepackage{multicol}
\usepackage{makecell}
\usepackage{xcolor}
\usepackage[normalem]{ulem}
\usepackage{wrapfig}
\usepackage{enumerate}
\usepackage{bbm,bm}
\usepackage{times,enumitem}
\usepackage{upgreek}
\usepackage{caption}
\usepackage{amsfonts}
\usepackage{nicefrac}
\usepackage{mathtools}
\usepackage{amsthm}
\usepackage{url}

\usepackage[]{natbib}

\usepackage{hyperref}
\usepackage{bookmark}
\IfFileExists{xurl.sty}{\usepackage{xurl}}{}
\hypersetup{
  pdftitle={Diffusion-Guided Counterfactual Distribution Learning with Intrinsic-Dimension Adaptivity},
  pdfauthor={Anonymous},
  pdfkeywords={3 to 6 keywords, that do not appear in the title},
  colorlinks=true,
  linkcolor={blue},
  filecolor={Maroon},
  citecolor={Blue},
  urlcolor={Blue},
  pdfcreator={LaTeX}
}

\usepackage[capitalize,noabbrev]{cleveref}

\theoremstyle{plain}
\newtheorem{theorem}{Theorem}[section]
\newtheorem{proposition}[theorem]{Proposition}
\newtheorem{lemma}[theorem]{Lemma}
\newtheorem{corollary}[theorem]{Corollary}
\theoremstyle{definition}

\newtheorem{assumption}[theorem]{Assumption}
\theoremstyle{remark}

\newcommand{\E}{\mathbb{E}}

\newcommand{\Pn}{\mathbb{P}_n}

\newcommand{\Pb}{\mathbb{P}}

\newcommand{\var}{\text{var}}

\DeclareMathOperator*{\argmin}{arg\,min}

\makeatletter
\newcommand{\customlabel}[2]{%
   \protected@write \@auxout {}{\string \newlabel {#1}{{#2}{\thepage}{#2}{#1}{}} }%
   \hypertarget{#1}{#2}
}
\makeatother

\newcommand{\anon}{1}

\begin{document}

\def\spacingset#1{\renewcommand{\baselinestretch}%
{#1}\small\normalsize} \spacingset{1}


\if1\anon
{
  \title{\bf Geometry Adaptive Counterfactual Distribution Learning with Diffusion-Guided Smoothing}
  \author{
    Kwangho Kim\thanks{The author gratefully acknowledges support from the Samsung Science and Technology Foundation under Project Number SSTF-BA2502-01, and thanks Edward Kennedy for many insightful comments and suggestions.}\hspace{.2cm}\\
    Department of Statistics, Korea University
  }
  \maketitle
} \fi

\if0\anon
{
  \bigskip
  \bigskip
  \bigskip
  \begin{center}
    {\LARGE\bf Geometry Adaptive Counterfactual Distribution Learning with Diffusion-Guided Smoothing}
  \end{center}
  \medskip
} \fi

\bigskip
\begin{abstract}
We study counterfactual distribution learning for high-dimensional outcomes
whose laws may concentrate near lower-dimensional structure. Standard isotropic
smoothing ignores this geometry, leading to unfavorable scaling and unstable
local inference. We propose semiparametrically debiased, diffusion-guided
estimators for smoothed counterfactual densities and their scores. These
estimators combine causal nuisance adjustment with geometry-adaptive
localization driven by a learned diffusion score, yielding second-order
nuisance remainders while aligning smoothing with local outcome geometry. We
derive asymptotic expansions, integrated risk bounds, and simultaneous
inference for smoothed densities and Stein functionals, with extensions to
ambient density and score targets under additional approximation conditions.
The variance term in the risk bounds is governed by the concentration of the smoothing
operator: suitable geometric conditions yield intrinsic rather than ambient
scaling, while an explicit drift term quantifies the cost of estimating the
geometry. CelebA-based semi-synthetic experiments show faster error decay and improved stability for geometry-adaptive one-step methods, illustrating their applicability to high-dimensional embeddings.
\end{abstract}

\noindent%
\vfill

\newpage
\spacingset{1.5} 

\section{Introduction}

Counterfactuals ask what would have happened under a specified intervention,
whether or not it was actually taken, and are foundational for scientific
explanation and data driven decision making. In many applications, however, a mean or risk contrast does not fully describe the scientific target; policy relevant features may instead appear in tail
probabilities, dispersion, multimodality, subgroup heterogeneity, or local
changes in the counterfactual distribution.

Recent work has begun to learn unobserved counterfactual distributions and
general functionals directly, but existing approaches leave a gap between
flexible generative models and semiparametric tools. Structural or graphical
generative methods can model counterfactual outcomes mechanistically, often
under a specified directed acyclic graph, but are not primarily designed for
statistical efficiency, robustness, or confidence bands
\citep[e.g.,][]{pawlowski2020deep,sanchez2022diffusion}. Other generative
methods provide flexible counterfactual sampling, but typically emphasize sample
quality or global distributional discrepancies rather than local density and
score inference
\citep[e.g.,][]{wu2024counterfactual,ma2024diffpo}. Semiparametric methods move
closer to inference by targeting smoothed counterfactual functionals and using
influence-function correction to reduce first-order sensitivity to nuisance
estimation
\citep[e.g.,][]{kim2018causal,kennedy2023semiparametric,martinezcounterfactual23}.
However, these methods generally rely on smoothing operators whose local shape
is fixed in ambient coordinates. The difficulty is most pronounced when outcomes are ultrahigh dimensional but the relevant counterfactual distribution concentrates near lower dimensional structure, as in
images, trajectories, or rich clinical measurements. In such settings,
isotropic smoothing can allocate bandwidth to directions along which the law has
little local variation, leading to a severe ambient dimensional bias variance
tradeoff.

In parallel, score-based diffusion models have become powerful tools for
high-dimensional distribution learning and sampling
\citep{vincent2011connection,ho2020denoising,song2019generative,song2021sde,song2021maximum}. Score fields provide local geometric information useful for data-adaptive neighborhoods and, under low-dimensional structure, can exhibit intrinsic rather than ambient complexity
\citep[e.g.,][]{hyvarinen2005,pidstrigach2022score,tang2024adaptivity,liang2024diffusion}. These advances in generative modeling have also begun to influence
counterfactual learning; for example, \citet{luedtke2025doublegen} use
debiasing corrections to train generative models under confounding. Yet using diffusion geometry for counterfactual inference is nontrivial. The
learned geometry changes the representation of the counterfactual distribution, and
thus the target being estimated. Standard one-step corrections remove
first-order bias from causal nuisance estimation, but do not automatically
account for this geometry-induced drift.

This paper develops diffusion-guided methods for counterfactual distribution
learning and inference. The core idea is to use diffusion geometry to smooth the
counterfactual distribution along directions of effective variation, rather than treating
all ambient coordinates equally. When outcomes concentrate near lower-dimensional
structure or exhibit anisotropic curvature, this can reduce the effective
smoothing complexity. We combine this geometry-guided smoothing with
semiparametric one-step estimation, so causal nuisance errors enter through
second-order product remainders, while geometry-learning error appears as an
explicit target drift. The resulting question is when the intrinsic gain from
geometry outweighs the cost of learning it.

Our contributions address this question in three steps. First, we formalize
well-posed smoothed counterfactual targets, including density and score targets,
that remain meaningful for ultrahigh-dimensional or nearly singular outcome
laws. Second, we develop two influence-function-based estimators that use
diffusion-guided smoothing to encode outcome geometry and one-step corrections
to remove first-order nuisance bias. Third, we derive expansions, risk bounds,
and confidence procedures that separate oracle stochastic error, smoothing bias,
causal nuisance error, and geometry-learning drift. This decomposition clarifies
when intrinsic rather than ambient scaling is achievable, and what additional
control is needed when the geometry is learned.

\section{Counterfactual Laws and Smoothed Targets}
\label{sec:prelim}

We adopt the potential outcomes framework used in
counterfactual density, distributional, and generative learning
\citep[e.g.,][]{kim2018causal,kennedy2023semiparametric,luedtke2025doublegen}.
Let \(Z=(X,A,Y)\sim\Pb\) be an independent observation, where
\(A\in\mathcal A\) is a discrete intervention,
\(X\in\mathcal X\subseteq\mathbb R^k\) is a covariate vector, and
\(Y\in\mathcal Y\subseteq\mathbb R^d\) is a potentially high dimensional
outcome. For each \(a\in\mathcal A\), let \(Y^a\) denote the potential outcome
under intervention \(A=a\), and define the marginal counterfactual law
\[
\Pb_a(B)=\Pb(Y^a\in B),\quad B\in\mathcal B(\mathbb R^d),
\]
where \(\mathcal B(\mathbb R^d)\) is the Borel sigma-field. This law is the primitive causal object from which all targets below are derived.

Throughout, we assume the standard identification conditions of consistency, conditional exchangeability, and positivity \citep[][Chapter 12]{imbens2015causal}:
$
Y=Y^A,
\,
Y^a\perp A\mid X,
\,
\pi_a(X)\equiv \Pb(A=a\mid X)\ge \pi_{\min}>0
\, \text{ a.s.}, \forall a \in \mathcal{A}.
$
Then, for any integrable
\(f:\mathbb R^d\to\mathbb R\), $\E\{f(Y^a)\}
=
\E\left[
\frac{\mathbbm 1\{A=a\}}{\pi_a(X)}f(Y)
\right]$, so the marginal counterfactual law is identified as
\[
\Pb_a(B)
=
\E\left[
\frac{\mathbbm 1\{A=a\}}{\pi_a(X)}
\mathbbm 1\{Y\in B\}
\right].
\]

Because \(\Pb_a\) may concentrate near, or on, a lower-dimensional set, it need
not admit an ordinary density with respect to \(d\)-dimensional Lebesgue
measure. We therefore target smoothed functionals of \(\Pb_a\). Let \(h>0\) be a smoothing scale and let
\(K_h:\mathbb R^d\times\mathbb R^d\to[0,\infty)\) be measurable and satisfy $\int_{\mathbb R^d}K_h(y\mid u)dy=1$ for every \(u\in\mathbb R^d\). For \(B\in\mathcal B(\mathbb R^d)\), the kernel induces the smoothed counterfactual law
\[
\Pb_{a,h}(B)
=
\int
\left\{
\int_B K_h(y\mid u)dy
\right\}
d\Pb_a(u),
\]
with ambient smoothed density
\begin{align}
\label{eqn:kernel-smoothing}
p_{a,h}(y)
=
\int K_h(y\mid u)d\Pb_a(u)
=
\E\{K_h(y\mid Y^a)\}.
\end{align}

The target \(p_{a,h}\) should not be confused with an ordinary ambient density
of \(\Pb_a\). If \(\Pb_a\) concentrates on a lower-dimensional structure, for
example an \(m\)-dimensional submanifold of \(\mathbb R^d\), then \(\Pb_a\) may
be singular in the ambient space even though it has meaningful intrinsic
structure. The smoothed density \(p_{a,h}\) is instead an ambient representation
of \(\Pb_a\), well defined for every fixed \(h>0\). Its statistical complexity
is governed by the smoothing kernel: isotropic kernels treat all \(d\)
coordinates symmetrically and can inherit ambient-dimensional scaling
\citep[e.g.,][]{kim2018causal,kennedy2023semiparametric}, whereas the
geometry-adaptive kernels studied below can concentrate along lower-dimensional
directions of variation.

We also consider the score of the smoothed counterfactual law. Define
\[
g_{a,h}(y)
=
\E\{\nabla_y K_h(y\mid Y^a)\}.
\]
Whenever differentiation under the integral sign is valid and \(p_{a,h}(y)>0\),
\begin{align}
\label{eqn:score-smoothing}
s_{a,h}(y)
=
\nabla_y\log p_{a,h}(y)
=
\frac{g_{a,h}(y)}{p_{a,h}(y)}.
\end{align}
The score provides local directional information for downstream tasks such as
Stein functionals, transport, and sampling. Under additional structure, its
statistical complexity may also depend on intrinsic rather than ambient
dimension \citep{pidstrigach2022score,chen2023score,tang2024adaptivity}.

Sections~\ref{sec:dis} and~\ref{sec:dss} specialize \(K_h\) to diffusion-guided smoothing operators, develop debiased estimators for the resulting density and score targets, and derive error decompositions that isolate causal nuisance remainders from geometry-induced drift.

\section{Diffusion-informed Density Smoothing}
\label{sec:dis}

\subsection{Estimation}
\label{subsec:dis-estimation}

We now define a diffusion-guided choice of \(K_h\) and construct the
corresponding one-step estimator for \(p_{a,h}\).
Classical smoothed counterfactual density estimators use ambient-coordinate
isotropic kernels, such as \(K_h(y\mid u)=h^{-d}K\{(y-u)/h\}\)
\citep{robins2001comment,kim2018causal,kennedy2023semiparametric}. Another valid choice is the analytic Gaussian perturbation induced by a forward
diffusion transition,
\(K_h(y\mid u)=q_{\varepsilon_h}(y\mid u)
=\mathcal N\{y;\bar m_{\varepsilon_h}(u),\bar\Sigma_{\varepsilon_h}\}\),
where \(\varepsilon_h\) is the diffusion time associated with scale \(h\), and
\(\bar m_{\varepsilon_h}\) and \(\bar\Sigma_{\varepsilon_h}\) are determined by
the forward drift and diffusion scale
\citep{song2021sde}. These choices define valid smoothed targets, but their local shape is fixed by the ambient coordinates or by the forward diffusion, rather than by the learned geometry of the counterfactual law. Consequently, their stochastic behavior can
still be governed by ambient-dimensional smoothing.

Our construction instead uses the diffusion score to warp the smoothing kernel.
Let \(q_{\varepsilon_h}(\cdot\mid u)\) be the transition density of a forward
diffusion \(dY_t=b(Y_t,t)dt+\sigma(t)dW_t\), with diffusion time
\(\varepsilon_h\) associated with scale \(h\). Let \(p_{\theta_0,t}\) denote the time-\(t\) density obtained by
forward-diffusing a common reference law for shared geometry or \(\Pb_a\) for
intervention-specific geometry. Define
\(s_{\theta_0}^{\rm diff}(z,t)=\nabla_z\log p_{\theta_0,t}(z)\), distinct from
\(s_{a,h}\) in \eqref{eqn:score-smoothing}. Inspired by the
probability-flow construction of score-based diffusion models \citep{song2021sde}, let
\(\Phi_{\varepsilon_h,\theta_0}\) be the reverse probability-flow map generated
by \(b(z,t)-\sigma(t)^2s_{\theta_0}^{\rm diff}(z,t)/2\), from time \(\varepsilon_h\) to
time \(0\), and write
\(J_{\varepsilon_h,\theta_0}(y)
=
|\det\{\nabla_y\Phi_{\varepsilon_h,\theta_0}^{-1}(y)\}|\). We define
\begin{equation}
\label{eqn:diff-informed-bump}
\kappa_{h,\theta_0}(y;u)
=
q_{\varepsilon_h}\!\left(
\Phi_{\varepsilon_h,\theta_0}^{-1}(y)
\mid u
\right)
J_{\varepsilon_h,\theta_0}(y).
\end{equation}
Thus \(\Phi_{\varepsilon_h,\theta_0}^{-1}(y)\) evaluates the forward density at
the noisy counterpart of \(y\), while \(J_{\varepsilon_h,\theta_0}(y)\) corrects
volume. Since \(\Phi_{\varepsilon_h,\theta_0}\) is driven by the diffusion score
\(s_{\theta_0}^{\rm diff}\), it can reshape local neighborhoods according to
the geometry of the reference law. Hence
\(y\mapsto\kappa_{h,\theta_0}(y;u)\) remains an ambient smoothing density, but
its local shape can align with lower-dimensional directions of variation
rather than fixed ambient coordinates.

The resulting diffusion-guided smoothed counterfactual density is
\begin{align} \label{eqn:pop-smoothed-density}
p^{\mathrm{geo}}_{a,h}(y)
=
\E\{\kappa_{h,\theta_0}(y;Y^a)\}
=
\int \kappa_{h,\theta_0}(y;u)\,d\Pb_a(u),
\end{align}
and is identified as
\[
p^{\mathrm{geo}}_{a,h}(y)
=
\E\left[
\E\{\kappa_{h,\theta_0}(y;Y)\mid X,A=a\}
\right].
\]

For fixed geometry \(\theta_0\) and smoothing scale \(h\), define
\(\pi_a(x)=\Pb(A=a\mid X=x)\),
\(\mu_{a,h,\theta_0}(x;y)=
\E\{\kappa_{h,\theta_0}(y;Y)\mid X=x,A=a\}\), and let
\(\eta_0=\{\pi_a,\mu_{a,h,\theta_0},\kappa_{h,\theta_0}\}\). The uncentered
influence-function map for \(p^{\mathrm{geo}}_{a,h}(y)\) is
\[
\varphi^{\mathrm{geo}}_h(Z;y,\eta_0)
=
\frac{\mathbbm 1\{A=a\}}{\pi_a(X)}
\{\kappa_{h,\theta_0}(y;Y)-\mu_{a,h,\theta_0}(X;y)\}
+
\mu_{a,h,\theta_0}(X;y),
\]
with the efficient influence function
\(\phi^{\mathrm{geo}}_h(Z;y,\eta_0)
=
\varphi^{\mathrm{geo}}_h(Z;y,\eta_0)-p^{\mathrm{geo}}_{a,h}(y)\).

We propose the diffusion-informed smoothing (DIS) estimator as
\begin{align}
\widehat p^{\mathrm{geo}}_{a,h}(y)
&=
\Pn\!\left[
\frac{\mathbbm 1\{A=a\}}{\widehat\pi_a(X)}
\left\{
\kappa_{h,\widehat\theta}(y;Y)
-
\widehat\mu_{a,h,\widehat\theta}(X;y)
\right\}
+
\widehat\mu_{a,h,\widehat\theta}(X;y)
\right],
\label{eqn:estimator-one-step-diff-smoothing}
\end{align}
where \(\widehat\theta\) indexes the estimated score geometry used to construct
the kernel, while \(\widehat\pi_a\) and \(\widehat\mu_{a,h,\widehat\theta}\) estimate the nuisance functions \(\pi_a(x)\) and \(\mu_{a,h,\widehat\theta}(x;y)\), respectively.

We implement DIS using \(K\)-fold cross-fitting. For each fold \(k\), we
estimate \(\widehat\theta^{(-k)}\), \(\widehat\pi_a^{(-k)}\), and
\(\widehat\mu_{a,h,\widehat\theta^{(-k)}}^{(-k)}\) on the complementary sample,
evaluate the one-step map on fold \(k\), and aggregate the fold estimates with
weights \(n_k/n\)
\citep[e.g.,][]{chernozhukov2018dml,kennedy2024semiparametric}.
Alternatively, the geometry may be pretrained or learned from an independent
auxiliary sample. Under cross-fitting, \(\widehat\theta\) denotes the collection of fold-specific fits \(\{\widehat\theta^{(-k)}\}_{k=1}^K\). Any \(\widehat\theta\)-indexed target or remainder is the \(n_k/n\)-weighted average of its fold-specific counterparts, while fitted functions, norms, and rate conditions are interpreted foldwise and uniformly over fixed \(K\). Fold indices are suppressed when unambiguous, and all conditioning is fold specific.

In \eqref{eqn:estimator-one-step-diff-smoothing}, \(\widehat\theta\) enters
through the estimated diffusion score \(s_{\widehat\theta}^{\rm diff}\) and the
induced smoothing density \(\kappa_{h,\widehat\theta}\). Thus learning the geometry amounts to
replacing \(\kappa_{h,\theta_0}\) by
\(\kappa_{h,\widehat\theta}\). In practice, \(s_{\widehat\theta}^{\rm diff}\) may be learned via denoising score matching \citep{song2019generative}. The geometry may be shared across interventions or tailored to \(\Pb_a\), for example through inverse-propensity-weighted score matching
\citep{wu2024counterfactual}. The theory below is
modular with respect to this learning step: it depends only on the perturbation
from replacing \(\theta_0\) by \(\widehat\theta\), not on the details of a
specific score-training algorithm.

\subsection{Theoretical Properties}
\label{subsec:DIS-theory}

We analyze DIS through two ingredients: the concentration of the oracle
diffusion-guided kernel and the risk decomposition of the resulting one-step
estimator.

\subsubsection{Geometry Adaptive Localization}
\label{subsec:localization}

We first provide localization properties of \(\kappa_{h,\theta_0}\): normalization
as a smoothing density and local concentration controlling stochastic error. Thus any dimension reduction comes from the geometry of the smoothing operator itself: variance improves when the kernel localizes along the
active directions of counterfactual variation rather than spreading across
all ambient coordinates.

Let \(m_h(Y)\) and \(\Sigma_h(Y)\) denote the mean and covariance matrix of the
density \(u\mapsto\kappa_{h,\theta_0}(u;Y)\). Assume \(\Sigma_h(Y)\) is
invertible, and write \(G_h(Y)=\Sigma_h(Y)^{-1}\) and
\(\|v\|_{G_h(Y)}=\{v^\top G_h(Y)v\}^{1/2}\). The matrix \(G_h(Y)\) describes the
local shape of the smoothing kernel centered at \(Y\).

\begin{assumption}[Sub-Gaussian localization]
\label{assump:subgauss-kappa}
There exist constants \(c_0,C_0>0\) such that, for all \(h>0\), all
\(y\in\mathbb R^d\), and almost every \(Y\),
\[
\kappa_{h,\theta_0}(y;Y)
\le
C_0\det\{G_h(Y)\}^{1/2}
\exp\{-c_0\|y-m_h(Y)\|_{G_h(Y)}^2\}.
\]
\end{assumption}

Assumption~\ref{assump:subgauss-kappa} is an anisotropic localization condition:
the kernel is dominated by a Gaussian envelope whose shape is determined by the
local precision matrix. Similar conditions appear in heat-kernel and
variable-bandwidth diffusion-kernel analyses
\citep[e.g.,][]{berry2016variable,de2022riemannian}. The following lemma summarizes the properties needed for the subsequent risk analysis.

\begin{lemma}
\label{lem:kappa-basic}
For the fixed \(h>0\), assume that, for each \(Y\),
\(\Phi_{\varepsilon_h,\theta_0}\) in \eqref{eqn:diff-informed-bump} is a
\(C^1\) diffeomorphism of \(\mathbb R^d\) onto itself, and that
\(q_{\varepsilon_h}(\cdot\mid Y)\) is a probability density. Then
\[
\int_{\mathbb R^d}\kappa_{h,\theta_0}(u;Y)du=1.
\]
Moreover, if
\(E_h(Y;c)=\{u:\{u-m_h(Y)\}^\top G_h(Y)\{u-m_h(Y)\}\le c\}\), then
\[
\mathrm{Vol}\{E_h(Y;c)\}
=
\mathrm{Vol}(B_d)c^{d/2}\det\{G_h(Y)\}^{-1/2},
\]
where \(B_d\) is the unit Euclidean ball in \(\mathbb R^d\). If
Assumption~\ref{assump:subgauss-kappa} also holds, then, for a constant
\(C<\infty\) depending only on \(c_0,C_0\), and \(d\),
\[
\int_{\mathbb R^d}\kappa_{h,\theta_0}(u;Y)^2du
\le
C\det\{G_h(Y)\}^{1/2}.
\]
\end{lemma}

Lemma~\ref{lem:kappa-basic} shows that \(\kappa_{h,\theta_0}\) is normalized,
that \(\det\{G_h(Y)\}^{-1/2}\) is the local smoothing volume, and that
\(\det\{G_h(Y)\}^{1/2}\) controls the squared \(L_2\) concentration of the
oracle kernel. These are oracle-geometry properties; the learned-geometry error from replacing
\(\theta_0\) by \(\widehat\theta\) is analyzed in the next subsection.

\subsubsection{Asymptotics and Risk Bounds}
\label{subsec:asymptotics-risk}

We now derive pointwise expansions and integrated risk bounds for the DIS estimator.
Throughout this subsection, \(h>0\) is fixed, but we keep its dependence explicit
because it controls the stochastic, smoothing, and geometry-learning terms when
\(h=h_n\). The analysis separates causal nuisance error from the target drift
induced by replacing \(\theta_0\) with \(\widehat\theta\). We impose two regularity conditions.

\begin{assumption}[Integrability]
\label{assumption:boundedness}
For each fixed \(y\in\mathbb R^d\), $\E\{\kappa_{h,\theta_0}(y;Y)^2\}<\infty$ and
$\E\{\mu_{a,h,\theta_0}(X;y)^2\}<\infty$. Moreover, the corresponding foldwise conditional square-integrability conditions hold with probability tending to one.
\end{assumption}

\begin{assumption}[Empirical process control]
\label{assumption:sample-splitting}
Either the estimator uses \(K\)-fold cross fitting, with each evaluation fold
independent of the corresponding training sample, or the relevant
influence-function class satisfies suitable empirical-process conditions. 
\end{assumption}

Assumption~\ref{assumption:boundedness} is a mild integrability condition ensuring that the influence function is well behaved. Assumption~\ref{assumption:sample-splitting} controls empirical process terms involving estimated nuisance functions and the learned geometry \citep[e.g.,][]{kennedy2016semiparametric}. 

Let \(\mathcal I_1,\ldots,\mathcal I_K\) be the evaluation folds, with
\(n_k=|\mathcal I_k|\asymp n\), and let \(\mathcal F_{-k}\) be the
sigma-field generated by the corresponding training sample. For any
\(\mathcal F_{-k}\)-measurable fold-specific function \(f^{(-k)}\), write
\(\Pb f^{(-k)}=\E\{f^{(-k)}(Z)\mid\mathcal F_{-k}\}\), where
\(Z\sim\Pb\) is an independent draw. Expectations under \(\Pb_a\), conditional
moments given \(X\) and \(A\), and \(L_2(\Pb)\) and \(L_2(\Pb_a)\) norms
involving fitted functions are interpreted analogously. All sample-splitting
arguments are applied conditional on \(\mathcal F_{-k}\) within fold \(k\)
and then aggregated unconditionally over fixed \(K\), with fold indices
suppressed when unambiguous.

For any fixed geometry parameter \(\theta\), write
\[
p^{\mathrm{geo}}_{a,h,\theta}(y)
=
\E_{\Pb_a}\{\kappa_{h,\theta}(y;Y^a)\},
\quad
\mu_{a,h,\theta}(x;y)
=
\E\{\kappa_{h,\theta}(y;Y)\mid X=x,A=a\}.
\]
In particular,
\(p^{\mathrm{geo}}_{a,h}=p^{\mathrm{geo}}_{a,h,\theta_0}\).
Under the preceding convention,
\(p^{\mathrm{geo}}_{a,h,\widehat\theta}\) denotes the fold-weighted aggregate
of the fitted-geometry targets
\(p^{\mathrm{geo}}_{a,h,\widehat\theta^{(-k)}}\). The geometry-induced target
drift is
\[
B_{\theta,h}(y)
=
p^{\mathrm{geo}}_{a,h,\widehat\theta}(y)
-
p^{\mathrm{geo}}_{a,h,\theta_0}(y).
\]
For fixed or independently trained geometry, the notation has its ordinary
single-geometry meaning. Likewise,
\(R_{\pi,\mu}^{\widehat\theta}(y)\) denotes the fold-weighted aggregate of the
fold-specific product remainders.

We begin with a pointwise one-step expansion for the DIS estimator.

\begin{theorem}
\label{thm:onestep-expansion}
Fix \(h>0\) and \(y\in\mathbb R^d\). Suppose
Assumptions~\ref{assumption:boundedness} and~\ref{assumption:sample-splitting}
hold, and
\[
\|\widehat\pi_a-\pi_a\|_{L_2(\Pb)}=o_\Pb(1),
\quad
\|\widehat\mu_{a,h,\widehat\theta}(\cdot;y)
-
\mu_{a,h,\widehat\theta}(\cdot;y)\|_{L_2(\Pb)}
=o_\Pb(1),
\]
\[
\|\kappa_{h,\widehat\theta}(y;Y^a)
-
\kappa_{h,\theta_0}(y;Y^a)\|_{L_2(\Pb_a)}
=o_\Pb(1),
\quad
\Pb\{\inf_x\widehat\pi_a(x)\ge\pi_{\min}/2\}\to1 .
\]
Then the DIS estimator in \eqref{eqn:estimator-one-step-diff-smoothing} satisfies
\[
\widehat p^{\mathrm{geo}}_{a,h}(y)
-
p^{\mathrm{geo}}_{a,h}(y)
=
(\Pn-\Pb)\{\varphi^{\mathrm{geo}}_h(Z;y,\eta_0)\}
+
B_{\theta,h}(y)
+
R_{\pi,\mu}^{\widehat\theta}(y)
+
o_\Pb(n^{-1/2}),
\]
where $|R_{\pi,\mu}^{\widehat\theta}(y)|
\lesssim
\|\widehat\pi_a-\pi_a\|_{L_2(\Pb)}
\|\widehat\mu_{a,h,\widehat\theta}(\cdot;y)
-
\mu_{a,h,\widehat\theta}(\cdot;y)\|_{L_2(\Pb)}$.
\end{theorem}
The expansion in Theorem~\ref{thm:onestep-expansion} separates the
second-order causal nuisance remainder from the geometry-induced target drift
\(B_{\theta,h}(y)\). The latter is not removed by the one-step correction
because the correction is orthogonal to causal nuisance perturbations for a
fixed smoothing rule, not to perturbations of the geometry defining that rule.
Appendix~\ref{appsec:pointwise-dis} gives the corresponding fixed-geometry
oracle expansion obtained when \(\widehat\theta=\theta_0\). We next control
\(B_{\theta,h}(y)\) through a pointwise stability condition for
\(\kappa_{h,\widehat\theta}-\kappa_{h,\theta_0}\).

Let \(n_{\rm geom}\) denote the geometry-training sample size, with
\(n_{\rm geom}\asymp n\) under cross-fitting and equal to the auxiliary-sample
size otherwise. Let \(\rho_{\theta,n_{\rm geom}}\) denote the corresponding
diffusion-score or flow learning rate in the norm required for kernel
stability, with dependence on \(\varepsilon_h\) suppressed.
Suppose that, for the fixed \(y\) and some amplification
exponent \(\lambda_{\rm geom}(y)\ge0\), uniformly over
\(k=1,\ldots,K\),
\begin{align}
\label{eqn:pointwise-flow-to-kernel-stability}
\|\kappa_{h,\widehat\theta^{(-k)}}(y;Y^a)
-
\kappa_{h,\theta_0}(y;Y^a)\|_{L_2(\Pb_a)}
=
O_\Pb\{h^{-\lambda_{\rm geom}(y)}\rho_{\theta,n_{\rm geom}}\}.
\end{align}
This modular pointwise stability condition is the local analogue of
Assumption~\ref{assump:integrated-geometry-drift}. A mean-value expansion of the
probability-flow representation of \(\kappa_{h,\theta}\)
\citep{song2021sde} gives such a bound, with
\(\lambda_{\rm geom}(y)\) absorbing local kernel concentration,
spatial-derivative scaling, and Jacobian sensitivity. Appendix~\ref{appsec:geometry-drift}
details this connection and discusses structural regimes in which
\(\rho_{\theta,n_{\rm geom}}\) can scale with intrinsic rather than ambient
complexity.

Because \(B_{\theta,h}(y)\) is a fold-weighted average, the preceding uniform
bound gives
\(|B_{\theta,h}(y)|
=
O_\Pb\{h^{-\lambda_{\rm geom}(y)}\rho_{\theta,n_{\rm geom}}\}\). Therefore,
Theorem~\ref{thm:onestep-expansion} yields
\[
\widehat p^{\mathrm{geo}}_{a,h}(y)
-
p^{\mathrm{geo}}_{a,h}(y)
=
(\Pn-\Pb)\{\varphi^{\mathrm{geo}}_h(Z;y,\eta_0)\}
+
R_{\pi,\mu}^{\widehat\theta}(y)
+
O_\Pb\{h^{-\lambda_{\rm geom}(y)}\rho_{\theta,n_{\rm geom}}\}
+
o_\Pb(n^{-1/2}),
\]
thus learned geometry contributes a pointwise target drift governed by the score or flow learning error and its amplification through \(\kappa_{h,\theta}\).

We next turn to integrated error. For a measurable evaluation region
\(\mathcal Y_0\subset\mathbb R^d\) with finite Lebesgue measure, define
\[
H_h(\mathcal Y_0)
=
\E\left\{
\int_{\mathcal Y_0}
\kappa_{h,\theta_0}(y;Y^a)^2dy
\right\}.
\]
This concentration functional is the primitive variance complexity in the
integrated risk bounds below. It concerns the oracle smoothing rule; geometric
conditions under which \(H_h(\mathcal Y_0)\) scales intrinsically are given
after the risk theorem.

The cost of learning the geometry enters separately through the following
modular drift condition. Here \(\rho_{\theta,n_{\rm geom}}\) is the score or flow
learning rate as defined above, and \(\Lambda_{\rm geom}(\mathcal Y_0)\) denotes the integrated
amplification exponent induced by the smoothing construction.

\begin{assumption}[Integrated geometry drift control]
\label{assump:integrated-geometry-drift}
For \(\mathcal Y_0\), there exists a geometry
amplification exponent \(\Lambda_{\rm geom}(\mathcal Y_0)\ge0\) such that,
uniformly over \(k=1,\ldots,K\),
\[
\E_{\Pb_a}\left[
\int_{\mathcal Y_0}
\{\kappa_{h,\widehat\theta^{(-k)}}(y;Y^a)
-
\kappa_{h,\theta_0}(y;Y^a)\}^2dy
\right]
=
O_\Pb\!\left(
h^{-\Lambda_{\rm geom}(\mathcal Y_0)}
\rho_{\theta,n_{\rm geom}}^2
\right).
\]
\end{assumption}

Jensen's inequality within each fold and across the fold weights gives $\int_{\mathcal Y_0}B_{\theta,h}(y)^2dy
=
O_\Pb\!\left(
h^{-\Lambda_{\rm geom}(\mathcal Y_0)}
\rho_{\theta,n_{\rm geom}}^2
\right)$. Appendix~\ref{appsec:geometry-drift} gives sufficient foldwise stability
conditions for Assumption~\ref{assump:integrated-geometry-drift}. Let
\(\Delta_\Phi\) and \(\Delta_J\) denote uniform inverse-flow and
relative-Jacobian errors on \(\mathcal Y_0\), and let \(D_\kappa\) be the
root-integrated kernel concentration exponent along the flow path. Applying
Lemma~\ref{lem:flow-kernel-stability} uniformly over the fixed folds shows
that \(\Delta_\Phi+\Delta_J
=
O_\Pb(\rho_{\theta,n_{\rm geom}})\) yields
\(\Lambda_{\rm geom}(\mathcal Y_0)=D_\kappa+2\), while a purely multiplicative
Jacobian perturbation gives
\(\Lambda_{\rm geom}(\mathcal Y_0)=D_\kappa\). The appendix also relates these
flow errors to score and score-gradient errors through ordinary differential
equation stability. Appendix~\ref{appsec:intrinsic-score-learning} then gives
structural regimes under which \(\rho_{\theta,n_{\rm geom}}\) may scale with
intrinsic rather than ambient complexity.

The next result gives the integrated risk decomposition for both the
fixed-geometry and learned-geometry DIS estimators. For the fixed-geometry
version, let \(\widehat p^{\mathrm{geo},0}_{a,h}\) denote the DIS estimator 
\eqref{eqn:estimator-one-step-diff-smoothing} evaluated with
\(\widehat\theta=\theta_0\) and
\(\widehat\mu_{a,h,\widehat\theta}\) replaced by
\(\widehat\mu_{a,h,\theta_0}\). Let \(R_{\pi,\mu}^0(y)\) be the corresponding
product remainder, with \(\vert R_{\pi,\mu}^0(y)\vert \lesssim
\|\widehat\pi_a-\pi_a\|_{L_2(\Pb)}
\|\widehat\mu_{a,h,\theta_0}(\cdot;y)-\mu_{a,h,\theta_0}(\cdot;y)\|_{L_2(\Pb)}\).

\begin{theorem}
\label{thm:dis-risk-global}
For fixed \(h>0\), let \(\mathcal Y_0\subset\mathbb R^d\) be a measurable
evaluation region with finite Lebesgue measure. Suppose
Assumptions~\ref{assumption:boundedness} and~\ref{assumption:sample-splitting}
hold. Assume further that
\[
\|\widehat\pi_a-\pi_a\|_{L_2(\Pb)}=o_\Pb(1),\quad
\int_{\mathcal Y_0}
\|\widehat\mu_{a,h,\theta_0}(\cdot;y)
-\mu_{a,h,\theta_0}(\cdot;y)\|_{L_2(\Pb)}^2dy=o_\Pb(1),
\]
\[
\Pb\{\inf_{x\in\mathcal X}\widehat\pi_a(x)\ge \pi_{\min}/2\}\to1,
\quad
\sup_x
\E\!\left[
\int_{\mathcal Y_0}\kappa_{h,\theta_0}(y;Y)^2dy
\,\middle|\,
X=x,A=a
\right]<\infty .
\]
Then the fixed-geometry estimator satisfies
\[
\int_{\mathcal Y_0}
\{\widehat p^{\mathrm{geo},0}_{a,h}(y)-p^{\mathrm{geo}}_{a,h}(y)\}^2dy
=
O_\Pb\left(
\frac{H_h(\mathcal Y_0)}{n}
+
\operatorname*{ess\,sup}_{y\in\mathcal Y_0}\{R_{\pi,\mu}^0(y)\}^2
\right)
+
o_\Pb(n^{-1}).
\]

For the learned-geometry DIS estimator, suppose additionally that
\[
\int_{\mathcal Y_0}
\|\widehat\mu_{a,h,\widehat\theta}(\cdot;y)
-\mu_{a,h,\widehat\theta}(\cdot;y)\|_{L_2(\Pb)}^2dy=o_\Pb(1),
\]
\[
\|\widehat\pi_a-\pi_a\|_{L_2(\Pb)}^2
\sup_x
\E\!\left[
\int_{\mathcal Y_0}\kappa_{h,\widehat\theta}(y;Y)^2dy
\,\middle|\,
X=x,A=a
\right]
=o_\Pb(1),
\]
\[
\int_{\mathcal Y_0}
\|\kappa_{h,\widehat\theta}(y;Y^a)
-\kappa_{h,\theta_0}(y;Y^a)\|_{L_2(\Pb_a)}^2dy
=o_\Pb(1).
\]
Then
\[
\int_{\mathcal Y_0}
\{\widehat p^{\mathrm{geo}}_{a,h}(y)-p^{\mathrm{geo}}_{a,h}(y)\}^2dy
=
O_\Pb\left(
\frac{H_h(\mathcal Y_0)}{n}
+
\operatorname*{ess\,sup}_{y\in\mathcal Y_0}
\{R_{\pi,\mu}^{\widehat\theta}(y)\}^2
+
\int_{\mathcal Y_0}B_{\theta,h}(y)^2dy
\right)
+
o_\Pb(n^{-1}).
\]
\end{theorem}

Theorem~\ref{thm:dis-risk-global} keeps the geometry-induced target drift
\(\int_{\mathcal Y_0}B_{\theta,h}(y)^2dy\) explicit. If this term is controlled
by Assumption~\ref{assump:integrated-geometry-drift}, the following bound follows
immediately.

\begin{corollary}
\label{cor:risk-geometry-stability}
Suppose the conditions of Theorem~\ref{thm:dis-risk-global} and
Assumption~\ref{assump:integrated-geometry-drift} hold. Then
\[
\int_{\mathcal Y_0}
\{\widehat p^{\mathrm{geo}}_{a,h}(y)
-p^{\mathrm{geo}}_{a,h}(y)\}^2dy
=
O_\Pb\left(
\frac{H_h(\mathcal Y_0)}{n}
+
\operatorname*{ess\,sup}_{y\in\mathcal Y_0}
\{R_{\pi,\mu}^{\widehat\theta}(y)\}^2
+
h^{-\Lambda_{\rm geom}(\mathcal Y_0)}
\rho_{\theta,n_{\rm geom}}^2
\right)
+
o_\Pb(n^{-1}).
\]
\end{corollary}

Any external score- or flow-learning guarantee can be combined with
Corollary~\ref{cor:risk-geometry-stability}. For example, if
\(\rho_{\theta,n_{\rm geom}}=O_\Pb(n_{\rm geom}^{-a(k)})\) under a structural
complexity parameter \(k\), then the geometry-drift term becomes
\(h^{-\Lambda_{\rm geom}(\mathcal Y_0)}n_{\rm geom}^{-2a(k)}\). Thus external
diffusion-learning theory determines \(\rho_{\theta,n_{\rm geom}}\), while the
present semiparametric analysis determines how this rate affects counterfactual
density estimation; see Appendices~\ref{appsec:geometry-drift}
and~\ref{appsec:intrinsic-score-learning} for more details.

It remains to identify when the oracle stochastic term
\(H_h(\mathcal Y_0)/n\) scales intrinsically rather than ambiently. The
following proposition provides an oracle tangent-kernel benchmark: if
\(\Pb_a\) is supported on an \(m\)-dimensional manifold and the kernel contracts
at scale \(h\) only along tangent directions, then
\(H_h^{\rm tan}(\mathbb R^d)\) scales as \(h^{-m}\), rather than \(h^{-d}\).

\begin{proposition}
\label{prop:manifold-concentration}
Suppose \(\Pb_a\) is supported on a compact \(m\)-dimensional \(C^2\)
submanifold \(M\subset\mathbb R^d\), with \(m\ll d\). For \(u\in M\), let
\(P_T(u)\) and \(P_N(u)=I_d-P_T(u)\) be the orthogonal projections onto
\(T_uM\) and \(T_uM^\perp\), respectively. Fix \(\sigma_0>0\), and define
\(\kappa_h^{\rm tan}(y;u)
=
\phi_d\{y-u;\,h^2P_T(u)+\sigma_0^2P_N(u)\}\), where
\(\phi_d(\cdot;\Sigma)\) denotes the \(N(0,\Sigma)\) density. For measurable
\(\mathcal Y_0\subseteq\mathbb R^d\), let
\(H_h^{\rm tan}(\mathcal Y_0)
=
\E[\int_{\mathcal Y_0}\kappa_h^{\rm tan}(y;Y^a)^2dy]\). Then, with
\(C_{d,m,\sigma_0}=(4\pi)^{-d/2}\sigma_0^{-(d-m)}\),
\[
H_h^{\rm tan}(\mathbb R^d)
=
C_{d,m,\sigma_0}h^{-m},
\quad
H_h^{\rm tan}(\mathcal Y_0)
\le
C_{d,m,\sigma_0}h^{-m}.
\]
\end{proposition}

The next theorem abstracts this mechanism through the local covariance: intrinsic
scaling follows when only a few active directions contract at scale \(h\).

\begin{theorem}
\label{thm:covariance-concentration}
Suppose Assumption~\ref{assump:subgauss-kappa} holds and that, for some
\(h_r>0\), deterministic function \(r(h)\ge0\), and \(C_r<\infty\), $\E[\det\{G_h(Y^a)\}^{1/2}]
\le
C_r h^{-r(h)}, \, 0<h\le h_r$. Then, for any measurable \(\mathcal Y_0\subseteq\mathbb R^d\),
\[
H_h(\mathcal Y_0)
\le
H_h(\mathbb R^d)
\lesssim
h^{-r(h)} .
\]
In particular, suppose that, for some \(d_\star\ll d\), the covariance
\(\Sigma_h(Y^a)=G_h(Y^a)^{-1}\) has \(d_\star\) eigenvalues of order \(h^2\)
and \(d-d_\star\) eigenvalues of order one for \(\Pb_a\)-almost every \(Y^a\),
uniformly over sufficiently small \(h\). Then
\[
H_h(\mathcal Y_0)
\le
H_h(\mathbb R^d)
\lesssim
h^{-d_\star}.
\]
If, additionally, $\int_{\mathbb R^d}\kappa_{h,\theta_0}(u;Y^a)^2du
\gtrsim
\det\{G_h(Y^a)\}^{1/2}$ for \(\Pb_a\)-almost every \(Y^a\), uniformly over sufficiently small \(h\), then
\[
H_h(\mathbb R^d)\asymp h^{-d_\star}.
\]
\end{theorem}

The preceding results illustrate that kernel concentration is the primitive
complexity measure. If, for some \(C<\infty\) and \(D_H\ge0\),
\[
H_h(\mathcal Y_0)\le C h^{-D_H},
\]
then Theorem~\ref{thm:dis-risk-global} gives an oracle stochastic term of order
\(O_\Pb(n^{-1}h^{-D_H})\). Under Theorem~\ref{thm:covariance-concentration}, one may take
\(D_H=d_\star\), while the tangent-kernel benchmark in
Proposition~\ref{prop:manifold-concentration} has concentration exponent \(m\).
Thus intrinsic scaling corresponds to a concentration exponent much smaller
than \(d\), and the learned estimator inherits this gain when the geometry-drift
term in Corollary~\ref{cor:risk-geometry-stability} is controlled.

Although the primary target is the fixed-\(h\) smoothed law
\(p^{\mathrm{geo}}_{a,h}\), the same decomposition extends to an ordinary
ambient density when one exists. Suppose that \(\Pb_a\) admits a Lebesgue
density \(p_a\) on a neighborhood of \(\mathcal Y_0\), and that
\[
A_h^2(\mathcal Y_0)
\equiv
\int_{\mathcal Y_0}
\{p^{\mathrm{geo}}_{a,h}(y)-p_a(y)\}^2dy
\le
C_Ah^{2\beta}.
\]
This is an \(L_2\) smoothing-bias condition analogous to those used in kernel
density estimation under Hölder or Sobolev smoothness
\citep[e.g.,][]{tsybakov2009introduction,gine2016mathematical}, with the
additional requirement that the diffusion-guided smoothing rule retain
approximate-identity behavior. Combining this condition with
Theorem~\ref{thm:dis-risk-global} and
Assumption~\ref{assump:integrated-geometry-drift} gives
\[
\begin{aligned}
\int_{\mathcal Y_0}
\{\widehat p^{\mathrm{geo}}_{a,h}(y)-p_a(y)\}^2dy
&=
O_\Pb\!\left[
\frac{H_h(\mathcal Y_0)}{n}
+
\operatorname*{ess\,sup}_{y\in\mathcal Y_0}
\{R_{\pi,\mu}^{\widehat\theta}(y)\}^2 \right. \\
&\quad\left.
\quad +
h^{-\Lambda_{\rm geom}(\mathcal Y_0)}
\rho_{\theta,n_{\rm geom}}^2
+
h^{2\beta}
\right]
+
o_\Pb(n^{-1}).
\end{aligned}
\]
The fixed-geometry analogue follows by dropping the geometry-drift term and
replacing \(R_{\pi,\mu}^{\widehat\theta}\) by \(R_{\pi,\mu}^0\). Hence, if
\(H_h(\mathcal Y_0)\lesssim h^{-D_H}\), the oracle stochastic contribution is
\(n^{-1}h^{-D_H}\). If the preceding bounds hold uniformly along
\(h=h_n\to0\), and the causal nuisance and geometry-drift terms are negligible
relative to \(n^{-1}h_n^{-D_H}+h_n^{2\beta}\), then $h_n\asymp n^{-1/(2\beta+D_H)}$, and $\int_{\mathcal Y_0}
\{\widehat p^{\mathrm{geo}}_{a,h_n}(y)-p_a(y)\}^2dy
=
O_\Pb\!\left(
n^{-2\beta/(2\beta+D_H)}
\right)$. Formal statements and proofs are given in
Appendix~\ref{appsec:ambient-density}.

\section{Counterfactual Score Learning}
\label{sec:dss}

\subsection{Estimation}
\label{subsec:score-estimation}

We now apply the diffusion-guided smoothing construction of
Section~\ref{sec:dis} to the score target in \eqref{eqn:score-smoothing}. For a
fixed geometry parameter \(\theta\), write
\(\dot\kappa_{h,\theta}(y;u)=\nabla_y\kappa_{h,\theta}(y;u)\), and define
\[
P_{a,h,\theta}(y)
=
\E_{\Pb_a}\{\kappa_{h,\theta}(y;Y^a)\},
\quad
G_{a,h,\theta}(y)
=
\E_{\Pb_a}\{\dot\kappa_{h,\theta}(y;Y^a)\}.
\]
Whenever \(P_{a,h,\theta}(y)>0\), the corresponding smoothed score is
\[
s_{a,h,\theta}(y)
=
G_{a,h,\theta}(y)/P_{a,h,\theta}(y).
\]
The population-geometry target is
\(s^{\mathrm{geo}}_{a,h}(y)=s_{a,h,\theta_0}(y)\), with
\(P_{a,h,\theta_0}(y)=p^{\mathrm{geo}}_{a,h}(y)\). This fixed-\(h\) target
remains well defined even when \(\Pb_a\) has no ordinary ambient density. Now, define the localized regression functions
\[
\mu_{a,h,\theta}(x;y)
=
\E\{\kappa_{h,\theta}(y;Y)\mid X=x,A=a\},
\quad
\nu_{a,h,\theta}(x;y)
=
\E\{\dot\kappa_{h,\theta}(y;Y)\mid X=x,A=a\}.
\]
Then \(P_{a,h,\theta}(y)=\E\{\mu_{a,h,\theta}(X;y)\}\) and
\(G_{a,h,\theta}(y)=\E\{\nu_{a,h,\theta}(X;y)\}\).

For the nuisance collection $\eta^s_\theta
=
\{\pi_a,\mu_{a,h,\theta},\nu_{a,h,\theta},
\kappa_{h,\theta},\dot\kappa_{h,\theta}\}$, define the uncentered influence functions
\[
\varphi_P(Z;y,\eta^s_\theta)
=
\frac{\mathbbm 1\{A=a\}}{\pi_a(X)}
\{\kappa_{h,\theta}(y;Y)-\mu_{a,h,\theta}(X;y)\}
+
\mu_{a,h,\theta}(X;y),
\]
\[
\varphi_G(Z;y,\eta^s_\theta)
=
\frac{\mathbbm 1\{A=a\}}{\pi_a(X)}
\{\dot\kappa_{h,\theta}(y;Y)-\nu_{a,h,\theta}(X;y)\}
+
\nu_{a,h,\theta}(X;y).
\]
The corresponding centered influence functions are
\[
\phi_P(Z;y,\eta^s_\theta)
=
\varphi_P(Z;y,\eta^s_\theta)-P_{a,h,\theta}(y),
\quad
\phi_G(Z;y,\eta^s_\theta)
=
\varphi_G(Z;y,\eta^s_\theta)-G_{a,h,\theta}(y).
\]
Hence, for fixed \(\theta\), the efficient influence function for
\(s_{a,h,\theta}(y)\) is
\[
\phi_s(Z;y,\eta^s_\theta)
=
\frac{
\phi_G(Z;y,\eta^s_\theta)
-
s_{a,h,\theta}(y)\phi_P(Z;y,\eta^s_\theta)
}{
P_{a,h,\theta}(y)
}.
\]
At \(\theta=\theta_0\), this is the efficient influence function for
\(s^{\mathrm{geo}}_{a,h}(y)\). Let $\widehat\eta^s_{\widehat\theta}
=
\{\widehat\pi_a,\widehat\mu_{a,h,\widehat\theta},
\widehat\nu_{a,h,\widehat\theta},
\kappa_{h,\widehat\theta},\dot\kappa_{h,\widehat\theta}\}$, and define the one-step estimators
\[
\widehat P_{a,h}(y)
=
\Pn\{\varphi_P(Z;y,\widehat\eta^s_{\widehat\theta})\},
\quad
\widehat G_{a,h}(y)
=
\Pn\{\varphi_G(Z;y,\widehat\eta^s_{\widehat\theta})\}.
\]
Here \(\widehat P_{a,h}\) is exactly the DIS estimator in
\eqref{eqn:estimator-one-step-diff-smoothing}, and
\(\widehat G_{a,h}\) is its derivative analogue. We define the
diffusion-informed score smoothing (DSS) estimator by
\begin{align}
\label{eqn:dss-estimator}
\widehat s^{\mathrm{geo}}_{a,h}(y)
=
\widehat G_{a,h}(y)/\widehat P_{a,h}(y),
\end{align}
when \(\widehat P_{a,h}(y)>0\), and set it to zero otherwise.
A first-order ratio expansion shows that its influence function is
\(\phi_s\) above, so it is asymptotically equivalent to a direct one-step
estimator based on \(\phi_s\). Under cross fitting, \(\widehat P_{a,h}\) and
\(\widehat G_{a,h}\) are aggregated over held-out folds before taking the
ratio. 

\subsection{Theoretical Properties}
\label{subsec:DSS-theory}

We give the score analogues of the DIS pointwise expansion and integrated risk
bound. The results are stated for fixed \(h>0\), with the dependence on \(h\)
displayed because it governs both derivative concentration and geometry drift.

\begin{assumption}[Score regularity]
\label{assumption:score-regularity}
Let \(\mathcal I=\{y\}\) for pointwise results and
\(\mathcal I=\mathcal Y_0\) for integrated results. There exists
\(p'_{\min}>0\) such that
\[
\inf_{u\in\mathcal I}P_{a,h,\theta_0}(u)\ge p'_{\min},
\quad
\Pb\left\{
\inf_{u\in\mathcal I}
\big[
P_{a,h,\widehat\theta}(u)\wedge\widehat P_{a,h}(u)
\big]
\ge p'_{\min}/2
\right\}\to1.
\]
Moreover,
\(\sup_{u\in\mathcal I}
\{\|s^{\mathrm{geo}}_{a,h}(u)\|_2+
\|s_{a,h,\widehat\theta}(u)\|_2\}=O_\Pb(1)\), where the supremum is interpreted
as an essential supremum when \(\mathcal I=\mathcal Y_0\). Assume, uniformly over \(u\in\mathcal I\),
\(\varphi_P(\cdot;u,\eta^s_{\theta_0})\in L_2(\Pb)\) and
\(\varphi_G(\cdot;u,\eta^s_{\theta_0})\in L_2(\Pb)^d\), with foldwise
conditional analogues.
\end{assumption}

Let
\(\delta_\pi=\|\widehat\pi_a-\pi_a\|_{L_2(\Pb)}\), and define the
learned-geometry regression errors by
\[
\delta_\mu^{\widehat\theta}(y)
=
\|\widehat\mu_{a,h,\widehat\theta}(\cdot;y)-\mu_{a,h,\widehat\theta}(\cdot;y)\|_{L_2(\Pb)},
\quad
\delta_\nu^{\widehat\theta}(y)
=
\|\widehat\nu_{a,h,\widehat\theta}(\cdot;y)-\nu_{a,h,\widehat\theta}(\cdot;y)\|_{L_2(\Pb)}.
\]
Furthermore, define the joint kernel perturbation by
\[
\delta_{\kappa,s}^{\widehat\theta}(y)
=
\|\kappa_{h,\widehat\theta}(y;Y^a)-\kappa_{h,\theta_0}(y;Y^a)\|_{L_2(\Pb_a)}
+
\|\dot\kappa_{h,\widehat\theta}(y;Y^a)-\dot\kappa_{h,\theta_0}(y;Y^a)\|_{L_2(\Pb_a)}.
\]
Under cross fitting, let \(P_{a,h,\widehat\theta}\) and
\(G_{a,h,\widehat\theta}\) denote the fold-weighted population counterparts of
\(\widehat P_{a,h}\) and \(\widehat G_{a,h}\), set
\(s_{a,h,\widehat\theta}(y)
=
G_{a,h,\widehat\theta}(y)/P_{a,h,\widehat\theta}(y)\), and define
\(B_{s,\theta,h}(y)
=
s_{a,h,\widehat\theta}(y)-s_{a,h,\theta_0}(y)\)
as the score target drift induced by geometry learning. The next theorem gives
the pointwise expansion of the DSS estimator.

\begin{theorem}
\label{thm:dss-expansion}
Fix \(h>0\) and \(y\in\mathbb R^d\). Suppose
Assumptions~\ref{assumption:boundedness},
\ref{assumption:sample-splitting}, and
\ref{assumption:score-regularity} hold at \(y\). Suppose also that
\(\delta_\pi\), \(\delta_\mu^{\widehat\theta}(y)\),
\(\delta_\nu^{\widehat\theta}(y)\), and
\(\delta_{\kappa,s}^{\widehat\theta}(y)\) are all \(o_\Pb(1)\). Then the DSS
estimator in \eqref{eqn:dss-estimator} satisfies
\[
\widehat s^{\mathrm{geo}}_{a,h}(y)
-
s^{\mathrm{geo}}_{a,h}(y)
=
(\Pn-\Pb)\{\phi_s(Z;y,\eta^s_{\theta_0})\}
+
B_{s,\theta,h}(y)
+
R_{s,\pi,\mu,\nu}^{\widehat\theta}(y)
+
o_\Pb(n^{-1/2}),
\]
where $\|R_{s,\pi,\mu,\nu}^{\widehat\theta}(y)\|_2
=
O_\Pb\!\left[
\delta_\pi
\{\delta_\mu^{\widehat\theta}(y)
+
\delta_\nu^{\widehat\theta}(y)\}
\right]$.
\end{theorem}

Theorem~\ref{thm:dss-expansion} is the score analogue of
Theorem~\ref{thm:onestep-expansion}. It separates the second-order causal
nuisance remainder from the score target drift
\(B_{s,\theta,h}(y)\). Ratio algebra gives
\[
B_{s,\theta,h}(y)
=
\frac{
G_{a,h,\widehat\theta}(y)-G_{a,h,\theta_0}(y)
}
{P_{a,h,\widehat\theta}(y)}
-
s^{\mathrm{geo}}_{a,h}(y)
\frac{
P_{a,h,\widehat\theta}(y)-P_{a,h,\theta_0}(y)
}
{P_{a,h,\widehat\theta}(y)}.
\]
Therefore, under Assumption~\ref{assumption:score-regularity},
\[
\|B_{s,\theta,h}(y)\|_2
\lesssim
\delta_{\kappa,s}^{\widehat\theta}(y).
\]
In particular, if $\delta_{\kappa,s}^{\widehat\theta}(y)
=
O_\Pb\{h^{-\lambda_{s,\rm geom}(y)}
\rho_{\theta,n_{\rm geom}}\}$, then
\[
\|B_{s,\theta,h}(y)\|_2
=
O_\Pb\{h^{-\lambda_{s,\rm geom}(y)}
\rho_{\theta,n_{\rm geom}}\}.
\]
This stability argument is the derivative analogue of that for DIS, with
\(\rho_{\theta,n_{\rm geom}}\) interpreted in the stronger geometry-learning
norm required for DSS. Let \(\Delta_{\Phi,1}\) and \(\Delta_{J,1}\) denote the
combined level and first-spatial-derivative errors of the inverse flow and
Jacobian factor, respectively. The same mean-value argument gives
\[
\delta_{\kappa,s}^{\widehat\theta}(y)
\lesssim
h^{-\lambda_{s,\rm geom}(y)}
\{\Delta_{\Phi,1}+\Delta_{J,1}\}.
\]
Standard continuous-dependence bounds for the probability-flow ordinary
differential equation relate these errors to estimation errors in the diffusion
score and its first two spatial derivatives. Thus
\(\lambda_{s,\rm geom}(y)\) additionally captures the derivative sensitivity of
\(\dot\kappa_{h,\theta}\); see
Lemma~\ref{lem:flow-kernel-gradient-stability} in
Appendix~\ref{appsec:geometry-drift}.

For integrated error, define
\[
H_h^s(\mathcal Y_0)
=
\E\left[
\int_{\mathcal Y_0}
\left\{
\kappa_{h,\theta_0}(y;Y^a)^2
+
\|\dot\kappa_{h,\theta_0}(y;Y^a)\|_2^2
\right\}dy
\right].
\]
This is the score analogue of \(H_h(\mathcal Y_0)\), incorporating both kernel
and derivative concentration. The following DSS analogue of
Assumption~\ref{assump:integrated-geometry-drift} jointly controls
geometry-induced perturbations of the kernel and its gradient; sufficient
flow-based conditions are given in Appendix~\ref{appsec:geometry-drift}.

\begin{assumption}[Integrated score geometry drift control]
\label{assump:integrated-score-geometry-drift}
For \(\mathcal Y_0\), there exists an amplification exponent
\(\Lambda_{s,\rm geom}(\mathcal Y_0)\ge0\) such that, uniformly over
\(k=1,\ldots,K\),
\[
\begin{aligned}
&\E_{\Pb_a}\!\left[
\int_{\mathcal Y_0}
\left\{
\bigl(\kappa_{h,\widehat\theta^{(-k)}}-\kappa_{h,\theta_0}\bigr)^2
+
\bigl\|\dot\kappa_{h,\widehat\theta^{(-k)}}-\dot\kappa_{h,\theta_0}\bigr\|_2^2
\right\}(y;Y^a)\,dy
\right] \\
&\quad=
O_\Pb\!\left\{
h^{-\Lambda_{s,\rm geom}(\mathcal Y_0)}
\rho_{\theta,n_{\rm geom}}^2
\right\}.
\end{aligned}
\]
\end{assumption}

The next DSS risk theorem keeps the learned-geometry target drift explicit.
\begin{theorem}
\label{thm:dss-risk-global}
For fixed \(h>0\), let \(\mathcal Y_0\subset\mathbb R^d\) be measurable with
finite Lebesgue measure. Suppose
Assumptions~\ref{assumption:boundedness},
\ref{assumption:sample-splitting}, and
\ref{assumption:score-regularity} hold over \(\mathcal Y_0\). Assume that
\(\int_{\mathcal Y_0}
\|\widehat\mu_{a,h,\widehat\theta}(\cdot;y)
-\mu_{a,h,\widehat\theta}(\cdot;y)\|_{L_2(\Pb)}^2dy=o_\Pb(1)\) and
\(\int_{\mathcal Y_0}
\|\widehat\nu_{a,h,\widehat\theta}(\cdot;y)
-\nu_{a,h,\widehat\theta}(\cdot;y)\|_{L_2(\Pb)}^2dy=o_\Pb(1)\). Assume also that
\[
\begin{aligned}
&\|\widehat\pi_a-\pi_a\|_{L_2(\Pb)}^2
\sup_x
\E\!\left[
\int_{\mathcal Y_0}
\left\{
\kappa_{h,\widehat\theta}(y;Y)^2+
\|\dot\kappa_{h,\widehat\theta}(y;Y)\|_2^2
\right\}dy
\,\middle|\,
X=x,A=a
\right]
=o_\Pb(1),
\end{aligned}
\]
\[
\int_{\mathcal Y_0}
\left[
\|\kappa_{h,\widehat\theta}(y;Y^a)
-\kappa_{h,\theta_0}(y;Y^a)\|_{L_2(\Pb_a)}^2
+
\|\dot\kappa_{h,\widehat\theta}(y;Y^a)
-\dot\kappa_{h,\theta_0}(y;Y^a)\|_{L_2(\Pb_a)}^2
\right]dy
=o_\Pb(1).
\]
Then
\[
\begin{aligned}
\int_{\mathcal Y_0}
\|\widehat s^{\mathrm{geo}}_{a,h}(y)
-s^{\mathrm{geo}}_{a,h}(y)\|_2^2dy
&=
O_\Pb\!\left[
\frac{H_h^s(\mathcal Y_0)}{n}
+
\operatorname*{ess\,sup}_{y\in\mathcal Y_0}
\|R_{s,\pi,\mu,\nu}^{\widehat\theta}(y)\|_2^2 \right.\\
&\quad\left.
+
\int_{\mathcal Y_0}
\|B_{s,\theta,h}(y)\|_2^2dy
\right]
+
o_\Pb(n^{-1}).
\end{aligned}
\]
\end{theorem}

Combining this with
Assumption~\ref{assump:integrated-score-geometry-drift} gives the following
result.

\begin{corollary}
\label{cor:dss-risk-geometry-stability}
Suppose the conditions of Theorem~\ref{thm:dss-risk-global} and
Assumption~\ref{assump:integrated-score-geometry-drift} hold. Then
\[
\begin{aligned}
\int_{\mathcal Y_0}
\|\widehat s^{\mathrm{geo}}_{a,h}(y)
-
s^{\mathrm{geo}}_{a,h}(y)\|_2^2dy
=
O_\Pb\!\left[
\frac{H_h^s(\mathcal Y_0)}{n}
+
\operatorname*{ess\,sup}_{y\in\mathcal Y_0}
\|R_{s,\pi,\mu,\nu}^{\widehat\theta}(y)\|_2^2 \right.\\
\left.
+
h^{-\Lambda_{s,\rm geom}(\mathcal Y_0)}
\rho_{\theta,n_{\rm geom}}^2
\right]
+
o_\Pb(n^{-1}).
\end{aligned}
\]
\end{corollary}

The next theorem bounds \(H_h^s(\mathcal Y_0)\) through the active covariance
directions.

\begin{theorem}
\label{thm:dss-covariance-concentration}
Suppose Assumption~\ref{assump:subgauss-kappa} holds and, for some
\(c'_0,C'_0>0\),
\[
\|\dot\kappa_{h,\theta_0}(u;Y)\|_2
\le
C'_0\det\{G_h(Y)\}^{1/2}
\|G_h(Y)\{u-m_h(Y)\}\|_2
\exp\{-c'_0\|u-m_h(Y)\|_{G_h(Y)}^2\}.
\]
Then, for every measurable \(\mathcal Y_0\subseteq\mathbb R^d\),
\[
H_h^s(\mathcal Y_0)
\le
H_h^s(\mathbb R^d)
\lesssim
\E\left[
\det\{G_h(Y^a)\}^{1/2}
\{1+\operatorname{tr}G_h(Y^a)\}
\right].
\]
In particular, suppose that, for \(\Pb_a\)-almost every \(Y^a\), uniformly over sufficiently small \(h\), \(\Sigma_h(Y^a)\) has \(d_\star\) eigenvalues of order \(h^2\) and the rest of
order one. Then
\[
H_h^s(\mathcal Y_0)
\le
H_h^s(\mathbb R^d)
\lesssim
h^{-(d_\star+2)}.
\]
\end{theorem}

Theorem~\ref{thm:dss-covariance-concentration} motivates an admissible DSS
score-concentration exponent \(D_{\rm score}\ge0\) satisfying
\(H_h^s(\mathcal Y_0)\lesssim h^{-D_{\rm score}}\) for all sufficiently small
\(h\). Under its active-covariance condition, one may take
\(D_{\rm score}=d_\star+2\), whereas the corresponding DIS bound allows
\(D_H=d_\star\). The additional two powers of \(h^{-1}\) reflect the derivative
cost of score estimation, not a larger structural dimension.
Appendix~\ref{appsec:intrinsic-score-learning} gives regimes in which
\(\rho_{\theta,n_{\rm geom}}\) likewise depends on intrinsic rather than ambient
complexity.

When \(\Pb_a\) admits an ambient density \(p_a\) with score
\(s_a=\nabla\log p_a\), the fixed-\(h\) DSS risk bound extends to \(s_a\) under
the score-level counterpart of the DIS smoothing-bias condition,
\(A_{s,h}^2(\mathcal Y_0)
\equiv
\int_{\mathcal Y_0}
\|s^{\mathrm{geo}}_{a,h}(y)-s_a(y)\|_2^2dy
\le
C'_A h^{2\beta_s}\). If the nuisance and geometry-drift terms are negligible
uniformly along \(h=h_n\to0\), balancing variance and smoothing bias gives
\(h_n\asymp n^{-1/(2\beta_s+D_{\rm score})}\) and
\(\int_{\mathcal Y_0}
\|\widehat s^{\mathrm{geo}}_{a,h_n}(y)-s_a(y)\|_2^2dy
=
O_\Pb\{n^{-2\beta_s/(2\beta_s+D_{\rm score})}\}\).
Full details are given in Appendix~\ref{appsec:dss-ambient-score}.

The choice between DIS and DSS depends on the target and its structural
representation. DIS is natural for density recovery and density-based functionals, whereas DSS
directly targets score- and Stein-based functionals and supports
gradient-driven sampling, such as Langevin or reverse-diffusion updates,
without fitting a separately normalized generative density model. For
score-based downstream tasks, DSS can be statistically advantageous when the
reduction in intrinsic score complexity offsets the additional derivative cost.
Ignoring nuisance and geometry-learning errors, its optimized \(L_2\)
convergence exponent exceeds that of DIS when $D_{\rm score}<(\beta_s/\beta)D_H$. Thus DSS benefits most when the score admits substantially lower-dimensional
structure than the density, particularly when \(\beta_s\le\beta\). Because the
two rates concern score and density losses, respectively, this comparison is
target specific rather than a universal ranking of the two methods.

\section{Statistical Inference}
\label{sec:inference}

Our primary inferential targets are the population-geometry smoothed law
\(p^{\mathrm{geo}}_{a,h}=p^{\mathrm{geo}}_{a,h,\theta_0}\) and analogous
fixed-\(h\) score-based functionals. We first establish multiplier validity
conditional on geometry fixed relative to the analysis sample, and then
transfer coverage to population-geometry targets through a uniform drift
envelope.

Let \(\theta^\dagger\) denote either the population geometry \(\theta_0\), a
fixed pretrained geometry, or a geometry learned from an auxiliary sample
independent of the analysis sample; in the conditional results below,
\(\theta^\dagger\) is treated as fixed relative to the analysis sample.
Write
\(\widehat\eta_{\theta^\dagger}
=
\{\widehat\pi_a,\widehat\mu_{a,h,\theta^\dagger},
\kappa_{h,\theta^\dagger}\}\), and define
\(\widehat p^{\mathrm{geo}}_{a,h,\theta^\dagger}(y)
=
\Pn\{\varphi_h^{\mathrm{geo}}(Z;y,\widehat\eta_{\theta^\dagger})\}\)
and
\(\widehat\phi^{\mathrm{geo}}_{h,i}(y;\theta^\dagger)
=
\varphi_h^{\mathrm{geo}}(Z_i;y,\widehat\eta_{\theta^\dagger})
-
\widehat p^{\mathrm{geo}}_{a,h,\theta^\dagger}(y)\), with all estimated
quantities interpreted foldwise under cross fitting.

\begin{theorem}
\label{thm:dis-band}
Let \(\mathcal Y_0\subset\mathbb R^d\) be compact. Suppose, uniformly over
\(y\in\mathcal Y_0\),
\[
\widehat p^{\mathrm{geo}}_{a,h,\theta^\dagger}(y)
-
p^{\mathrm{geo}}_{a,h,\theta^\dagger}(y)
=
(\Pn-\Pb)\phi^{\mathrm{geo}}_h(Z;y,\eta_{\theta^\dagger})
+
o_\Pb(n^{-1/2}).
\]
Let
\(\sigma_{\theta^\dagger}^2(y)
=
\var\{\phi^{\mathrm{geo}}_h(Z;y,\eta_{\theta^\dagger})
\mid\theta^\dagger\}\), and suppose the studentized multiplier process
consistently approximates the law of
\[
\left\{
\frac{\sqrt n(\Pn-\Pb)
\phi^{\mathrm{geo}}_h(Z;y,\eta_{\theta^\dagger})}
{\sigma_{\theta^\dagger}(y)}
:
y\in\mathcal Y_0
\right\}
\]
in \(\ell^\infty(\mathcal Y_0)\), conditionally on \(\theta^\dagger\) when
random. Suppose also that the supremum law is continuous at its
\((1-\alpha)\)-quantile and, for some \(c_\sigma>0\),
\(\Pb\{\inf_{y\in\mathcal Y_0}\sigma_{\theta^\dagger}(y)\ge c_\sigma\}\to1\)
and
\(\sup_{y\in\mathcal Y_0}
|\widehat\sigma_{\theta^\dagger}(y)-\sigma_{\theta^\dagger}(y)|
=o_\Pb(1)\), with
\(\widehat\sigma_{\theta^\dagger}^2(y)
=
n^{-1}\sum_{i=1}^n
\widehat\phi^{\mathrm{geo}}_{h,i}(y;\theta^\dagger)^2\). Define
\[
\widehat{\mathbb Z}_n(y)
=
n^{-1/2}\sum_{i=1}^n
\xi_i\widehat\phi^{\mathrm{geo}}_{h,i}(y;\theta^\dagger),
\]
and let \(\widehat c_{1-\alpha}\) be the conditional \((1-\alpha)\)-quantile of
\(\sup_{y\in\mathcal Y_0}
|\widehat{\mathbb Z}_n(y)|/\widehat\sigma_{\theta^\dagger}(y)\). Then
\[
\widehat p^{\mathrm{geo}}_{a,h,\theta^\dagger}(y)
\pm
\widehat c_{1-\alpha}
\widehat\sigma_{\theta^\dagger}(y)/\sqrt n,
\quad y\in\mathcal Y_0,
\]
has asymptotic simultaneous coverage \(1-\alpha\) for
\(\{p^{\mathrm{geo}}_{a,h,\theta^\dagger}(y):y\in\mathcal Y_0\}\). When
\(\theta^\dagger\) is random, coverage holds conditionally in probability and
hence also unconditionally.
\end{theorem}

Theorem~\ref{thm:dis-band} establishes simultaneous coverage for the
\(\theta^\dagger\)-indexed target. A sufficient route is a uniform nuisance
product remainder of order \(o_\Pb(n^{-1/2})\), stochastic equicontinuity of
the influence-function class over \(\mathcal Y_0\), uniformly consistent and
nondegenerate variance estimation, and valid Gaussian or multiplier
approximation
\citep[e.g.,][]{kennedy2016semiparametric,van1996weak,chernozhukov2016empirical}.
For a fixed or slowly growing evaluation grid, these reduce to standard
high-dimensional Gaussian approximation conditions. The following corollary
transfers the fitted-geometry band to the population geometry using a uniform
drift envelope.

\begin{corollary}
\label{cor:population-geometry-band}
Let the conditions of Theorem~\ref{thm:dis-band} hold with
\(\theta^\dagger=\widehat\theta\), where \(\widehat\theta\) is learned from an
independent auxiliary sample. Suppose an auxiliary-sample-measurable envelope
\(\Delta_{\rm geom,n}\ge0\) satisfies
\(\sup_{y\in\mathcal Y_0}|p^{\mathrm{geo}}_{a,h,\widehat\theta}(y)
-p^{\mathrm{geo}}_{a,h,\theta_0}(y)|
\le\Delta_{\rm geom,n}\) with probability tending to one. Then
\[
\widehat p^{\mathrm{geo}}_{a,h}(y)
\pm
\left\{
\widehat c_{1-\alpha}\widehat\sigma_{\widehat\theta}(y)/\sqrt n
+
\Delta_{\rm geom,n}
\right\},
\quad y\in\mathcal Y_0,
\]
has asymptotic simultaneous coverage at least \(1-\alpha\) for
\(\{p^{\mathrm{geo}}_{a,h,\theta_0}(y):y\in\mathcal Y_0\}\). If \(\Delta_{\rm geom,n}=o_\Pb(n^{-1/2})\), the standard multiplier band
without geometry-drift inflation is valid for the population-geometry target.
\end{corollary}

For example, if, with probability tending to one, $\sup_{y\in\mathcal Y_0}|B_{\theta,h}(y)|
\le
C_{\rm geom}
h^{-\lambda_{\rm geom,\infty}(\mathcal Y_0)}
\rho_{\theta,n_{\rm geom}}$, then one may take
\(\Delta_{\rm geom,n}
=
C_{\rm geom}
h^{-\lambda_{\rm geom,\infty}(\mathcal Y_0)}
\rho_{\theta,n_{\rm geom}}\).
A uniform version of the interpolation argument in
Lemma~\ref{lem:flow-kernel-stability} of
Appendix~\ref{appsec:geometry-drift} gives sufficient flow-based conditions for
this envelope. Thus a sufficient condition for validity of the population-geometry band
without geometry-drift inflation is \(h^{-\lambda_{\rm geom,\infty}(\mathcal Y_0)}\rho_{\theta,n_{\rm geom}}=o_\Pb(n^{-1/2})\). For fixed \(h\), if
\(\rho_{\theta,n_{\rm geom}}=O_\Pb(n_{\rm geom}^{-1/2})\), this entails
\(n_{\rm geom}\gg n\). This restriction concerns only population-geometry
inference without geometry-drift inflation, not consistency or the preceding
risk bounds. Without a valid drift envelope, the standard multiplier band
should instead be interpreted as inference for the fitted-geometry target.

The same multiplier construction applies to Stein functionals of the DSS
score. For a probability measure \(\nu\) on \(\mathcal Y_0\) and a class
\(\mathcal G\) of smooth vector fields, define
\[
\begin{aligned}
\Psi_{a,h,\theta^\dagger}(g)
&=\int_{\mathcal Y_0}
\{\nabla_y\!\cdot g(y)+g(y)^\top s_{a,h,\theta^\dagger}(y)\}\,d\nu(y),\\
\widehat\Psi_{a,h,\theta^\dagger}(g)
&=\int_{\mathcal Y_0}
\{\nabla_y\!\cdot g(y)+g(y)^\top
\widehat s_{a,h,\theta^\dagger}(y)\}\,d\nu(y),
\end{aligned}
\]
where
\(\widehat s_{a,h,\theta^\dagger}
=\widehat G_{a,h,\theta^\dagger}/\widehat P_{a,h,\theta^\dagger}\).
Its influence function is
\(\phi_{\Psi,\theta^\dagger}(Z;g,\eta^s_{\theta^\dagger})
=\int_{\mathcal Y_0}g(y)^\top
\phi_s(Z;y,\eta^s_{\theta^\dagger})\,d\nu(y)\).
Let \(\widehat\phi_{s,i}(y;\theta^\dagger)\) denote the foldwise plug-in,
empirically centered version of \(\phi_s(Z_i;y,\eta^s_{\theta^\dagger})\), and set $\widehat\phi_{\Psi,\theta^\dagger,i}(g)
=
\int_{\mathcal Y_0}
g(y)^\top\widehat\phi_{s,i}(y;\theta^\dagger)\,d\nu(y)$. The next result is the DSS counterpart of
Theorem~\ref{thm:dis-band}.

\begin{theorem}
\label{thm:dss-stein-band}
Suppose, uniformly over \(g\in\mathcal G\),
\[
\widehat\Psi_{a,h,\theta^\dagger}(g)
-
\Psi_{a,h,\theta^\dagger}(g)
=
(\Pn-\Pb)
\phi_{\Psi,\theta^\dagger}(Z;g,\eta^s_{\theta^\dagger})
+
o_\Pb(n^{-1/2}).
\]
Suppose the multiplier approximation, quantile-continuity, and uniform
variance conditions of Theorem~\ref{thm:dis-band} hold with
\(\mathcal Y_0\), \(y\), and
\(\phi_h^{\mathrm{geo}}\) replaced by
\(\mathcal G\), \(g\), and \(\phi_{\Psi,\theta^\dagger}\), respectively.
Let
\(\widehat{\mathbb Z}_n^\Psi(g)
=
n^{-1/2}\sum_{i=1}^n
\xi_i\widehat\phi_{\Psi,\theta^\dagger,i}(g)\) and
\(\widehat\sigma_{\theta^\dagger}^2(g)
=
n^{-1}\sum_{i=1}^n
\widehat\phi_{\Psi,\theta^\dagger,i}(g)^2\),
and let \(\widehat c'_{1-\alpha}\) be the conditional \((1-\alpha)\)-quantile
of
\(\sup_{g\in\mathcal G}
|\widehat{\mathbb Z}_n^\Psi(g)|/\widehat\sigma_{\theta^\dagger}(g)\). Then
\[
\widehat\Psi_{a,h,\theta^\dagger}(g)
\pm
\widehat c'_{1-\alpha}
\widehat\sigma_{\theta^\dagger}(g)/\sqrt n,
\quad g\in\mathcal G,
\]
has asymptotic simultaneous coverage \(1-\alpha\) for
\(\{\Psi_{a,h,\theta^\dagger}(g):g\in\mathcal G\}\), conditionally in probability when \(\theta^\dagger\) is random, and hence also unconditionally.
\end{theorem}

As with DIS, population-geometry coverage follows by controlling the
geometry-induced functional drift.

\begin{corollary}
\label{cor:population-geometry-stein-band}
Let the conditions of Theorem~\ref{thm:dss-stein-band} hold with
\(\theta^\dagger=\widehat\theta\), where \(\widehat\theta\) is learned from an
independent auxiliary sample. Suppose an auxiliary-sample-measurable envelope
\(\Delta_{\Psi,n}\ge0\) satisfies
\(\sup_{g\in\mathcal G}|\Psi_{a,h,\widehat\theta}(g)
-\Psi_{a,h,\theta_0}(g)|\le\Delta_{\Psi,n}\) with probability tending to one. Then
\[
\widehat\Psi_{a,h,\widehat\theta}(g)
\pm
\left\{
\widehat c'_{1-\alpha}\widehat\sigma_{\widehat\theta}(g)/\sqrt n
+
\Delta_{\Psi,n}
\right\},
\quad g\in\mathcal G,
\]
has asymptotic simultaneous coverage at least \(1-\alpha\) for
\(\{\Psi_{a,h,\theta_0}(g):g\in\mathcal G\}\). If \(\Delta_{\Psi,n}=o_\Pb(n^{-1/2})\), the standard multiplier band without geometry-drift inflation is valid for the population-geometry functional.
\end{corollary}

If
\(\sup_{g\in\mathcal G}\int_{\mathcal Y_0}\|g(y)\|_2\,d\nu(y)
\le C_{\mathcal G}\) and
\(\sup_{y\in\mathcal Y_0}\|B_{s,\theta,h}(y)\|_2\le\Delta_{s,n}\), then $\sup_{g\in\mathcal G}
|\Psi_{a,h,\widehat\theta}(g)-\Psi_{a,h,\theta_0}(g)|
\le
C_{\mathcal G}\Delta_{s,n}$, so one may take
\(\Delta_{\Psi,n}=C_{\mathcal G}\Delta_{s,n}\). A uniform analogue of
Lemma~\ref{lem:flow-kernel-gradient-stability} in
Appendix~\ref{appsec:geometry-drift} gives sufficient flow-based conditions for
\(\Delta_{s,n}\).

Thus, for both DIS and DSS, under the fixed or independently learned geometry
regime above, the uninflated multiplier bands quantify empirical uncertainty
for the corresponding geometry-indexed targets, whereas coverage of the
population-geometry targets requires either asymptotically negligible geometry
drift or an explicit drift envelope. When a sufficiently regular
ambient density or score exists, the same construction extends to the
unsmoothed target under additional uniform control of the smoothing bias; we
omit this extension for brevity.

\section{Experiments}
\label{sec:experiments}

We conduct a semi-synthetic counterfactual experiment based on the CelebA
observational design of \citet{luedtke2025doublegen}. Each image has binary
treatment \(A\) indicating smiling, baseline confounders \(X\) given by CelebA
attributes, and outcome \(Y\) given by a fixed pretrained image embedding. To
stress-test ambient dimensionality, we construct
\(\widetilde Y_r\in\mathbb R^r\) for
\(r\in\{500,2500,10000\}\). Within each training fold, we estimate an
intervention-specific diffusion score by inverse-propensity-weighted denoising
score matching, solve the associated probability-flow ODE, and construct the
transported smoothing kernel. Geometry is learned in the \(r\)-dimensional
embedding space, while errors are evaluated on a fixed projection
\(Y_r^{\mathrm{proj}}\in\mathbb R^{r_{\mathrm{proj}}}\), with
\(r_{\mathrm{proj}}\in\{2,3\}\), over a compact region \(\mathcal Y_0\).
Geometry and causal nuisance functions are cross-fitted throughout.

We evaluate the estimators against independently constructed held-out smoothed
reference targets. Figure~\ref{fig:exp-density} compares density integrated
squared error (ISE) for an isotropic one-step estimator, a diffusion-guided
plug-in estimator, and DIS. The latter two use the same transported kernel, so
their comparison isolates the one-step correction. Each smoothing rule is
evaluated against its matched reference target; comparisons across smoothing
rules therefore concern relative error decay rather than a common estimand.
Figure~\ref{fig:exp-stein} compares mean squared error (MSE) of Stein
functionals over a fixed class \(\mathcal G\) for a treated-only plug-in estimator, DSS with an isotropic Gaussian kernel, and diffusion-guided DSS. The diffusion-guided estimators exhibit generally faster error decay relative to their matched targets. In the density experiment, DIS also improves on the
diffusion-guided plug-in estimator using the same kernel. Extended curves
across ambient dimensions and complete implementation details are provided in
Appendix~\ref{appsec:sim-details}.

\begin{figure}[t]
  \centering
  \begin{minipage}[t]{0.49\linewidth}
    \centering
    \includegraphics[width=\linewidth]{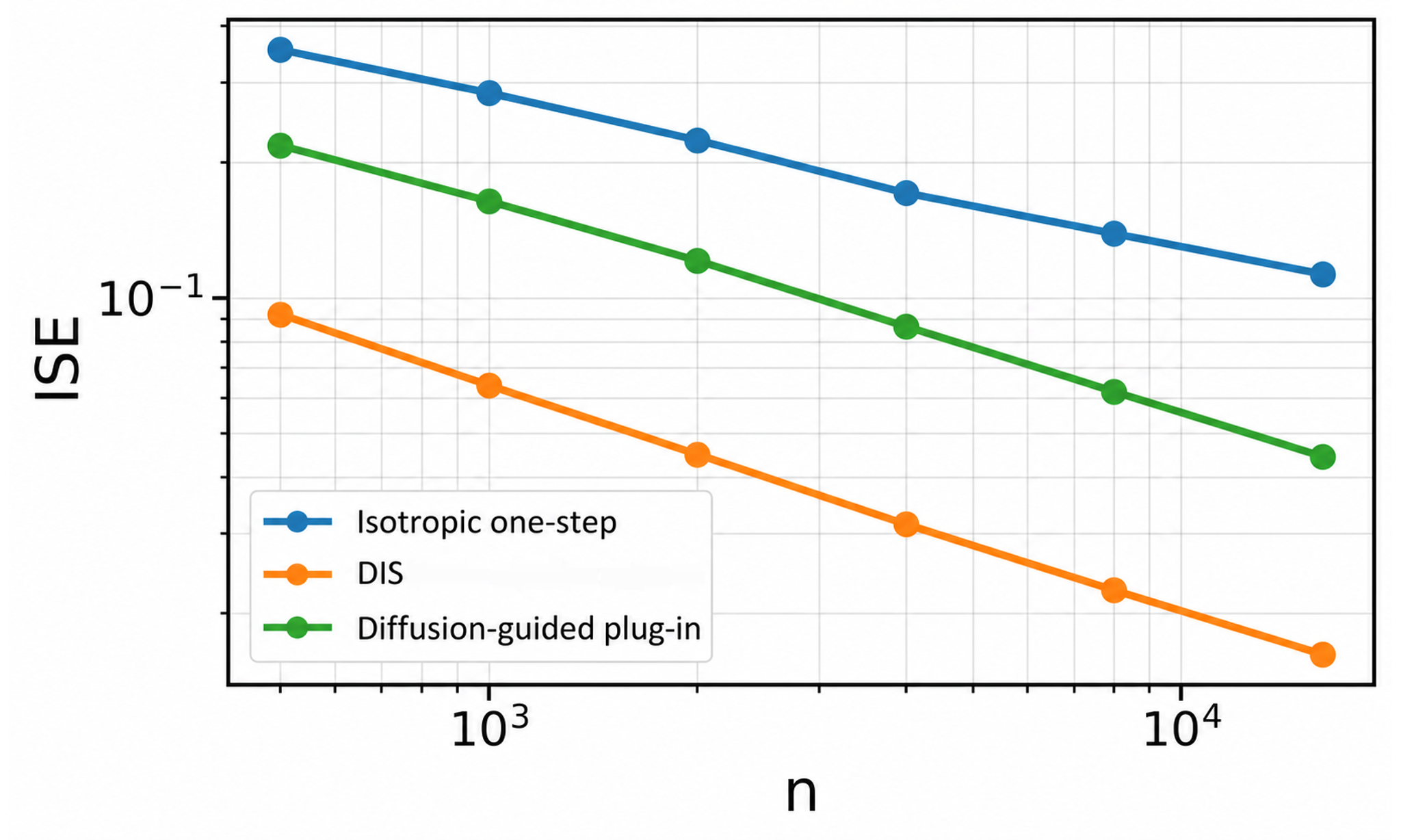}
    \caption{Density ISE versus sample size against matched targets.}
    \label{fig:exp-density}
  \end{minipage}\hfill
  \begin{minipage}[t]{0.49\linewidth}
    \centering
    \includegraphics[width=\linewidth]{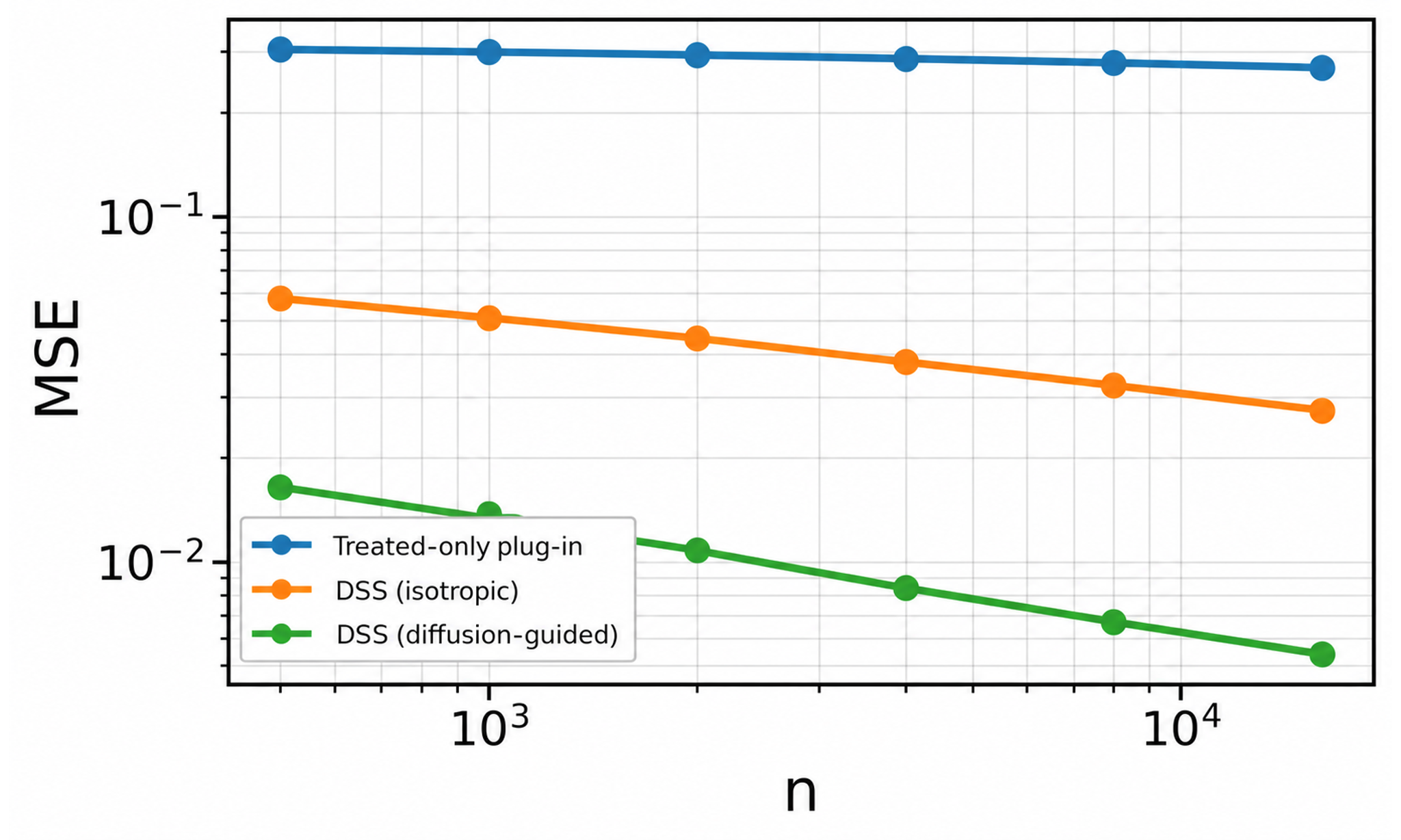}
    \caption{Stein-functional MSE versus sample size against matched targets.}
    \label{fig:exp-stein}
  \end{minipage}
\end{figure}

\section{Conclusion}
\label{sec:conclusion}

We introduce DIS and DSS, a unified framework combining diffusion-guided
geometry-adaptive smoothing with one-step estimation and simultaneous
inference for geometry-indexed targets. By
separating empirical uncertainty from geometry-learning drift, the framework
also provides drift-adjusted coverage for population-geometry targets and
links score and probability-flow errors to the transported smoothing operator. Important avenues
remain for future work. Although the theory treats geometry learning modularly
and provides explicit stability conditions, sharper primitive rates for
practical score-network training under confounding are still needed.
Computation may also be demanding for ultrahigh-dimensional outcomes, where
diffusion-geometry learning and transported-kernel evaluation are costly.
These challenges motivate end-to-end diffusion implementations, principled
data-adaptive selection of \(h\), and tuning procedures that preserve valid
inference under adaptive smoothing.

\vspace*{-.15in}
\section*{Disclosure Statement}
\vspace*{-.05in}
The authors declare no conflicts of interest.
\vspace*{-.15in}
\section*{Data Availability Statement}
\vspace*{-.05in}
The simulation code is available at \url{https://github.com/kwangho-joshua-kim/cf-diffusion}.

\bibliographystyle{agsm}
\bibliography{bibliography}

\newpage
\bigskip

\appendix
\newpage
\vspace*{.01in}
\section*{\Large \centerline{SUPPLEMENTARY MATERIALS}}
\begin{center}
   {\large Geometry Adaptive Counterfactual Distribution Learning\\ with Diffusion-Guided Smoothing}\\
   by Kwangho Kim
\end{center}
\vspace*{.1in}
\setcounter{equation}{0}
\renewcommand{\theequation}{A.\arabic{equation}}
\setcounter{figure}{0}
\renewcommand{\thefigure}{A.\arabic{figure}}
\setcounter{theorem}{0}
\renewcommand{\thetheorem}{A.\arabic{theorem}}

\section{Simulation Details}
\label{appsec:sim-details}

The experiment examines whether diffusion-guided smoothing benefits from
low-dimensional structure in a high-dimensional outcome representation and
whether one-step correction reduces error relative to plug-in estimation with
the same transported kernel.

\subsection{Data source and representation}

We follow the CelebA observational design used in the counterfactual faces
experiment of \citet{luedtke2025doublegen}. Each unit is an image with binary
treatment \(A\in\{0,1\}\), where \(A=1\) denotes smiling, and baseline
confounders \(X\) given by CelebA attributes. A fixed pretrained image encoder
maps each image to an outcome embedding \(Y\).

To vary ambient dimension, we apply fixed maps to obtain
\(\widetilde Y_r\in\mathbb R^r\), with
\(r\in\{500,2500,10000\}\). These maps are generated once and held fixed
across sample sizes, methods, and replications. Let
\[
Y_r^{\mathrm{proj}}
=
P_r\widetilde Y_r
\in\mathbb R^{r_{\mathrm{proj}}},
\quad
r_{\mathrm{proj}}\in\{2,3\},
\]
where \(P_r\) is also fixed across replications. Diffusion geometry is learned
from \(\widetilde Y_r\), whereas density and Stein-functional errors are
evaluated on a compact region
\(\mathcal Y_0\subset\mathbb R^{r_{\mathrm{proj}}}\).

\subsection{Cross-fitted diffusion geometry}

Let \(\mathcal I_1,\ldots,\mathcal I_K\) denote the evaluation folds. On the
complementary sample for fold \(k\), we first estimate
\(\widehat\pi_a^{(-k)}(x)\) and define normalized inverse-propensity weights
\[
\overline w_{i,a}^{(-k)}
=
\frac{
\mathbbm 1\{A_i=a\}/\widehat\pi_a^{(-k)}(X_i)
}{
\sum_{j\notin\mathcal I_k}
\mathbbm 1\{A_j=a\}/\widehat\pi_a^{(-k)}(X_j)
}.
\]
The intervention-specific diffusion score is learned by weighted denoising
score matching:
\[
\widehat\theta^{(-k)}
\in
\argmin_{\theta}
\sum_{i\notin\mathcal I_k}
\overline w_{i,a}^{(-k)}
\E_{t,Z_t\sim q_t(\cdot\mid\widetilde Y_{r,i})}
\left[
\left\|
s_{\theta}^{\rm diff}(Z_t,t)
-
\nabla_z\log q_t(Z_t\mid\widetilde Y_{r,i})
\right\|_2^2
\right].
\]
We then solve the probability-flow ODE
\[
\frac{d}{dt}z_t
=
b(z_t,t)
-
\frac{1}{2}\sigma(t)^2
s_{\widehat\theta^{(-k)}}^{\rm diff}(z_t,t)
\]
and construct the transported kernel and its gradient using the inverse flow
and Jacobian factor in \eqref{eqn:diff-informed-bump}. The learned flow induces
the projected smoothing kernel
\[
\kappa_{h,r,\widehat\theta^{(-k)}}^{\mathrm{proj}}
(y;\widetilde u),
\quad
\dot\kappa_{h,r,\widehat\theta^{(-k)}}^{\mathrm{proj}}
(y;\widetilde u)
=
\nabla_y
\kappa_{h,r,\widehat\theta^{(-k)}}^{\mathrm{proj}}
(y;\widetilde u).
\]
Here
\(\kappa_{h,r,\widehat\theta^{(-k)}}^{\mathrm{proj}}\) is the Monte Carlo
pushforward density of
\(P_r\Phi_{\varepsilon_h,\widehat\theta^{(-k)}}(Z_{\varepsilon_h})\), where
\(Z_{\varepsilon_h}\sim
q_{\varepsilon_h}(\cdot\mid\widetilde u)\), and its gradient is obtained by
differentiating the same representation. These quantities are used to evaluate
the DIS and DSS estimating functions on fold \(k\). Thus the
geometry-adaptive estimators use the learned diffusion score and probability
flow directly rather than a local-PCA proxy.

The encoder, representation maps, number of folds, score-network architecture,
diffusion schedule, optimizer, ODE solver and tolerances, numerical integration
scheme, bandwidth constant, replication count, and random seeds are fixed
across methods and reported in the accompanying reproducibility code.

\subsection{Held-out reference targets}

Because the true counterfactual law is unavailable, we split the observations
into a large reference pool \(\mathcal D_{\mathrm{ref}}\) and a working pool
from which samples \(\mathcal D_n\) are repeatedly drawn. The reference pool
is not used to fit the working-sample estimators.

On \(\mathcal D_{\mathrm{ref}}\), we estimate
\(\widehat\pi_a^{\mathrm{ref}}\) and independently train a diffusion score and
probability flow using the same weighted score-matching procedure. Let
\(\kappa_{h,r}^{\mathrm{ref}}\) and
\(\dot\kappa_{h,r}^{\mathrm{ref}}\) denote the resulting projected kernel and
gradient. The diffusion-guided density reference is
\[
\widehat p_{a,h,r}^{\mathrm{ref}}(y)
=
\frac{
\sum_{i\in\mathcal D_{\mathrm{ref}}}
\mathbbm 1\{A_i=a\}w_i^{\mathrm{ref}}
\kappa_{h,r}^{\mathrm{ref}}
(y;\widetilde Y_{r,i})
}{
\sum_{i\in\mathcal D_{\mathrm{ref}}}
\mathbbm 1\{A_i=a\}w_i^{\mathrm{ref}}
},
\quad
w_i^{\mathrm{ref}}
=
\{\widehat\pi_a^{\mathrm{ref}}(X_i)\}^{-1}.
\]
The corresponding score reference is
\[
\widehat s_{a,h,r}^{\mathrm{ref}}(y)
=
\frac{
\sum_{i\in\mathcal D_{\mathrm{ref}}}
\mathbbm 1\{A_i=a\}w_i^{\mathrm{ref}}
\dot\kappa_{h,r}^{\mathrm{ref}}
(y;\widetilde Y_{r,i})
}{
\sum_{i\in\mathcal D_{\mathrm{ref}}}
\mathbbm 1\{A_i=a\}w_i^{\mathrm{ref}}
\kappa_{h,r}^{\mathrm{ref}}
(y;\widetilde Y_{r,i})
}.
\]
The isotropic estimator is evaluated against the analogous reference formed
using its isotropic Gaussian kernel. DIS and the diffusion-guided plug-in
estimator share the same diffusion-guided reference target. Hence their
comparison isolates one-step correction, whereas isotropic-versus-diffusion
comparisons concern error decay relative to matched smoothing targets.

For Stein-functional evaluation, define the fixed held-out measure
\[
\widehat\nu_a^{\mathrm{ref}}
=
\frac{
\sum_{i\in\mathcal D_{\mathrm{ref}}}
\mathbbm 1\{A_i=a\}w_i^{\mathrm{ref}}
\delta_{Y_{r,i}^{\mathrm{proj}}}
}{
\sum_{i\in\mathcal D_{\mathrm{ref}}}
\mathbbm 1\{A_i=a\}w_i^{\mathrm{ref}}
}.
\]
Conditional on \(\mathcal D_{\mathrm{ref}}\), this measure and all reference
targets are treated as fixed across working-sample replications.

\subsection{Estimators}

\textbf{Density estimators.}
We compare three estimators. The isotropic one-step estimator uses the
augmented inverse-propensity-weighted map with an isotropic Gaussian kernel.
The diffusion-guided plug-in estimator uses
\(\kappa_{h,r,\widehat\theta}^{\mathrm{proj}}\) but omits the residual
correction. DIS uses the same transported kernel and adds the one-step
correction in \eqref{eqn:estimator-one-step-diff-smoothing}. Thus comparing
the diffusion-guided plug-in estimator with DIS isolates semiparametric
correction without changing the geometry or target.

\textbf{Stein-functional estimators.}
We compare a treated-only plug-in estimator with isotropic and
diffusion-guided DSS. Both DSS estimators use the ratio construction in
\eqref{eqn:dss-estimator}; they differ only in whether the kernel is isotropic
or induced by the learned probability flow. The treated-only estimator uses
only observations with \(A=a\), omits causal adjustment, and serves as a
biased baseline under confounding.

Accordingly, the two experiments use different auxiliary baselines: the
density plug-in isolates one-step correction, whereas the treated-only Stein
estimator illustrates the effect of omitting causal adjustment.

\subsection{Evaluation criteria and test class}

Density error is computed by numerical quadrature:
\[
\operatorname{ISE}
=
\int_{\mathcal Y_0}
\left\{
\widehat p_{a,h,r}(y)
-
\widehat p_{a,h,r}^{\mathrm{ref}}(y)
\right\}^2dy.
\]
For Stein evaluation, let
\[
g_v(y)
=
v\exp(-\|y\|_2^2/2),
\quad
v\in\mathcal V,
\]
where \(\mathcal V\) contains the coordinate vectors and a fixed collection of
normalized random directions generated once and reused throughout. For
\(\mathcal G=\{g_v:v\in\mathcal V\}\), the reference functional is
\[
\widehat\Psi_{a,h}^{\mathrm{ref}}(g)
=
\int_{\mathcal Y_0}
\left\{
\nabla_y\!\cdot g(y)
+
g(y)^\top\widehat s_{a,h,r}^{\mathrm{ref}}(y)
\right\}
d\widehat\nu_a^{\mathrm{ref}}(y).
\]
We report mean squared error averaged over \(\mathcal G\) and the
working-sample replications.

We use the common deterministic bandwidth schedule
\(h_n=c_hn^{-1/5}\), with \(c_h\) selected once from the independent reference
pool and then fixed. The experiment uses eight approximately log-spaced sample
sizes from \(500\) to \(50000\). Log error is plotted against log sample size,
and fitted slopes are interpreted as descriptive evidence of error decay
rather than formal estimates of asymptotic exponents.

\subsection{Results and extended diagnostics}

Figures~\ref{fig:exp-density} and~\ref{fig:exp-stein} present the main results.
The diffusion-guided estimators exhibit generally faster error decay relative
to their matched targets than their isotropic counterparts. In the
density experiment, DIS also improves on the diffusion-guided plug-in
estimator, although both use the same score-driven probability-flow kernel.
These patterns persist over the extended sample-size range and across
\(r\in\{500,2500,10000\}\).

\begin{figure}[t!]
  \centering
  \includegraphics[width=0.7\linewidth]{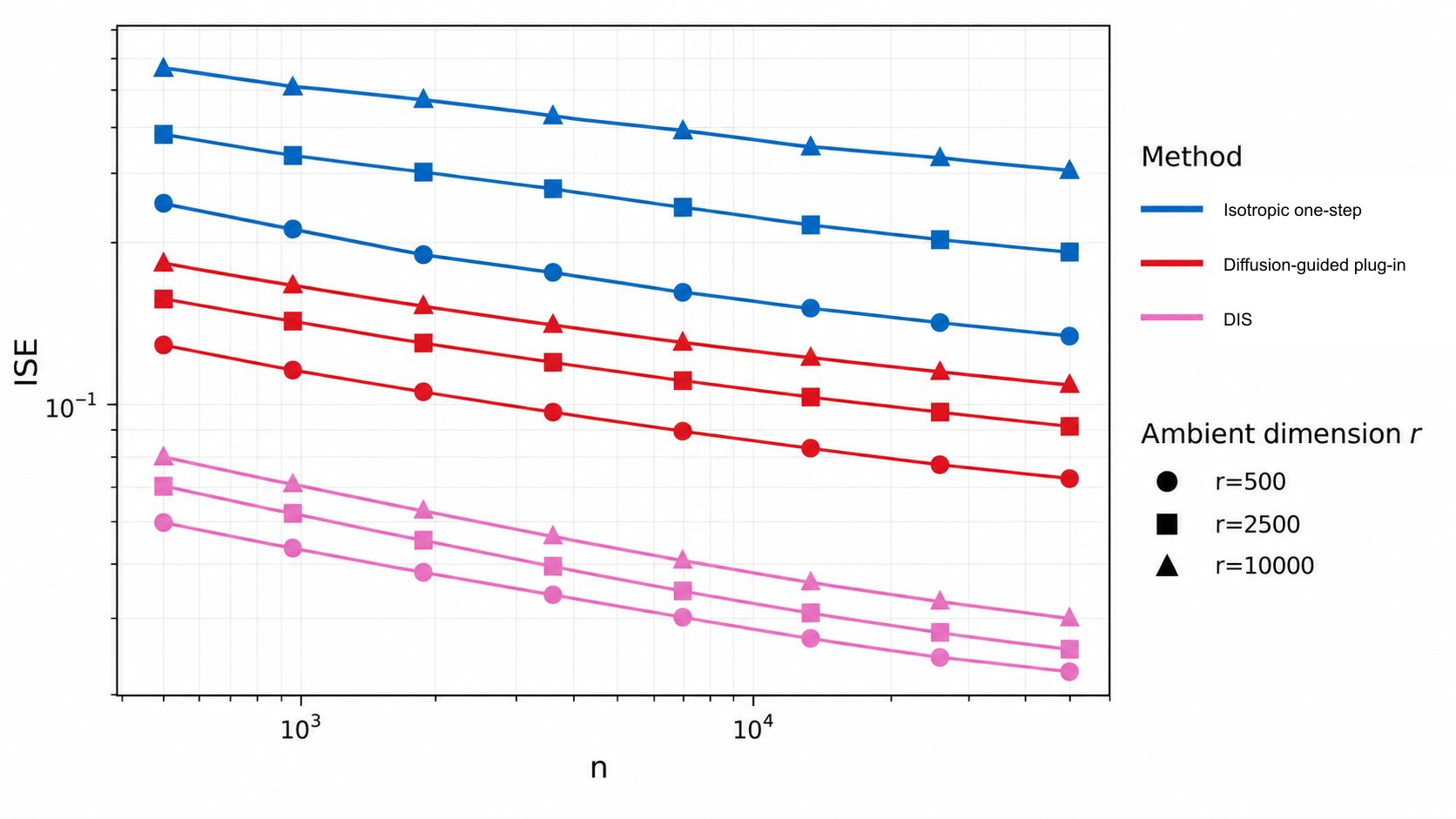}
  \caption{Density ISE relative to matched held-out targets over eight sample
  sizes and \(r\in\{500,2500,10000\}\).}
  \label{fig:exp-density-ext}
\end{figure}

\begin{figure}[t!]
  \centering
  \includegraphics[width=0.7\linewidth]{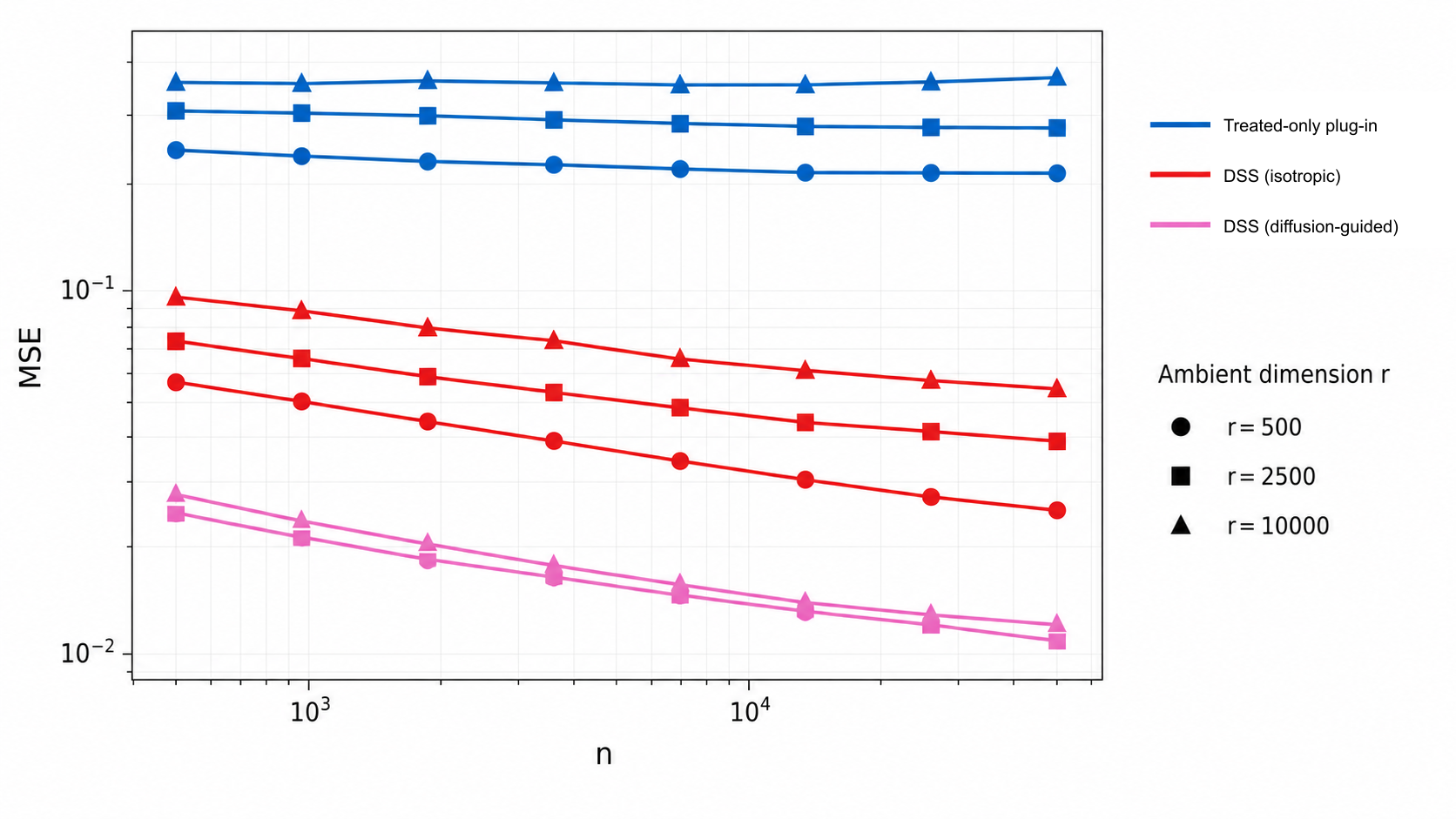}
  \caption{Stein-functional MSE relative to matched held-out targets over eight
  sample sizes and \(r\in\{500,2500,10000\}\).}
  \label{fig:exp-stein-ext}
\end{figure}

\subsection{Counterfactual triptych samples}

To complement the quantitative results, we include illustrative
counterfactual triptychs for two subjects, with columns corresponding to
\(A=0\), the factual image, and \(A=1\). These triptychs illustrate the
semi-synthetic intervention and are not outputs of DIS or DSS or inputs to the
quantitative evaluation.

\begin{figure}[t]
  \centering
  \includegraphics[width=0.85\linewidth]{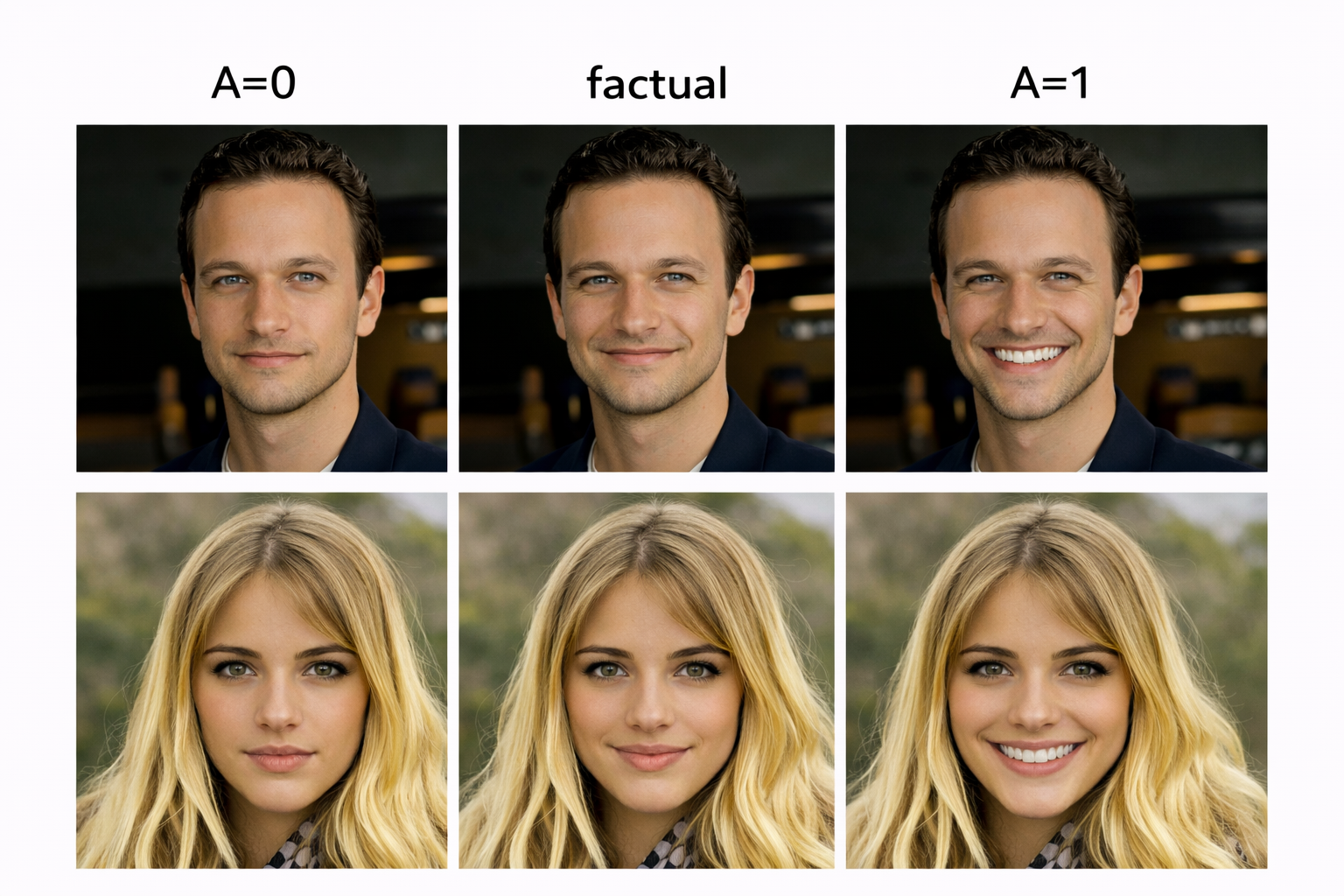}
  \caption{Illustrative counterfactual triptychs. The columns show \(A=0\),
  factual, and \(A=1\); the images illustrate intervention semantics and are
  not outputs of DIS or DSS.}
  \label{fig:cf-triptychs}
\end{figure}

\section{Fixed-Geometry Pointwise DIS Expansion}
\label{appsec:pointwise-dis}

Set \(\widehat\theta=\theta_0\), and write
\[
\delta_\pi
=
\|\widehat\pi_a-\pi_a\|_{L_2(\Pb)},
\quad
\delta_\mu^0(y)
=
\|\widehat\mu_{a,h,\theta_0}(\cdot;y)
-
\mu_{a,h,\theta_0}(\cdot;y)\|_{L_2(\Pb)}.
\]
Under cross fitting, all estimated quantities are interpreted foldwise under
the convention of Section~\ref{subsec:asymptotics-risk}, with fold indices
suppressed. Let
\[
\widehat\eta^0
=
\{\widehat\pi_a,\widehat\mu_{a,h,\theta_0},
\kappa_{h,\theta_0}\},
\]
and define
\[
\widehat p^{\mathrm{geo},0}_{a,h}(y)
=
\Pn\{\varphi^{\mathrm{geo}}_h(Z;y,\widehat\eta^0)\}.
\]

\begin{lemma}[Fixed-geometry pointwise expansion]
\label{lem:fixed-geometry-expansion}
Fix \(h>0\) and \(y\in\mathbb R^d\). Suppose
Assumptions~\ref{assumption:boundedness} and
\ref{assumption:sample-splitting} hold, and
\[
\delta_\pi=o_\Pb(1),
\quad
\delta_\mu^0(y)=o_\Pb(1),
\quad
\Pb\{\inf_x\widehat\pi_a(x)\ge\pi_{\min}/2\}\to1.
\]
Then
\[
\widehat p^{\mathrm{geo},0}_{a,h}(y)
-
p^{\mathrm{geo}}_{a,h}(y)
=
(\Pn-\Pb)\{\varphi^{\mathrm{geo}}_h(Z;y,\eta_0)\}
+
R_{\pi,\mu}^0(y)
+
o_\Pb(n^{-1/2}),
\]
where
\[
|R_{\pi,\mu}^0(y)|
\lesssim
\delta_\pi\delta_\mu^0(y).
\]
\end{lemma}

\section{Ambient Density Extension for DIS}
\label{appsec:ambient-density}

The main analysis treats \(p^{\mathrm{geo}}_{a,h}\) as the primary target
because it remains well defined when \(\Pb_a\) is singular in the ambient
space. This section records the extension to the case where an ordinary ambient
density exists.

\begin{assumption}[Ambient density approximation]
\label{assump:ambient-density-approx}
Suppose that \(\Pb_a\) is absolutely continuous with respect to
\(d\)-dimensional Lebesgue measure on an open set containing
\(\mathcal Y_0\), with density \(p_a\). For some \(\beta>0\) and
\(C_A<\infty\), assume
\[
A_h^2(\mathcal Y_0)
\equiv
\int_{\mathcal Y_0}
\{p^{\mathrm{geo}}_{a,h}(y)-p_a(y)\}^2dy
\le
C_Ah^{2\beta}.
\]
\end{assumption}

Assumption~\ref{assump:ambient-density-approx} is an \(L_2\) smoothing-bias
condition. For ordinary kernels, analogous bounds follow from standard
smoothness and kernel-order conditions; here the condition is imposed on the
diffusion-guided smoothing rule.

\begin{corollary}
\label{cor:risk-ambient-density}
Suppose the conditions of Theorem~\ref{thm:dis-risk-global} and
Assumption~\ref{assump:ambient-density-approx} hold. Then the fixed-geometry
estimator satisfies
\[
\int_{\mathcal Y_0}
\{\widehat p^{\mathrm{geo},0}_{a,h}(y)-p_a(y)\}^2dy
=
O_\Pb\!\left[
\frac{H_h(\mathcal Y_0)}{n}
+
\operatorname*{ess\,sup}_{y\in\mathcal Y_0}
\{R_{\pi,\mu}^0(y)\}^2
+
h^{2\beta}
\right]
+
o_\Pb(n^{-1}).
\]
The learned-geometry estimator satisfies
\[
\begin{aligned}
\int_{\mathcal Y_0}
\{\widehat p^{\mathrm{geo}}_{a,h}(y)-p_a(y)\}^2dy
&=
O_\Pb\!\left[
\frac{H_h(\mathcal Y_0)}{n}
+
\operatorname*{ess\,sup}_{y\in\mathcal Y_0}
\{R_{\pi,\mu}^{\widehat\theta}(y)\}^2 \right. \\
&\quad\left.
+
\int_{\mathcal Y_0}B_{\theta,h}(y)^2dy
+
h^{2\beta}
\right]
+
o_\Pb(n^{-1}).
\end{aligned}
\]
If Assumption~\ref{assump:integrated-geometry-drift} also holds, then
\[
\begin{aligned}
\int_{\mathcal Y_0}
\{\widehat p^{\mathrm{geo}}_{a,h}(y)-p_a(y)\}^2dy
&=
O_\Pb\!\left[
\frac{H_h(\mathcal Y_0)}{n}
+
\operatorname*{ess\,sup}_{y\in\mathcal Y_0}
\{R_{\pi,\mu}^{\widehat\theta}(y)\}^2 \right. \\
&\quad\left.
+
h^{-\Lambda_{\rm geom}(\mathcal Y_0)}
\rho_{\theta,n_{\rm geom}}^2
+
h^{2\beta}
\right]
+
o_\Pb(n^{-1}).
\end{aligned}
\]
Consequently, any bound
\(H_h(\mathcal Y_0)\lesssim h^{-D_H}\) converts the oracle stochastic term to
\(n^{-1}h^{-D_H}\).
\end{corollary}

\begin{proof}
For the fixed-geometry estimator, the inequality
\((u+v)^2\le2u^2+2v^2\) gives
\[
\begin{aligned}
\int_{\mathcal Y_0}
\{\widehat p^{\mathrm{geo},0}_{a,h}(y)-p_a(y)\}^2dy
&\le
2\int_{\mathcal Y_0}
\{\widehat p^{\mathrm{geo},0}_{a,h}(y)
-
p^{\mathrm{geo}}_{a,h}(y)\}^2dy \\
&\quad+
2\int_{\mathcal Y_0}
\{p^{\mathrm{geo}}_{a,h}(y)-p_a(y)\}^2dy .
\end{aligned}
\]
The first term is controlled by Theorem~\ref{thm:dis-risk-global}, while the
second is \(A_h^2(\mathcal Y_0)\le C_Ah^{2\beta}\) by
Assumption~\ref{assump:ambient-density-approx}. This proves the fixed-geometry
bound.

The same argument gives
\[
\begin{aligned}
\int_{\mathcal Y_0}
\{\widehat p^{\mathrm{geo}}_{a,h}(y)-p_a(y)\}^2dy
&\le
2\int_{\mathcal Y_0}
\{\widehat p^{\mathrm{geo}}_{a,h}(y)
-
p^{\mathrm{geo}}_{a,h}(y)\}^2dy \\
&\quad+
2A_h^2(\mathcal Y_0)
\end{aligned}
\]
for the learned-geometry estimator. Applying
Theorem~\ref{thm:dis-risk-global} to the first term yields the bound containing
\(\int_{\mathcal Y_0}B_{\theta,h}(y)^2dy\). Finally,
Assumption~\ref{assump:integrated-geometry-drift} and Jensen's inequality imply
\[
\int_{\mathcal Y_0}B_{\theta,h}(y)^2dy
=
O_\Pb\!\left(
h^{-\Lambda_{\rm geom}(\mathcal Y_0)}
\rho_{\theta,n_{\rm geom}}^2
\right),
\]
which gives the final result.
\end{proof}

If \(H_h(\mathcal Y_0)\lesssim h^{-D_H}\) for all sufficiently small \(h\),
and, along a sequence \(h=h_n\to0\),
\[
\operatorname*{ess\,sup}_{y\in\mathcal Y_0}
\{R_{\pi,\mu}^{\widehat\theta}(y)\}^2
+
h_n^{-\Lambda_{\rm geom}(\mathcal Y_0)}
\rho_{\theta,n_{\rm geom}}^2
=
o_\Pb\!\left(
\frac{h_n^{-D_H}}{n}
+
h_n^{2\beta}
\right),
\]
then
\[
\int_{\mathcal Y_0}
\{\widehat p^{\mathrm{geo}}_{a,h_n}(y)-p_a(y)\}^2dy
=
O_\Pb\!\left(
\frac{h_n^{-D_H}}{n}
+
h_n^{2\beta}
\right).
\]
Balancing the stochastic and smoothing-bias terms gives
\[
h_n\asymp n^{-1/(2\beta+D_H)},
\quad
\int_{\mathcal Y_0}
\{\widehat p^{\mathrm{geo}}_{a,h_n}(y)-p_a(y)\}^2dy
=
O_\Pb\!\left(
n^{-2\beta/(2\beta+D_H)}
\right).
\]
The fixed-geometry analogue follows by omitting the geometry-drift term and
replacing \(R_{\pi,\mu}^{\widehat\theta}\) by \(R_{\pi,\mu}^0\).

\section{Ambient Score Extension for DSS}
\label{appsec:dss-ambient-score}

The main DSS theory targets the fixed-\(h\) smoothed score
\(s^{\mathrm{geo}}_{a,h}\), which is well defined even when the counterfactual
law \(\Pb_a\) does not admit an ordinary ambient density. This appendix records
the standard extension to the case where an ambient density and score exist.

\begin{assumption}[Ambient score approximation]
\label{assump:ambient-score-approx}
Suppose that \(\Pb_a\) has a positive continuously differentiable density
\(p_a\) on an open set containing \(\mathcal Y_0\), with score
\(s_a=\nabla\log p_a\). Assume that, for some \(\beta_s>0\) and
\(C'_A<\infty\),
\[
A_{s,h}^2(\mathcal Y_0)
\equiv
\int_{\mathcal Y_0}
\|s^{\mathrm{geo}}_{a,h}(y)-s_a(y)\|_2^2dy
\le
C'_A h^{2\beta_s}.
\]
\end{assumption}

Assumption~\ref{assump:ambient-score-approx} is an \(L_2\) score-smoothing bias
condition. It is stronger than the ambient-density approximation because it
requires approximation of \(\nabla\log p_a\), not only \(p_a\). For ordinary
kernels, such a condition typically requires additional smoothness and a lower
bound on \(p_a\) over \(\mathcal Y_0\); here we impose the analogous condition
directly on the diffusion-guided smoothed score.

\begin{corollary}
\label{cor:dss-ambient-score}
Suppose the conditions of
Corollary~\ref{cor:dss-risk-geometry-stability} and
Assumption~\ref{assump:ambient-score-approx} hold. Then
\[
\begin{aligned}
&\int_{\mathcal Y_0}
\|\widehat s^{\mathrm{geo}}_{a,h}(y)-s_a(y)\|_2^2dy \\
&\quad=
O_\Pb\!\left[
\frac{H_h^s(\mathcal Y_0)}{n}
+
\operatorname*{ess\,sup}_{y\in\mathcal Y_0}
\|R_{s,\pi,\mu,\nu}^{\widehat\theta}(y)\|_2^2
+
h^{-\Lambda_{s,\rm geom}(\mathcal Y_0)}
\rho_{\theta,n_{\rm geom}}^2
+
h^{2\beta_s}
\right]
+
o_\Pb(n^{-1}).
\end{aligned}
\]
\end{corollary}

\begin{proof}
By \(\|u+v\|_2^2\le2\|u\|_2^2+2\|v\|_2^2\),
\[
\begin{aligned}
\int_{\mathcal Y_0}
\|\widehat s^{\mathrm{geo}}_{a,h}(y)-s_a(y)\|_2^2dy
&\le
2\int_{\mathcal Y_0}
\|\widehat s^{\mathrm{geo}}_{a,h}(y)
-s^{\mathrm{geo}}_{a,h}(y)\|_2^2dy \\
&\quad+
2A_{s,h}^2(\mathcal Y_0).
\end{aligned}
\]
The first term is controlled by
Corollary~\ref{cor:dss-risk-geometry-stability}, while the second is bounded by
\(C'_Ah^{2\beta_s}\) under
Assumption~\ref{assump:ambient-score-approx}. The result follows.
\end{proof}

The preceding corollary separates four contributions: oracle score stochastic
error, causal nuisance error, learned-geometry drift, and score-smoothing bias.
The next result records the optimized rate when the score concentration
functional has intrinsic scaling.

\begin{theorem}
\label{thm:dss-structural-rate}
Suppose the conditions of Corollary~\ref{cor:dss-ambient-score} hold. Assume
that, for some \(D_{\mathrm{score}}\ge0\) and all sufficiently small \(h\),
\[
H_h^s(\mathcal Y_0)\lesssim h^{-D_{\mathrm{score}}},
\]
and that
\[
\operatorname*{ess\,sup}_{y\in\mathcal Y_0}
\|R_{s,\pi,\mu,\nu}^{\widehat\theta}(y)\|_2^2
+
h^{-\Lambda_{s,\rm geom}(\mathcal Y_0)}
\rho_{\theta,n_{\rm geom}}^2
=
o_\Pb\!\left(
\frac{h^{-D_{\mathrm{score}}}}{n}
+
h^{2\beta_s}
\right).
\]
Then
\[
\int_{\mathcal Y_0}
\|\widehat s^{\mathrm{geo}}_{a,h}(y)-s_a(y)\|_2^2dy
=
O_\Pb\!\left(
\frac{h^{-D_{\mathrm{score}}}}{n}
+
h^{2\beta_s}
\right)
+
o_\Pb(n^{-1}).
\]
If these conditions hold along a sequence \(h=h_n\) with
\[
h_n
\asymp
n^{-1/(2\beta_s+D_{\mathrm{score}})},
\]
then
\[
\int_{\mathcal Y_0}
\|\widehat s^{\mathrm{geo}}_{a,h_n}(y)-s_a(y)\|_2^2dy
=
O_\Pb\!\left(
n^{-2\beta_s/(2\beta_s+D_{\mathrm{score}})}
\right).
\]
\end{theorem}

\begin{proof}
By Corollary~\ref{cor:dss-ambient-score} and the assumed concentration bound,
\[
\begin{aligned}
&\int_{\mathcal Y_0}
\|\widehat s^{\mathrm{geo}}_{a,h}(y)-s_a(y)\|_2^2dy \\
&\quad=
O_\Pb\!\left[
\frac{h^{-D_{\mathrm{score}}}}{n}
+
\operatorname*{ess\,sup}_{y\in\mathcal Y_0}
\|R_{s,\pi,\mu,\nu}^{\widehat\theta}(y)\|_2^2
+
h^{-\Lambda_{s,\rm geom}(\mathcal Y_0)}
\rho_{\theta,n_{\rm geom}}^2
+
h^{2\beta_s}
\right]
+
o_\Pb(n^{-1}).
\end{aligned}
\]
The assumed negligibility of the causal nuisance and geometry-drift terms gives
the first claim. Balancing \(n^{-1}h^{-D_{\mathrm{score}}}\) and
\(h^{2\beta_s}\) gives
\(h_n\asymp n^{-1/(2\beta_s+D_{\mathrm{score}})}\), under which both terms are
of order \(n^{-2\beta_s/(2\beta_s+D_{\mathrm{score}})}\). Since
\(2\beta_s/(2\beta_s+D_{\mathrm{score}})\le1\), the \(o_\Pb(n^{-1})\) term is
negligible relative to this rate. This proves the second claim.
\end{proof}

Under Theorem~\ref{thm:dss-covariance-concentration}, one may take
\(D_{\mathrm{score}}=d_\star+2\).

Theorem~\ref{thm:dss-structural-rate} clarifies when DSS can achieve favorable
rates for ambient score targets. The condition
\(H_h^s(\mathcal Y_0)\lesssim h^{-D_{\mathrm{score}}}\) means that the combined
\(L_2\) concentration of the transported kernel and its \(y\)-gradient scales
with score-concentration exponent \(D_{\mathrm{score}}\), rather than with the
ambient isotropic exponent \(d+2\). Such a reduction can occur when the
transported kernel and its gradient inherit low-dimensional structure.

The causal nuisance and geometry-drift terms must also be negligible relative
to the oracle score rate. Along
\(h_n\asymp n^{-1/(2\beta_s+D_{\mathrm{score}})}\), this requires
\[
\operatorname*{ess\,sup}_{y\in\mathcal Y_0}
\|R_{s,\pi,\mu,\nu}^{\widehat\theta}(y)\|_2^2
+
h_n^{-\Lambda_{s,\rm geom}(\mathcal Y_0)}
\rho_{\theta,n_{\rm geom}}^2
=
o_\Pb\!\left\{
n^{-2\beta_s/(2\beta_s+D_{\mathrm{score}})}
\right\}.
\]
Thus the benefit of DSS is structural, not automatic: the reduction from the
ambient score-concentration exponent must offset both the derivative cost of
score estimation and the additional geometry-learning drift.

Comparing their respective optimized \(L_2\) convergence exponents, the DSS
exponent exceeds the DIS exponent when
\[
\frac{\beta_s}{2\beta_s+D_{\mathrm{score}}}
>
\frac{\beta}{2\beta+D_H},
\]
or equivalently,
\[
D_{\mathrm{score}}
<
\frac{\beta_s}{\beta}D_H.
\]
This favors DSS when intrinsic score structure offsets its additional
derivative cost, particularly for score-based sampling and Stein functionals.

\section{Geometry Drift and Flow Stability}
\label{appsec:geometry-drift}

Assumptions~\ref{assump:integrated-geometry-drift}
and~\ref{assump:integrated-score-geometry-drift} use
\(\rho_{\theta,n_{\rm geom}}\) to summarize the error from learning the
diffusion geometry. This error is distinct from estimation of the
counterfactual density or score itself. The learned diffusion score is used to
orient and transport the smoothing rule, so its effect enters through the
induced perturbation of the kernel for DIS and of the kernel and its gradient
for DSS.

Let \(\mathcal T_h\) denote the diffusion time interval relevant for smoothing
at scale \(h\), for example \(\mathcal T_h=[0,\varepsilon_h]\), and let
\(\mathcal R_h\) be a tube containing the corresponding probability-flow
trajectories. To distinguish the diffusion score from the counterfactual score,
write \(s^{\rm diff}_\theta(z,t)\) for the former. For \(r\in\{1,2\}\), define
\[
\|s^{\rm diff}_{\widehat\theta}
-
s^{\rm diff}_{\theta_0}\|_{\mathcal S_h^{(r)}}
=
\int_{\mathcal T_h}\sigma(t)^2
\sum_{j=0}^r
\sup_{z\in\mathcal R_h}
\|\nabla^j s^{\rm diff}_{\widehat\theta}(z,t)
-
\nabla^j s^{\rm diff}_{\theta_0}(z,t)\|dt,
\]
where \(\nabla^0s=s\), with compatible Euclidean, operator, or tensor norms.
For DIS, control in \(\mathcal S_h^{(1)}\) is sufficient, whereas DSS requires
the stronger norm \(\mathcal S_h^{(2)}\) because the transported kernel
gradient depends on spatial derivatives of the flow and Jacobian. We write
\(\rho_{\theta,n_{\rm geom}}\) for the corresponding rate when the relevant
norm is clear from context. Alternatively, the required flow perturbation
bounds may be verified directly.

\subsection{From Flow Perturbations to Kernel Drift}

We first give a sufficient route from flow perturbations to the integrated DIS
geometry-drift condition. For notational simplicity, let
\[
\Psi_0=\Phi_{\varepsilon_h,\theta_0}^{-1},
\quad
\Psi_1=\Phi_{\varepsilon_h,\widehat\theta}^{-1},
\quad
J_0=J_{\varepsilon_h,\theta_0},
\quad
J_1=J_{\varepsilon_h,\widehat\theta}.
\]
Define the uniform inverse-flow and relative-Jacobian errors over
\(\mathcal Y_0\) by
\[
\Delta_\Phi
=
\sup_{y\in\mathcal Y_0}\|\Psi_1(y)-\Psi_0(y)\|_2,
\quad
\Delta_J
=
\sup_{y\in\mathcal Y_0}
\left|
\frac{J_1(y)}{J_0(y)}-1
\right|.
\]
For \(t\in[0,1]\), let
\[
\Psi_t=\Psi_0+t(\Psi_1-\Psi_0),
\quad
J_t
=
J_0\left\{
1+t\left(\frac{J_1}{J_0}-1\right)
\right\},
\]
and define the interpolated smoothing density
\[
\kappa_{h,t}(y;u)
=
q_{\varepsilon_h}\{\Psi_t(y)\mid u\}J_t(y).
\]
Thus \(\kappa_{h,0}=\kappa_{h,\theta_0}\) and
\(\kappa_{h,1}=\kappa_{h,\widehat\theta}\).

The next two lemmas are stated for a single generic fitted geometry. Under
\(K\)-fold cross fitting, they are applied separately to
\(\widehat\theta^{(-k)}\), conditional on \(\mathcal F_{-k}\), uniformly over
the fixed folds. Accordingly, \(\mathcal F_{\rm geom}\) below denotes the
geometry-training sigma-field for this generic fit and equals
\(\mathcal F_{-k}\) when the lemma is applied to fold \(k\).

\begin{lemma}[Flow-to-kernel stability]
\label{lem:flow-kernel-stability}
Let \(\mathcal F_{\rm geom}\) denote the geometry-training sigma-field.
Suppose there exists an event \(E_n\in\mathcal F_{\rm geom}\), with
\(\Pb(E_n)\to1\), on which \(\Delta_J\le1/2\) and, uniformly over
\(t\in[0,1]\), for some \(D_\kappa\ge0\),
\[
\left[
\E\left\{
\int_{\mathcal Y_0}
\kappa_{h,t}(y;Y^a)^2dy
\,\middle|\,
\mathcal F_{\rm geom}
\right\}
\right]^{1/2}
\lesssim
h^{-D_\kappa/2},
\]
and
\[
\left[
\E\left\{
\int_{\mathcal Y_0}
\|\nabla_zq_{\varepsilon_h}\{\Psi_t(y)\mid Y^a\}J_t(y)\|_2^2dy
\,\middle|\,
\mathcal F_{\rm geom}
\right\}
\right]^{1/2}
\lesssim
h^{-D_\kappa/2-1}.
\]
Then, on \(E_n\),
\[
\begin{aligned}
&\left[
\E\left\{
\int_{\mathcal Y_0}
\{\kappa_{h,\widehat\theta}(y;Y^a)
-
\kappa_{h,\theta_0}(y;Y^a)\}^2dy
\,\middle|\,
\mathcal F_{\rm geom}
\right\}
\right]^{1/2} \\
&\quad\lesssim
h^{-D_\kappa/2-1}\Delta_\Phi
+
h^{-D_\kappa/2}\Delta_J.
\end{aligned}
\]
Consequently, if
\[
\Delta_\Phi+\Delta_J
=
O_\Pb(\rho_{\theta,n_{\rm geom}}),
\]
then Assumption~\ref{assump:integrated-geometry-drift} holds with
\[
\Lambda_{\rm geom}(\mathcal Y_0)=D_\kappa+2.
\]
If \(\Delta_\Phi=0\), one may instead take
\(\Lambda_{\rm geom}(\mathcal Y_0)=D_\kappa\).
\end{lemma}

\begin{proof}
Write
\[
\delta\Psi=\Psi_1-\Psi_0,
\quad
\delta J=J_1-J_0.
\]
For fixed \(y\) and \(u\), the fundamental theorem of calculus gives
\[
\kappa_{h,1}(y;u)-\kappa_{h,0}(y;u)
=
\int_0^1
\frac{d}{dt}\kappa_{h,t}(y;u)\,dt,
\]
where
\[
\frac{d}{dt}\kappa_{h,t}(y;u)
=
\nabla_zq_{\varepsilon_h}\{\Psi_t(y)\mid u\}^{\top}
\delta\Psi(y)J_t(y)
+
q_{\varepsilon_h}\{\Psi_t(y)\mid u\}\delta J(y).
\]
On \(E_n\), \(\|\delta\Psi(y)\|_2\le\Delta_\Phi\). Moreover,
\(\Delta_J\le1/2\) implies \(J_t/J_0\in[1/2,3/2]\), and hence
\[
\frac{|\delta J(y)|}{J_t(y)}
=
\frac{J_0(y)}{J_t(y)}
\left|
\frac{J_1(y)}{J_0(y)}-1
\right|
\lesssim
\Delta_J.
\]
Since \(\kappa_{h,t}(y;u)
=
q_{\varepsilon_h}\{\Psi_t(y)\mid u\}J_t(y)\), it follows that
\[
\left|
\frac{d}{dt}\kappa_{h,t}(y;u)
\right|
\lesssim
\Delta_\Phi
\|\nabla_zq_{\varepsilon_h}\{\Psi_t(y)\mid u\}J_t(y)\|_2
+
\Delta_J\kappa_{h,t}(y;u).
\]
Minkowski's inequality and the assumed uniform conditional bounds therefore
yield
\[
\begin{aligned}
&\left[
\E\left\{
\int_{\mathcal Y_0}
\{\kappa_{h,1}(y;Y^a)-\kappa_{h,0}(y;Y^a)\}^2dy
\,\middle|\,
\mathcal F_{\rm geom}
\right\}
\right]^{1/2} \\
&\quad\lesssim
\int_0^1
\left[
\Delta_\Phi h^{-D_\kappa/2-1}
+
\Delta_J h^{-D_\kappa/2}
\right]dt \\
&\quad\lesssim
h^{-D_\kappa/2-1}\Delta_\Phi
+
h^{-D_\kappa/2}\Delta_J.
\end{aligned}
\]
Because
\(\kappa_{h,1}=\kappa_{h,\widehat\theta}\) and
\(\kappa_{h,0}=\kappa_{h,\theta_0}\), this proves the first claim.

If
\(\Delta_\Phi+\Delta_J
=
O_\Pb(\rho_{\theta,n_{\rm geom}})\), then, since \(0<h\le1\),
\[
\begin{aligned}
&\E\left[
\int_{\mathcal Y_0}
\{\kappa_{h,\widehat\theta}(y;Y^a)
-
\kappa_{h,\theta_0}(y;Y^a)\}^2dy
\,\middle|\,
\mathcal F_{\rm geom}
\right] \\
&\quad=
O_\Pb\!\left(
h^{-D_\kappa-2}\rho_{\theta,n_{\rm geom}}^2
\right),
\end{aligned}
\]
which gives
\(\Lambda_{\rm geom}(\mathcal Y_0)=D_\kappa+2\). In the special case \(\Delta_\Phi=0\), the preceding bound reduces to
\[
\E\left[
\int_{\mathcal Y_0}
\{\kappa_{h,\widehat\theta}(y;Y^a)
-
\kappa_{h,\theta_0}(y;Y^a)\}^2dy
\,\middle|\,
\mathcal F_{\rm geom}
\right]
=
O_\Pb\!\left(
h^{-D_\kappa}\rho_{\theta,n_{\rm geom}}^2
\right).
\]
Hence,
\(\Lambda_{\rm geom}(\mathcal Y_0)=D_\kappa\).
\end{proof}

Thus \(\Lambda_{\rm geom}(\mathcal Y_0)\) is a kernel-amplification exponent,
not an intrinsic-dimension assumption. Perturbations of the inverse-flow
location incur one derivative of the smoothing density, whereas purely
multiplicative perturbations need not.

\subsection{Extension to Kernel-Gradient Drift}

DSS additionally requires stability of the transported kernel gradient. Define
\[
\Delta_{\Phi,1}
=
\Delta_\Phi
+
\sup_{y\in\mathcal Y_0}
\|D\Psi_1(y)-D\Psi_0(y)\|_{\rm op},
\]
and
\[
\Delta_{J,1}
=
\Delta_J
+
\sup_{y\in\mathcal Y_0}
\frac{\|\nabla J_1(y)-\nabla J_0(y)\|_2}{J_0(y)}.
\]
These quantities additionally control the first spatial derivatives of the
inverse flow and Jacobian factor.

\begin{lemma}[Flow-to-kernel-gradient stability]
\label{lem:flow-kernel-gradient-stability}
Let \(\mathcal F_{\rm geom}\) denote the geometry-training sigma-field.
Suppose \(0<h\le1\) and there exists an event
\(E_n\in\mathcal F_{\rm geom}\), with \(\Pb(E_n)\to1\), on which
\(\Delta_J\le1/2\). Assume that, on \(E_n\),
\[
\sup_{t\in[0,1]}
\sup_{y\in\mathcal Y_0}
\left\{
\|D\Psi_t(y)\|_{\rm op}
+
\frac{\|\nabla J_t(y)\|_2}{J_t(y)}
\right\}
\lesssim1.
\]
Suppose also that, on \(E_n\), conditionally on
\(\mathcal F_{\rm geom}\) and uniformly over \(t\in[0,1]\), for some
\(D_\kappa\ge0\) and \(j=0,1,2\),
\[
\left[
\E\left\{
\int_{\mathcal Y_0}
\|\nabla_z^j
q_{\varepsilon_h}\{\Psi_t(y)\mid Y^a\}J_t(y)\|^2dy
\,\middle|\,
\mathcal F_{\rm geom}
\right\}
\right]^{1/2}
\lesssim
h^{-D_\kappa/2-j},
\]
where \(\nabla_z^0q=q\). Then, on \(E_n\),
\[
\begin{aligned}
&\left[
\E\left\{
\int_{\mathcal Y_0}
\left[
\{\kappa_{h,\widehat\theta}(y;Y^a)
-
\kappa_{h,\theta_0}(y;Y^a)\}^2 \right.\right.\right.\\
&\hspace{46mm}\left.\left.\left.
+
\|\dot\kappa_{h,\widehat\theta}(y;Y^a)
-
\dot\kappa_{h,\theta_0}(y;Y^a)\|_2^2
\right]dy
\,\middle|\,
\mathcal F_{\rm geom}
\right\}
\right]^{1/2} \\
&\quad\lesssim
h^{-D_\kappa/2-2}\Delta_{\Phi,1}
+
h^{-D_\kappa/2-1}\Delta_{J,1}.
\end{aligned}
\]
Consequently, if
\(\Delta_{\Phi,1}+\Delta_{J,1}
=
O_\Pb(\rho_{\theta,n_{\rm geom}})\), then
Assumption~\ref{assump:integrated-score-geometry-drift} holds with the
conservative exponent
\(\Lambda_{s,\rm geom}(\mathcal Y_0)=D_\kappa+4\). If
\(\Delta_{\Phi,1}=0\), the purely multiplicative bound gives
\(\Lambda_{s,\rm geom}(\mathcal Y_0)=D_\kappa+2\).
\end{lemma}

\begin{proof}
Differentiating
\(\kappa_{h,t}(y;u)
=
q_{\varepsilon_h}\{\Psi_t(y)\mid u\}J_t(y)\)
with respect to \(y\) gives
\[
\dot\kappa_{h,t}(y;u)
=
D\Psi_t(y)^\top
\nabla_zq_{\varepsilon_h}\{\Psi_t(y)\mid u\}J_t(y)
+
q_{\varepsilon_h}\{\Psi_t(y)\mid u\}\nabla J_t(y).
\]
Write
\[
q_t=q_{\varepsilon_h}\{\Psi_t(y)\mid u\},
\quad
\delta\Psi=\Psi_1-\Psi_0,
\quad
\delta J=J_1-J_0.
\]
Differentiating the preceding display along the interpolation path yields
\[
\begin{aligned}
\frac{d}{dt}\dot\kappa_{h,t}(y;u)
&=
D\delta\Psi(y)^\top\nabla_zq_tJ_t
+
D\Psi_t(y)^\top\nabla_z^2q_t\,\delta\Psi(y)J_t \\
&\quad+
D\Psi_t(y)^\top\nabla_zq_t\,\delta J(y)
+
\{\nabla_zq_t^\top\delta\Psi(y)\}\nabla J_t(y)
+
q_t\nabla\delta J(y).
\end{aligned}
\]
Because \(J_t\asymp J_0\), the assumed pathwise bounds and the definitions of
\(\Delta_{\Phi,1}\) and \(\Delta_{J,1}\) imply
\[
\begin{aligned}
\left\|
\frac{d}{dt}\dot\kappa_{h,t}(y;u)
\right\|_2
&\lesssim
\Delta_{\Phi,1}
\left[
\|\nabla_zq_tJ_t\|_2
+
\|\nabla_z^2q_tJ_t\|
\right] \\
&\quad+
\Delta_{J,1}
\left[
|q_tJ_t|
+
\|\nabla_zq_tJ_t\|_2
\right].
\end{aligned}
\]
The fundamental theorem of calculus, Minkowski's inequality, and the assumed
concentration bounds then give
\[
\begin{aligned}
&\left[
\E\left\{
\int_{\mathcal Y_0}
\|\dot\kappa_{h,1}(y;Y^a)
-
\dot\kappa_{h,0}(y;Y^a)\|_2^2dy
\,\middle|\,
\mathcal F_{\rm geom}
\right\}
\right]^{1/2} \\
&\quad\lesssim
h^{-D_\kappa/2-2}\Delta_{\Phi,1}
+
h^{-D_\kappa/2-1}\Delta_{J,1}.
\end{aligned}
\]
Combining this display with
Lemma~\ref{lem:flow-kernel-stability} gives the joint kernel-and-gradient
bound. Squaring the result yields
\[
\begin{aligned}
&\E\left[
\int_{\mathcal Y_0}
\left[
\{\kappa_{h,\widehat\theta}(y;Y^a)
-
\kappa_{h,\theta_0}(y;Y^a)\}^2 \right.\right.\\
&\hspace{42mm}\left.\left.
+
\|\dot\kappa_{h,\widehat\theta}(y;Y^a)
-
\dot\kappa_{h,\theta_0}(y;Y^a)\|_2^2
\right]dy
\,\middle|\,
\mathcal F_{\rm geom}
\right] \\
&\quad=
O_\Pb\left(
h^{-D_\kappa-4}\Delta_{\Phi,1}^2
+
h^{-D_\kappa-2}\Delta_{J,1}^2
\right),
\end{aligned}
\]
which proves the stated consequences.
\end{proof}

The exponents in
Lemmas~\ref{lem:flow-kernel-stability}
and~\ref{lem:flow-kernel-gradient-stability} are sufficient conservative
bounds. Additional structure in the probability flow or smoothing rule may
yield smaller amplification exponents.

\subsection{From Diffusion-Score Learning to Flow Stability}

Let
\[
v_\theta(z,t)
=
b(z,t)-\frac{1}{2}\sigma(t)^2s^{\rm diff}_\theta(z,t)
\]
be the probability-flow vector field. Suppose the vector fields, their required
spatial derivatives, and the inverse-flow Jacobians are uniformly regular on
\(\mathcal R_h\) over \(\mathcal T_h\). Standard continuous-dependence bounds
for the probability-flow ordinary differential equation and its spatial
derivative equations, together with Grönwall's inequality, then give
\[
\Delta_\Phi+\Delta_J
\lesssim
\|s^{\rm diff}_{\widehat\theta}
-
s^{\rm diff}_{\theta_0}\|_{\mathcal S_h^{(1)}},
\]
and
\[
\Delta_{\Phi,1}+\Delta_{J,1}
\lesssim
\|s^{\rm diff}_{\widehat\theta}
-
s^{\rm diff}_{\theta_0}\|_{\mathcal S_h^{(2)}}.
\]
The constants depend on the smoothness, stability, and inverse-Jacobian bounds
of the population flow. The first display reflects that DIS requires control
of the flow and Jacobian, whereas the second reflects the additional derivative
control required for the DSS kernel gradient.

Combining these displays with
Lemmas~\ref{lem:flow-kernel-stability}
and~\ref{lem:flow-kernel-gradient-stability} gives primitive sufficient
conditions for
Assumptions~\ref{assump:integrated-geometry-drift}
and~\ref{assump:integrated-score-geometry-drift}, respectively. Alternatively,
one may begin with direct flow-learning guarantees for
\(\Delta_\Phi+\Delta_J\) or
\(\Delta_{\Phi,1}+\Delta_{J,1}\). The main results require only the induced
kernel or kernel-gradient perturbation and therefore remain modular with
respect to the particular diffusion-score training method.

The corresponding pointwise arguments are obtained by fixing
\(y\in\mathcal Y_0\) rather than integrating over \(\mathcal Y_0\). They yield
the pointwise stability bounds indexed by
\(\lambda_{\rm geom}(y)\) for DIS and
\(\lambda_{s,\rm geom}(y)\) for DSS in the main text.

\section{Intrinsic Score Learning Regimes}
\label{appsec:intrinsic-score-learning}

The geometry-adaptive gain is preserved only when
\(\rho_{\theta,n_{\rm geom}}\) is governed by intrinsic rather than ambient
complexity. Recall that \(\rho_{\theta,n_{\rm geom}}\) denotes a rate in the
geometry-learning norm required for flow stability:
\(\mathcal S_h^{(1)}\) for DIS and \(\mathcal S_h^{(2)}\) for DSS. The examples
below illustrate structural regimes in which such rates may depend on intrinsic
complexity; they do not derive the required strong-norm rates from ordinary
score-estimation rates.

Throughout this section,
\(\widetilde O_\Pb(r_N)\) denotes
\(O_\Pb\{r_N(\log N)^c\}\) for some fixed \(c\ge0\); thus the notation
suppresses only polylogarithmic factors in the relevant sample size
\(N=n_{\rm geom}\).

\textbf{Low-dimensional support.}
Suppose that, for each \(t\in\mathcal T_h\), the diffusion-perturbed
counterfactual law \(\Pb_{a,t}\) is concentrated in a tubular neighborhood of an
\(m\)-dimensional \(C^2\) submanifold \(\mathcal M_t\subset\mathbb R^d\), with
\(m\ll d\). Let \(\mathcal U_t\) denote the relevant tubular neighborhood and let
\(\tau_t:\mathcal U_t\to\mathbb R^m\) be a local coordinate map. Assume that the tangential component of the diffusion score
admits a representation
\[
s_{\theta_0}^{\rm diff}(z,t)
=
S_t(\tau_t(z))
+
r_t(z),
\quad
\|r_t\|_{L_2(\Pb_{a,t})}
\le
\epsilon_{\mathrm{nor}},
\]
where \(S_t\) is \(\alpha\)-smooth uniformly over \(t\in\mathcal T_h\), and the
normal residual \(\epsilon_{\mathrm{nor}}\) is negligible relative to the
statistical rate. If, under this structure, a learner attains the corresponding rate in the
required geometry-learning norm, then
\[
\rho_{\theta,n_{\rm geom}}
=
\widetilde O_\Pb\!\left(
n_{\rm geom}^{-\alpha/(2\alpha+m)}
\right),
\]
rather than the ambient rate
\(n_{\rm geom}^{-\alpha/(2\alpha+d)}\). Under the pointwise drift condition, the
squared geometry drift is then
\[
\widetilde O_\Pb\!\left(
h^{-2\lambda_{\rm geom}(y)}
n_{\rm geom}^{-2\alpha/(2\alpha+m)}
\right).
\]
This type of intrinsic scaling is consistent with recent analyses of
score-based diffusion models under low-dimensional subspace or manifold
structure.

\textbf{Block structure.}
Suppose \(y=(y_{S_1},\ldots,y_{S_B})\) with disjoint blocks \(S_b\) of sizes
\(d_b\), and assume the diffusion score decomposes as
\[
s_{\theta_0}^{\rm diff}(y,t)
=
\{s_{\theta_0,b}^{\rm diff}(y_{S_b},t): b=1,\ldots,B\},
\]
with \(s_{\theta_0,b}^{\rm diff}(\cdot,t)\) belonging to an
\(\alpha_b\)-smooth class on \(\mathbb R^{d_b}\), uniformly over
\(t\in\mathcal T_h\). If the block decomposition is known or recovered with negligible error and the
learner attains the corresponding blockwise rate in the required
geometry-learning norm, then
\[
\rho_{\theta,n_{\rm geom}}^2
=
\widetilde O_\Pb\!\left(
\sum_{b=1}^B
n_{\rm geom}^{-2\alpha_b/(2\alpha_b+d_b)}
\right).
\]
For fixed \(B\), this rate is governed by the hardest block rather than by the
full ambient dimension \(d=\sum_{b=1}^B d_b\). If the kernel perturbation is also
block adaptive, then the pointwise amplification exponent
\(\lambda_{\rm geom}(y)\) may likewise depend on the largest active block rather
than on the ambient dimension.

\textbf{Graphical sparsity.}
Suppose the score field has a local Markov structure: for each coordinate \(j\),
\[
s_{\theta_0,j}^{\rm diff}(y,t)
=
S_{j,t}(y_{N(j)}),
\quad
|N(j)|\le d_0\ll d,
\]
where \(S_{j,t}\) is \(\alpha\)-smooth uniformly over \(t\in\mathcal T_h\). If the neighborhoods \(N(j)\) are known or selected consistently and a sparse
learner attains the corresponding rate in the required geometry-learning norm,
then
\[
\rho_{\theta,n_{\rm geom}}^2
=
\widetilde O_\Pb\!\left(
s_d\,n_{\rm geom}^{-2\alpha/(2\alpha+d_0)}
\right),
\]
where \(s_d\) denotes the effective number or complexity of active local
components. Thus the nonparametric part of the score learning error is governed
by the local neighborhood size \(d_0\), rather than by the ambient dimension
\(d\). The ambient dimension enters only through the combinatorial or sparsity
complexity \(s_d\).

\textbf{Low-rank local curvature.}
Suppose that, locally around the evaluation region, score variation occurs
through an \(r\)-dimensional active subspace. More precisely, assume that there exists \(P_t\in\mathbb R^{r\times d}\) satisfying
\(P_tP_t^\top=I_r\) such that
\[
s_{\theta_0}^{\rm diff}(y,t)
=
P_t^\top S_t(P_t y)+r_t(y)
\quad
\|r_t\|_{L_2(\Pb_{a,t})}
\le
\epsilon_{\mathrm{lr}},
\]
where \(S_t:\mathbb R^r\to\mathbb R^r\) is \(\alpha\)-smooth and
\(\epsilon_{\mathrm{lr}}\) is negligible. Equivalently, the Jacobian of the
score field or the local curvature of the probability flow has rank at most
\(r\) in the directions that affect the transported kernel. If, under this structure, a learner attains the corresponding rate in the
required geometry-learning norm, then
\[
\rho_{\theta,n_{\rm geom}}
=
\widetilde O_\Pb\!\left(
n_{\rm geom}^{-\alpha/(2\alpha+r)}
\right),
\quad r\ll d.
\]
If the kernel perturbation is concentrated along these active directions, the
pointwise geometry drift scales as
\[
\widetilde O_\Pb\!\left(
h^{-\lambda_{\rm geom}(y)}
n_{\rm geom}^{-\alpha/(2\alpha+r)}
\right).
\]
In the ideal active-subspace case
\(\lambda_{\rm geom}(y)=r/2\), this becomes
\[
\widetilde O_\Pb\!\left(
h^{-r/2}
n_{\rm geom}^{-\alpha/(2\alpha+r)}
\right).
\]

These are conditional rate illustrations, not primitive rate results.
Existing diffusion-learning theory supports intrinsic dependence for score
approximation, estimation, or distribution recovery, but does not by itself
establish the \(\mathcal S_h^{(1)}\) or \(\mathcal S_h^{(2)}\) rates required
here
\citep[e.g.,][]{pidstrigach2022score,chen2023score,tang2024adaptivity}.

For DIS, the learned-geometry contribution is
\(h^{-\Lambda_{\rm geom}(\mathcal Y_0)}
\rho_{\theta,n_{\rm geom}}^2\) in integrated risk and
\(h^{-2\lambda_{\rm geom}(y)}\rho_{\theta,n_{\rm geom}}^2\) pointwise. For DSS,
the corresponding terms use \(\Lambda_{s,\rm geom}(\mathcal Y_0)\) and
\(\lambda_{s,\rm geom}(y)\), with
\(\rho_{\theta,n_{\rm geom}}\) interpreted in the stronger
\(\mathcal S_h^{(2)}\) norm. These terms should be compared with the oracle
smoothing variance and bias.

\section{Proofs}

\textbf{Cross-fitting convention.}
Under \(K\)-fold cross fitting, each sample-splitting argument below is applied
separately on evaluation fold \(k\), conditional on \(\mathcal F_{-k}\), and
the resulting fold-specific statements are then aggregated over the fixed
number \(K\) of folds. Within a generic fold, \(\Pn\),
\(\widehat\theta\), and the fitted nuisance functions denote fold-specific
quantities, with fold indices suppressed. Whenever an identity combines all folds, empirical and population terms,
targets, and remainders involving fitted quantities denote their corresponding
fold-weighted aggregates. Within a generic fold, all expectations and
\(L_2\) norms involving fitted quantities are understood conditionally on
\(\mathcal F_{-k}\). No argument conditions jointly on all fold-specific
training samples. Since \(n_k\asymp n\), this convention does not alter any
stated rate.

\subsection{Section \ref{subsec:localization}}

\subsubsection{Proof of Lemma \ref{lem:kappa-basic}}

\begin{proof}
We condition on a value of \(Y\) for which the stated assumptions hold, and
suppress this conditioning throughout.

\emph{(i) Normalization.}
Write \(\Phi=\Phi_{\varepsilon_h,\theta_0}\),
\(q(v)=q_{\varepsilon_h}(v\mid Y)\), and
\(J(u)=|\det D\Phi^{-1}(u)|\). Then
\[
\kappa_{h,\theta_0}(u;Y)
=
q\{\Phi^{-1}(u)\}J(u).
\]
Since \(\Phi\) is a \(C^1\) diffeomorphism from \(\mathbb R^d\) onto
\(\mathbb R^d\), the change of variables \(v=\Phi^{-1}(u)\), equivalently
\(u=\Phi(v)\), gives \(du=|\det D\Phi(v)|dv\) and
\(J\{\Phi(v)\}=|\det D\Phi(v)|^{-1}\). Hence
\[
\begin{aligned}
\int_{\mathbb R^d}\kappa_{h,\theta_0}(u;Y)\,du
&=
\int_{\mathbb R^d}
q\{\Phi^{-1}(u)\}
|\det D\Phi^{-1}(u)|\,du \\
&=
\int_{\mathbb R^d}
q(v)
|\det D\Phi(v)|^{-1}
|\det D\Phi(v)|\,dv \\
&=
\int_{\mathbb R^d}q(v)\,dv
=
1 .
\end{aligned}
\]

\emph{(ii) Local ellipsoid volume.}
Since \(G_h(Y)\succ0\), let \(G_h(Y)^{1/2}\) be its symmetric positive
definite square root. Set
\[
v=G_h(Y)^{1/2}\{u-m_h(Y)\}.
\]
Then
\[
\{u-m_h(Y)\}^{\top}G_h(Y)\{u-m_h(Y)\}
=
\|v\|_2^2,
\]
and
\[
u=m_h(Y)+G_h(Y)^{-1/2}v,
\quad
du=\det\{G_h(Y)\}^{-1/2}dv.
\]
Thus \(u\in E_h(Y;c)\) if and only if \(v\in \sqrt c\,B_d\). Therefore
\[
\begin{aligned}
\mathrm{Vol}\{E_h(Y;c)\}
&=
\int_{\sqrt c\,B_d}
\det\{G_h(Y)\}^{-1/2}\,dv  \\
&=
\det\{G_h(Y)\}^{-1/2}\mathrm{Vol}(\sqrt c\,B_d) \\
&=
\mathrm{Vol}(B_d)c^{d/2}
\det\{G_h(Y)\}^{-1/2}.
\end{aligned}
\]

\emph{(iii) Kernel \(L_2\) size.}
By Assumption~\ref{assump:subgauss-kappa},
\[
\kappa_{h,\theta_0}(u;Y)^2
\le
C_0^2
\det\{G_h(Y)\}
\exp\{-2c_0\|u-m_h(Y)\|_{G_h(Y)}^2\}.
\]
Using the same change of variables
\(v=G_h(Y)^{1/2}\{u-m_h(Y)\}\), we obtain
\[
\begin{aligned}
\int_{\mathbb R^d}\kappa_{h,\theta_0}(u;Y)^2\,du
&\le
C_0^2
\det\{G_h(Y)\}
\int_{\mathbb R^d}
\exp\{-2c_0\|u-m_h(Y)\|_{G_h(Y)}^2\}\,du \\
&=
C_0^2
\det\{G_h(Y)\}^{1/2}
\int_{\mathbb R^d}
\exp(-2c_0\|v\|_2^2)\,dv \\
&=
C_0^2
\left(\frac{\pi}{2c_0}\right)^{d/2}
\det\{G_h(Y)\}^{1/2}.
\end{aligned}
\]
Thus the result holds with
\[
C=C_0^2\left(\frac{\pi}{2c_0}\right)^{d/2}.
\]
\end{proof}

\subsection{Section \ref{subsec:asymptotics-risk}}

\subsubsection{Proof of Lemma \ref{lem:fixed-geometry-expansion}}

\begin{proof}
Fix \(h>0\) and \(y\in\mathbb R^d\). For notational simplicity, write
\[
\kappa_0(Y)=\kappa_{h,\theta_0}(y;Y),
\quad
\mu_0(X)=\mu_{a,h,\theta_0}(X;y),
\quad
\widehat\mu(X)=\widehat\mu_{a,h,\theta_0}(X;y),
\]
and \(\pi(X)=\pi_a(X)\), \(\widehat\pi(X)=\widehat\pi_a(X)\). Let
\begin{align*}  
\varphi_0(Z)
&=
\frac{\mathbbm 1\{A=a\}}{\pi(X)}
\{\kappa_0(Y)-\mu_0(X)\}
+
\mu_0(X),\\
\widehat\varphi(Z)
&=
\frac{\mathbbm 1\{A=a\}}{\widehat\pi(X)}
\{\kappa_0(Y)-\widehat\mu(X)\}
+
\widehat\mu(X).
\end{align*}
Then \(\varphi_0(Z)=\varphi^{\mathrm{geo}}_h(Z;y,\eta_0)\), \(\widehat\varphi(Z)=\varphi^{\mathrm{geo}}_h(Z;y,\widehat\eta^0)\), and $\widehat p^{\mathrm{geo},0}_{a,h}(y)=\Pn\widehat\varphi$.

Moreover, since \(\E\{\kappa_0(Y)\mid X,A=a\}=\mu_0(X)\),
\[
\Pb\varphi_0
=
\E\{\mu_0(X)\}
=
p^{\mathrm{geo}}_{a,h}(y).
\]
Hence we get the decomposition:
\[
\widehat p^{\mathrm{geo},0}_{a,h}(y)-p^{\mathrm{geo}}_{a,h}(y)
=
(\Pn-\Pb)\varphi_0
+
\Pb(\widehat\varphi-\varphi_0)
+
(\Pn-\Pb)(\widehat\varphi-\varphi_0).
\]

We first compute the drift term. Taking conditional expectations given \(X\),
\[
\Pb\widehat\varphi
=
\Pb\left[
\frac{\pi(X)}{\widehat\pi(X)}
\{\mu_0(X)-\widehat\mu(X)\}
+
\widehat\mu(X)
\right],
\]
and simple algebra yields
\[
\Pb(\widehat\varphi-\varphi_0)
=
\Pb\left[
\left\{1-\frac{\pi(X)}{\widehat\pi(X)}\right\}
\{\widehat\mu(X)-\mu_0(X)\}
\right].
\]
Define
\[
R_{\pi,\mu}^0(y)
=
\Pb\left[
\left\{1-\frac{\pi(X)}{\widehat\pi(X)}\right\}
\{\widehat\mu(X)-\mu_0(X)\}
\right].
\]
On the event \(\inf_x\widehat\pi(x)\ge \pi_{\min}/2\),
\[
\left|1-\frac{\pi(X)}{\widehat\pi(X)}\right|
=
\frac{|\widehat\pi(X)-\pi(X)|}{\widehat\pi(X)}
\le
\frac{2}{\pi_{\min}}
|\widehat\pi(X)-\pi(X)|.
\]
By Cauchy Schwarz,
\[
|R_{\pi,\mu}^0(y)|
\le
\frac{2}{\pi_{\min}}
\|\widehat\pi_a-\pi_a\|_{L_2(\Pb)}
\,
\|\widehat\mu_{a,h,\theta_0}(\cdot;y)-\mu_{a,h,\theta_0}(\cdot;y)\|_{L_2(\Pb)}
\lesssim
\delta_\pi\,\delta_\mu^0(y).
\]

It remains to control the empirical process remainder. By the stated consistency conditions, positivity, and square integrability in Assumption~\ref{assumption:boundedness},
\[
\|\widehat\varphi-\varphi_0\|_{L_2(\Pb)}
=
o_\Pb(1).
\]
By the foldwise cross-fitting convention and
\citet[][Lemma~2]{kennedy2020sharp},
\[
(\Pn-\Pb)(\widehat\varphi-\varphi_0)
=
O_\Pb\left(
n^{-1/2}
\|\widehat\varphi-\varphi_0\|_{L_2(\Pb)}
\right)
=
o_\Pb(n^{-1/2}).
\]
Under the Donsker alternative in Assumption~\ref{assumption:sample-splitting}, the same conclusion follows by stochastic equicontinuity \citep[][Chapter 2]{van1996weak}.

Combining the preceding displays gives
\[
\widehat p^{\mathrm{geo},0}_{a,h}(y)
-
p^{\mathrm{geo}}_{a,h}(y)
=
(\Pn-\Pb)\{\varphi^{\mathrm{geo}}_h(Z;y,\eta_0)\}
+
R_{\pi,\mu}^0(y)
+
o_\Pb(n^{-1/2}),
\]
with
\[
|R_{\pi,\mu}^0(y)|
\lesssim
\delta_\pi\delta_\mu^0(y),
\]
\end{proof}

\subsubsection{Proof of Theorem \ref{thm:onestep-expansion}}

\begin{proof}

We use the foldwise cross-fitting convention stated above.

Fix \(h>0\) and \(y\in\mathbb R^d\). For notational simplicity, write
\[
\kappa_{\widehat\theta}(Y)=\kappa_{h,\widehat\theta}(y;Y),
\quad
\kappa_0(Y)=\kappa_{h,\theta_0}(y;Y),
\]
\[
\mu_{\widehat\theta}(X)=\mu_{a,h,\widehat\theta}(X;y),
\quad
\widehat\mu_{\widehat\theta}(X)=\widehat\mu_{a,h,\widehat\theta}(X;y),
\quad
\mu_0(X)=\mu_{a,h,\theta_0}(X;y),
\]
and \(\pi(X)=\pi_a(X)\), \(\widehat\pi(X)=\widehat\pi_a(X)\). Also define
\[
\eta_{\widehat\theta}
=
\{\pi_a,\mu_{a,h,\widehat\theta},\kappa_{h,\widehat\theta}\},
\quad
\widehat\eta_{\widehat\theta}
=
\{\widehat\pi_a,\widehat\mu_{a,h,\widehat\theta},\kappa_{h,\widehat\theta}\}.
\]
Let
\[
\varphi_{\widehat\theta}(Z;y)
=
\varphi_h^{\mathrm{geo}}(Z;y,\eta_{\widehat\theta})
=
\frac{\mathbbm 1\{A=a\}}{\pi(X)}
\{\kappa_{\widehat\theta}(Y)-\mu_{\widehat\theta}(X)\}
+
\mu_{\widehat\theta}(X),
\]
\[
\widehat\varphi_{\widehat\theta}(Z;y)
=
\varphi_h^{\mathrm{geo}}(Z;y,\widehat\eta_{\widehat\theta})
=
\frac{\mathbbm 1\{A=a\}}{\widehat\pi(X)}
\{\kappa_{\widehat\theta}(Y)-\widehat\mu_{\widehat\theta}(X)\}
+
\widehat\mu_{\widehat\theta}(X),
\]
and
\[
\varphi_0(Z;y)
=
\varphi_h^{\mathrm{geo}}(Z;y,\eta_0)
=
\frac{\mathbbm 1\{A=a\}}{\pi(X)}
\{\kappa_0(Y)-\mu_0(X)\}
+
\mu_0(X).
\]
Then
\[
\widehat p^{\mathrm{geo}}_{a,h}(y)
=
\Pn \widehat\varphi_{\widehat\theta}.
\]
Moreover, by the definition of the localized regressions,
\[
\Pb\varphi_{\widehat\theta}
=
p^{\mathrm{geo}}_{a,h,\widehat\theta}(y),
\quad
\Pb\varphi_0
=
p^{\mathrm{geo}}_{a,h,\theta_0}(y)
=
p^{\mathrm{geo}}_{a,h}(y).
\]

By adding and subtracting \(\varphi_{\widehat\theta}\) and \(\varphi_0\), we obtain
\[
\begin{aligned}
\widehat p^{\mathrm{geo}}_{a,h}(y)-p^{\mathrm{geo}}_{a,h}(y)
&=
(\Pn-\Pb)\varphi_0 \\
&\quad+
(\Pn-\Pb)(\varphi_{\widehat\theta}-\varphi_0) \\
&\quad+
(\Pn-\Pb)(\widehat\varphi_{\widehat\theta}-\varphi_{\widehat\theta}) \\
&\quad+
\Pb(\widehat\varphi_{\widehat\theta}-\varphi_{\widehat\theta}) \\
&\quad+
\{\Pb\varphi_{\widehat\theta}-\Pb\varphi_0\}.
\end{aligned}
\]
Note that the last term is
\[
\Pb\varphi_{\widehat\theta}-\Pb\varphi_0
=
p^{\mathrm{geo}}_{a,h,\widehat\theta}(y)
-
p^{\mathrm{geo}}_{a,h,\theta_0}(y)
=
B_{\theta,h}(y).
\]

We now analyze the second, third, and fourth terms in turn.

First consider the geometry empirical term
\((\Pn-\Pb)(\varphi_{\widehat\theta}-\varphi_0)\). Let
\[
\Delta_\kappa(Y)=\kappa_{\widehat\theta}(Y)-\kappa_0(Y),
\quad
\Delta_\mu(X)=\mu_{\widehat\theta}(X)-\mu_0(X).
\]
Then
\[
\varphi_{\widehat\theta}(Z)-\varphi_0(Z)
=
\frac{\mathbbm 1\{A=a\}}{\pi(X)}
\{\Delta_\kappa(Y)-\Delta_\mu(X)\}
+
\Delta_\mu(X),
\]
where, by the foldwise convention,
\(\Delta_\mu(X)=\E\{\Delta_\kappa(Y)\mid X,A=a,\mathcal F_{-k}\}\).
Thus Jensen's inequality implies
\[
\|\varphi_{\widehat\theta}-\varphi_0\|_{L_2(\Pb)}
\lesssim
\|\kappa_{h,\widehat\theta}(y;Y^a)-\kappa_{h,\theta_0}(y;Y^a)\|_{L_2(\Pb_a)}
=
o_\Pb(1).
\]
Hence, under Assumption~\ref{assumption:sample-splitting},
\[
(\Pn-\Pb)(\varphi_{\widehat\theta}-\varphi_0)
=
O_\Pb\!\left(
n^{-1/2}
\|\varphi_{\widehat\theta}-\varphi_0\|_{L_2(\Pb)}
\right)
=
o_\Pb(n^{-1/2}),
\]
which follows by the sample splitting lemma
\citep[e.g.,][Lemma~2]{kennedy2020sharp}. Under the Donsker alternative, the same conclusion follows by stochastic equicontinuity
\citep[][Chapter~2]{van1996weak}.

Next consider the empirical nuisance term
\((\Pn-\Pb)(\widehat\varphi_{\widehat\theta}-\varphi_{\widehat\theta})\). On the event
\(\inf_x\widehat\pi_a(x)\ge \pi_{\min}/2\),
\begin{align*}
\widehat\varphi_{\widehat\theta}(Z;y)-\varphi_{\widehat\theta}(Z;y)
&=
\mathbbm 1\{A=a\}
\left\{\frac{\pi_a(X)-\widehat\pi_a(X)}{\widehat\pi_a(X)\,\pi_a(X)}\right\}
\{\kappa_{h,\widehat\theta}(y;Y)-\mu_{a,h,\widehat\theta}(X;y)\} \\
&\quad+
\left\{1-\frac{\mathbbm 1\{A=a\}}{\widehat\pi_a(X)}\right\}
\{\widehat\mu_{a,h,\widehat\theta}(X;y)-\mu_{a,h,\widehat\theta}(X;y)\}.
\end{align*}
Therefore, both terms in RHS are \(o_\Pb(1)\) in \(L_2(\Pb)\) by the given conditions \(\delta_\pi=o_\Pb(1)\), \(\delta_\mu^{\widehat\theta}(y)=o_\Pb(1)\), and the foldwise conditional square integrability of the kernel residual. Thus
\[
\|\widehat\varphi_{\widehat\theta}-\varphi_{\widehat\theta}\|_{L_2(\Pb)}
=
o_\Pb(1).
\]
Under Assumption~\ref{assumption:sample-splitting}, applying the sample splitting lemma
\citep[e.g.,][Lemma~2]{kennedy2020sharp} or the Donsker alternative with stochastic equicontinuity argument gives
\[
(\Pn-\Pb)(\widehat\varphi_{\widehat\theta}-\varphi_{\widehat\theta})
=
o_\Pb(n^{-1/2}).
\]

Finally, consider the mean nuisance term
\(\Pb(\widehat\varphi_{\widehat\theta}-\varphi_{\widehat\theta})\). Since, within fold \(k\),
\(\E\{\kappa_{\widehat\theta}(Y)\mid X,A=a,\mathcal F_{-k}\}
=
\mu_{\widehat\theta}(X)\), it follows
\[
R_{\pi,\mu}^{\widehat\theta}(y)
=
\Pb(\widehat\varphi_{\widehat\theta}-\varphi_{\widehat\theta})
=
\Pb\left[
\left\{1-\frac{\pi(X)}{\widehat\pi(X)}\right\}
\{\widehat\mu_{\widehat\theta}(X)-\mu_{\widehat\theta}(X)\}
\right].
\]
On the event \(\inf_x\widehat\pi_a(x)\ge \pi_{\min}/2\),
\[
\left|1-\frac{\pi(X)}{\widehat\pi(X)}\right|
\le
\frac{2}{\pi_{\min}}|\widehat\pi(X)-\pi(X)|.
\]
Therefore, by Cauchy Schwarz,
\[
|R_{\pi,\mu}^{\widehat\theta}(y)|
\lesssim
\|\widehat\pi_a-\pi_a\|_{L_2(\Pb)}
\,
\|\widehat\mu_{a,h,\widehat\theta}(\cdot;y)-\mu_{a,h,\widehat\theta}(\cdot;y)\|_{L_2(\Pb)}
=
\delta_\pi\,\delta_\mu^{\widehat\theta}(y).
\]

Bringing it all together, the final decomposition is
\[
\widehat p^{\mathrm{geo}}_{a,h}(y)
-
p^{\mathrm{geo}}_{a,h}(y)
=
(\Pn-\Pb)\{\varphi^{\mathrm{geo}}_h(Z;y,\eta_0)\}
+
B_{\theta,h}(y)
+
R_{\pi,\mu}^{\widehat\theta}(y)
+
o_\Pb(n^{-1/2}),
\]
with $|R_{\pi,\mu}^{\widehat\theta}(y)|
\lesssim
\delta_\pi\delta_\mu^{\widehat\theta}(y)$, which gives the desired result.
\end{proof}

\subsubsection{Proof of Theorem \ref{thm:dis-risk-global}}

\begin{proof}
Throughout the proof, constants may depend on \(\pi_{\min}\) and
\(\mathrm{Vol}(\mathcal Y_0)\), but not on \(n\). Any dependence on \(h\) is
displayed explicitly through \(H_h(\mathcal Y_0)\) and, for the learned-geometry
estimator, through \(\int_{\mathcal Y_0}B_{\theta,h}(y)^2dy\).

First consider the fixed-geometry estimator
\(\widehat p^{\mathrm{geo},0}_{a,h}\). Letting
\[
G_n(y)
=
(\Pn-\Pb)\{\varphi_h^{\mathrm{geo}}(Z;y,\eta_0)\},
\]
Under cross fitting, let
\(\mathbb P_{n,k}f=n_k^{-1}\sum_{i\in\mathcal I_k}f(Z_i)\) and
\(\widehat\eta^{0,(-k)}
=
\{\widehat\pi_a^{(-k)},
\widehat\mu_{a,h,\theta_0}^{(-k)},
\kappa_{h,\theta_0}\}\). Define
\[
e_{n,k}^0(y)
=
(\mathbb P_{n,k}-\Pb)
\left[
\varphi_h^{\mathrm{geo}}(Z;y,\widehat\eta^{0,(-k)})
-
\varphi_h^{\mathrm{geo}}(Z;y,\eta_0)
\right],
\quad
e_n^0(y)
=
\sum_{k=1}^K\frac{n_k}{n}e_{n,k}^0(y).
\]
Under the empirical-process alternative without sample splitting,
\(e_n^0\) retains its original unsplit definition.
\[
R_{\pi,\mu}^0(y)
=
\Pb
\left[
\varphi_h^{\mathrm{geo}}(Z;y,\widehat\eta^0)
-
\varphi_h^{\mathrm{geo}}(Z;y,\eta_0)
\right],
\]
we have the decomposition
\[
\widehat p^{\mathrm{geo},0}_{a,h}(y)
-
p^{\mathrm{geo}}_{a,h}(y)
=
G_n(y)
+
R_{\pi,\mu}^0(y)
+
e_n^0(y).
\]

We first control the empirical process term \(e_n^0\). On the event
\(\inf_x\widehat\pi_a(x)\ge\pi_{\min}/2\),
\[
\begin{aligned}
&\varphi_h^{\mathrm{geo}}(Z;y,\widehat\eta^0)
-
\varphi_h^{\mathrm{geo}}(Z;y,\eta_0) \\
&\quad =
\mathbbm 1\{A=a\}
\left\{
\frac{\pi_a(X)-\widehat\pi_a(X)}
{\widehat\pi_a(X)\pi_a(X)}
\right\}
\{\kappa_{h,\theta_0}(y;Y)-\mu_{a,h,\theta_0}(X;y)\} \\
&\quad+
\left\{1-\frac{\mathbbm 1\{A=a\}}{\widehat\pi_a(X)}\right\}
\{\widehat\mu_{a,h,\theta_0}(X;y)-\mu_{a,h,\theta_0}(X;y)\}.
\end{aligned}
\]
Using positivity, Jensen's inequality, and the assumed integrated nuisance
consistency, we obtain
\[
\begin{aligned}
&\int_{\mathcal Y_0}
\left\|
\varphi_h^{\mathrm{geo}}(\cdot;y,\widehat\eta^0)
-
\varphi_h^{\mathrm{geo}}(\cdot;y,\eta_0)
\right\|_{L_2(\Pb)}^2dy \\
&\quad\lesssim
\|\widehat\pi_a-\pi_a\|_{L_2(\Pb)}^2
\sup_x
\E\!\left[
\int_{\mathcal Y_0}\kappa_{h,\theta_0}(y;Y)^2dy
\,\middle|\,
X=x,A=a
\right]  \\
&\quad+
\int_{\mathcal Y_0}
\|\widehat\mu_{a,h,\theta_0}(\cdot;y)
-
\mu_{a,h,\theta_0}(\cdot;y)\|_{L_2(\Pb)}^2dy
=
o_\Pb(1).
\end{aligned}
\]
For each fold \(k\), conditional on \(\mathcal F_{-k}\), the observations in
\(\mathcal I_k\) are independent of \(\widehat\eta^{0,(-k)}\). Hence
\[
\begin{aligned}
\E\left[
\int_{\mathcal Y_0}\{e_{n,k}^0(y)\}^2dy
\,\middle|\,
\mathcal F_{-k}
\right]
&\le
\frac{1}{n_k}
\int_{\mathcal Y_0}
\left\|
\varphi_h^{\mathrm{geo}}(\cdot;y,\widehat\eta^{0,(-k)})
-
\varphi_h^{\mathrm{geo}}(\cdot;y,\eta_0)
\right\|_{L_2(\Pb)}^2dy \\
&=
o_\Pb(n^{-1}),
\end{aligned}
\]
where \(n_k\asymp n\). The conditional Markov argument therefore gives
\(\int_{\mathcal Y_0}\{e_{n,k}^0(y)\}^2dy=o_\Pb(n^{-1})\) for every \(k\).
Since \(K\) is fixed, Jensen's inequality yields
\[
\int_{\mathcal Y_0}e_n^0(y)^2dy
\le
\sum_{k=1}^K\frac{n_k}{n}
\int_{\mathcal Y_0}\{e_{n,k}^0(y)\}^2dy
=
o_\Pb(n^{-1}).
\]

Under the Donsker alternative in Assumption~\ref{assumption:sample-splitting},
the same conclusion follows from integrated stochastic equicontinuity
\citep[Chapter~2]{van1996weak}.

Consequently,
\[
\int_{\mathcal Y_0}
\{\widehat p^{\mathrm{geo},0}_{a,h}(y)-p^{\mathrm{geo}}_{a,h}(y)\}^2dy
\lesssim
\int_{\mathcal Y_0}G_n(y)^2dy
+
\int_{\mathcal Y_0}R_{\pi,\mu}^0(y)^2dy
+
o_\Pb(n^{-1}).
\]
Since \(\mathcal Y_0\) has finite Lebesgue measure,
\[
\int_{\mathcal Y_0}R_{\pi,\mu}^0(y)^2dy
\lesssim
\operatorname*{ess\,sup}_{y\in\mathcal Y_0}
\{R_{\pi,\mu}^0(y)\}^2.
\]

It remains to bound the stochastic term. For fixed \(y\), write
\[
\kappa_0(y;Y)=\kappa_{h,\theta_0}(y;Y),
\quad
\mu_0(X;y)=\mu_{a,h,\theta_0}(X;y).
\]
Then
\[
\varphi_h^{\mathrm{geo}}(Z;y,\eta_0)
=
\frac{\mathbbm 1\{A=a\}}{\pi_a(X)}
\{\kappa_0(y;Y)-\mu_0(X;y)\}
+
\mu_0(X;y).
\]
Since
\[
\E\{\kappa_0(y;Y)\mid X,A=a\}
=
\mu_0(X;y),
\]
we have
\[
\E\{\varphi_h^{\mathrm{geo}}(Z;y,\eta_0)^2\mid X\}
=
\frac{1}{\pi_a(X)}
\E\!\left[
\{\kappa_0(y;Y)-\mu_0(X;y)\}^2
\mid X,A=a
\right]
+
\mu_0(X;y)^2.
\]
By Jensen's inequality and positivity,
\[
\E\{\varphi_h^{\mathrm{geo}}(Z;y,\eta_0)^2\}
\lesssim
\E\{\kappa_{h,\theta_0}(y;Y^a)^2\}.
\]
Now, by Fubini's theorem and independence of observations,
\[
\begin{aligned}
\E\left[
\int_{\mathcal Y_0}G_n(y)^2dy
\right]
&=
\frac{1}{n}
\int_{\mathcal Y_0}
\var\{\varphi_h^{\mathrm{geo}}(Z;y,\eta_0)\}dy \\
&\le
\frac{1}{n}
\int_{\mathcal Y_0}
\E\{\varphi_h^{\mathrm{geo}}(Z;y,\eta_0)^2\}dy \\
&\lesssim
\frac{1}{n}
\E\left[
\int_{\mathcal Y_0}
\kappa_{h,\theta_0}(y;Y^a)^2dy
\right] \\
&=
\frac{H_h(\mathcal Y_0)}{n}.
\end{aligned}
\]
Hence
\[
\int_{\mathcal Y_0}G_n(y)^2dy
=
O_\Pb\!\left\{\frac{H_h(\mathcal Y_0)}{n}\right\}.
\]
Combining the preceding displays gives
\[
\int_{\mathcal Y_0}
\{\widehat p^{\mathrm{geo},0}_{a,h}(y)-p^{\mathrm{geo}}_{a,h}(y)\}^2dy
=
O_\Pb\!\left[
\frac{H_h(\mathcal Y_0)}{n}
+
\operatorname*{ess\,sup}_{y\in\mathcal Y_0}
\{R_{\pi,\mu}^0(y)\}^2
\right]
+
o_\Pb(n^{-1}),
\]
which gives the first result.

Now consider the proposed DIS estimator with learned geometry. For notational
simplicity, write
\[
\eta_{\widehat\theta}
=
\{\pi_a,\mu_{a,h,\widehat\theta},\kappa_{h,\widehat\theta}\},
\quad
\widehat\eta_{\widehat\theta}
=
\{\widehat\pi_a,\widehat\mu_{a,h,\widehat\theta},\kappa_{h,\widehat\theta}\},
\]
and define
\[
\varphi_{\widehat\theta}(Z;y)
=
\varphi_h^{\mathrm{geo}}(Z;y,\eta_{\widehat\theta}),
\quad
\widehat\varphi_{\widehat\theta}(Z;y)
=
\varphi_h^{\mathrm{geo}}(Z;y,\widehat\eta_{\widehat\theta}).
\]
Then
\[
\widehat p^{\mathrm{geo}}_{a,h}(y)
=
\Pn\widehat\varphi_{\widehat\theta}(Z;y),
\quad
\Pb\varphi_{\widehat\theta}(Z;y)
=
p^{\mathrm{geo}}_{a,h,\widehat\theta}(y),
\]
and
\[
\Pb\varphi_h^{\mathrm{geo}}(Z;y,\eta_0)
=
p^{\mathrm{geo}}_{a,h}(y)
=
p^{\mathrm{geo}}_{a,h,\theta_0}(y).
\]
Adding and subtracting
\(\Pn\varphi_{\widehat\theta}(Z;y)\),
\(\Pn\varphi_h^{\mathrm{geo}}(Z;y,\eta_0)\), and
\(\Pb\varphi_{\widehat\theta}(Z;y)\), we obtain
\[
\widehat p^{\mathrm{geo}}_{a,h}(y)
-
p^{\mathrm{geo}}_{a,h}(y)
=
G_n(y)
+
B_{\theta,h}(y)
+
R_{\pi,\mu}^{\widehat\theta}(y)
+
e_n^{\widehat\theta}(y),
\]
where
\[
G_n(y)
=
(\Pn-\Pb)\{\varphi_h^{\mathrm{geo}}(Z;y,\eta_0)\},
\]
\[
B_{\theta,h}(y)
=
\Pb\varphi_{\widehat\theta}(Z;y)
-
\Pb\varphi_h^{\mathrm{geo}}(Z;y,\eta_0),
\]
\[
R_{\pi,\mu}^{\widehat\theta}(y)
=
\Pb\{\widehat\varphi_{\widehat\theta}(Z;y)-\varphi_{\widehat\theta}(Z;y)\},
\]
and
\[
e_n^{\widehat\theta}(y)
=
(\Pn-\Pb)
\{\widehat\varphi_{\widehat\theta}(Z;y)-\varphi_h^{\mathrm{geo}}(Z;y,\eta_0)\}.
\]
Equivalently, by writing
\[
e_{n,1}^{\widehat\theta}(y)
=
(\Pn-\Pb)
\{\varphi_{\widehat\theta}(Z;y)-\varphi_h^{\mathrm{geo}}(Z;y,\eta_0)\},
\]
\[
e_{n,2}^{\widehat\theta}(y)
=
(\Pn-\Pb)
\{\widehat\varphi_{\widehat\theta}(Z;y)-\varphi_{\widehat\theta}(Z;y)\},
\]
we decompose the empirical process term as
\[
e_n^{\widehat\theta}(y)
=
e_{n,1}^{\widehat\theta}(y)
+
e_{n,2}^{\widehat\theta}(y).
\]

We next show that \(e_n^{\widehat\theta}\) is negligible in integrated mean
square. Applying the foldwise conditional-variance argument of the
fixed-geometry case to \(e_{n,2}^{\widehat\theta}\) and aggregating over
the fixed number of folds, it suffices to control the corresponding
integrated \(L_2(\Pb)\) norm. On the event \(\inf_x\widehat\pi_a(x)\ge\pi_{\min}/2\), the same algebra as in
the fixed-geometry case gives
\[
\begin{aligned}
&\int_{\mathcal Y_0}
\|\widehat\varphi_{\widehat\theta}(\cdot;y)
-
\varphi_{\widehat\theta}(\cdot;y)\|_{L_2(\Pb)}^2dy \\
&\quad\lesssim
\|\widehat\pi_a-\pi_a\|_{L_2(\Pb)}^2
\sup_x
\E\!\left[
\int_{\mathcal Y_0}\kappa_{h,\widehat\theta}(y;Y)^2dy
\,\middle|\,
X=x,A=a,\mathcal F_{-k}
\right] \\
&\quad+
\int_{\mathcal Y_0}
\|\widehat\mu_{a,h,\widehat\theta}(\cdot;y)
-
\mu_{a,h,\widehat\theta}(\cdot;y)\|_{L_2(\Pb)}^2dy
=
o_\Pb(1).
\end{aligned}
\]
Therefore, by the same conditional Markov argument,
\[
\int_{\mathcal Y_0}\{e_{n,2}^{\widehat\theta}(y)\}^2dy
=
o_\Pb(n^{-1}).
\]

Similarly, applying the same foldwise argument to
\(e_{n,1}^{\widehat\theta}\), it suffices to control the corresponding
integrated \(L_2(\Pb)\) norm. Letting
\[
\Delta_\kappa(y;Y)
=
\kappa_{h,\widehat\theta}(y;Y)
-
\kappa_{h,\theta_0}(y;Y),
\quad
\Delta_\mu(X;y)
=
\mu_{a,h,\widehat\theta}(X;y)
-
\mu_{a,h,\theta_0}(X;y),
\]
we get
\[
\varphi_{\widehat\theta}(Z;y)
-
\varphi_h^{\mathrm{geo}}(Z;y,\eta_0)
=
\frac{\mathbbm 1\{A=a\}}{\pi_a(X)}
\{\Delta_\kappa(y;Y)-\Delta_\mu(X;y)\}
+
\Delta_\mu(X;y).
\]
Since
\[
\Delta_\mu(X;y)
=
\E\{\Delta_\kappa(y;Y)\mid X,A=a,\mathcal F_{-k}\},
\]
Jensen's inequality gives
\[
\Delta_\mu(X;y)^2
\le
\E\{\Delta_\kappa(y;Y)^2\mid X,A=a,\mathcal F_{-k}\}.
\]
Therefore, by positivity,
\[
\int_{\mathcal Y_0}
\left\|
\varphi_{\widehat\theta}(\cdot;y)
-
\varphi_h^{\mathrm{geo}}(\cdot;y,\eta_0)
\right\|_{L_2(\Pb)}^2dy
\lesssim
\int_{\mathcal Y_0}
\left\|
\kappa_{h,\widehat\theta}(y;Y^a)
-
\kappa_{h,\theta_0}(y;Y^a)
\right\|_{L_2(\Pb_a)}^2dy
=
o_\Pb(1).
\]
Applying the same conditional Markov argument gives
\[
\int_{\mathcal Y_0}\{e_{n,1}^{\widehat\theta}(y)\}^2dy
=
o_\Pb(n^{-1}).
\]
Combining the two bounds gives
\[
\int_{\mathcal Y_0}\{e_n^{\widehat\theta}(y)\}^2dy
=
o_\Pb(n^{-1}).
\]
Under the Donsker alternative in Assumption~\ref{assumption:sample-splitting},
the same conclusion follows from an integrated stochastic equicontinuity
argument.

Consequently,
\[
\begin{aligned}
\int_{\mathcal Y_0}
\{\widehat p^{\mathrm{geo}}_{a,h}(y)-p^{\mathrm{geo}}_{a,h}(y)\}^2dy
&\lesssim
\int_{\mathcal Y_0}G_n(y)^2dy
+
\int_{\mathcal Y_0}B_{\theta,h}(y)^2dy \\
&\quad+
\int_{\mathcal Y_0}\{R_{\pi,\mu}^{\widehat\theta}(y)\}^2dy
+
o_\Pb(n^{-1}).
\end{aligned}
\]
Using the stochastic bound above and the finite measure of \(\mathcal Y_0\),
\[
\int_{\mathcal Y_0}
\{\widehat p^{\mathrm{geo}}_{a,h}(y)-p^{\mathrm{geo}}_{a,h}(y)\}^2dy
=
O_\Pb\!\left(
\frac{H_h(\mathcal Y_0)}{n}
+
\operatorname*{ess\,sup}_{y\in\mathcal Y_0}
\{R_{\pi,\mu}^{\widehat\theta}(y)\}^2
+
\int_{\mathcal Y_0}B_{\theta,h}(y)^2dy
\right)
+
o_\Pb(n^{-1}).
\]
This completes the proof.
\end{proof}

\subsubsection{Proof of Proposition \ref{prop:manifold-concentration}}

\begin{proof}
For fixed \(u\in M\), write
\[
\Sigma_h(u)
=
h^2P_T(u)+\sigma_0^2P_N(u).
\]
Since \(P_T(u)\) and \(P_N(u)\) are orthogonal projections onto complementary
subspaces of dimensions \(m\) and \(d-m\), respectively, the eigenvalues of
\(\Sigma_h(u)\) are \(h^2\) with multiplicity \(m\) and \(\sigma_0^2\) with
multiplicity \(d-m\). Hence
\[
\det\{\Sigma_h(u)\}
=
h^{2m}\sigma_0^{2(d-m)}.
\]

For a \(d\)-dimensional Gaussian density
\(\phi_d(\cdot;\Sigma)\),
\[
\int_{\mathbb R^d}\phi_d(x;\Sigma)^2dx
=
(4\pi)^{-d/2}\det(\Sigma)^{-1/2}.
\]
Therefore,
\[
\begin{aligned}
\int_{\mathbb R^d}\kappa_h^{\rm tan}(y;u)^2dy
&=
\int_{\mathbb R^d}
\phi_d\{y-u;\Sigma_h(u)\}^2dy \\
&=
(4\pi)^{-d/2}
\det\{\Sigma_h(u)\}^{-1/2} \\
&=
(4\pi)^{-d/2}
\sigma_0^{-(d-m)}h^{-m}.
\end{aligned}
\]
The right-hand side does not depend on \(u\). Thus
\[
H_h^{\rm tan}(\mathbb R^d)
=
\E\left[
\int_{\mathbb R^d}
\kappa_h^{\rm tan}(y;Y^a)^2dy
\right]
=
C_{d,m,\sigma_0}h^{-m},
\]
with
\[
C_{d,m,\sigma_0}
=
(4\pi)^{-d/2}\sigma_0^{-(d-m)}.
\]
Finally, since \(\kappa_h^{\rm tan}(y;u)^2\ge0\),
\[
H_h^{\rm tan}(\mathcal Y_0)
=
\E\left[
\int_{\mathcal Y_0}
\kappa_h^{\rm tan}(y;Y^a)^2dy
\right]
\le
H_h^{\rm tan}(\mathbb R^d)
=
C_{d,m,\sigma_0}h^{-m}.
\]
\end{proof}

\subsubsection{Proof of Theorem~\ref{thm:covariance-concentration}}

\begin{proof}
By the squared \(L_2\) concentration bound in Lemma~\ref{lem:kappa-basic}, or
directly from Assumption~\ref{assump:subgauss-kappa}, there exists a constant
\(C<\infty\) such that, for \(\Pb_a\)-almost every \(Y^a\),
\[
\int_{\mathbb R^d}
\kappa_{h,\theta_0}(u;Y^a)^2du
\le
C\det\{G_h(Y^a)\}^{1/2}.
\]
Indeed, Assumption~\ref{assump:subgauss-kappa} gives
\[
\kappa_{h,\theta_0}(u;Y^a)^2
\le
C_0^2\det\{G_h(Y^a)\}
\exp\{-2c_0\|u-m_h(Y^a)\|_{G_h(Y^a)}^2\},
\]
and the change of variables
\(v=G_h(Y^a)^{1/2}\{u-m_h(Y^a)\}\) yields the displayed bound.

Therefore, for any measurable \(\mathcal Y_0\subseteq\mathbb R^d\),
\[
\begin{aligned}
H_h(\mathcal Y_0)
&=
\E\left[
\int_{\mathcal Y_0}
\kappa_{h,\theta_0}(u;Y^a)^2du
\right] \\
&\le
\E\left[
\int_{\mathbb R^d}
\kappa_{h,\theta_0}(u;Y^a)^2du
\right] \\
&\le
C\,\E\left[\det\{G_h(Y^a)\}^{1/2}\right].
\end{aligned}
\]
If
\[
\E[\det\{G_h(Y^a)\}^{1/2}]
\le
C_rh^{-r(h)},
\quad 0<h\le h_r,
\]
then
\[
H_h(\mathcal Y_0)
\le
H_h(\mathbb R^d)
\lesssim
h^{-r(h)}.
\]
This proves the first claim.

We next verify the eigenvalue consequence. Suppose that, for some
\(d_\star\ll d\), the covariance
\(\Sigma_h(Y^a)=G_h(Y^a)^{-1}\) has \(d_\star\) eigenvalues of order \(h^2\)
and \(d-d_\star\) eigenvalues of order one, with \(\Pb_a\)-probability one,
uniformly over sufficiently small \(h\). Then there exist constants
\(0<c<C<\infty\) and \(h_0>0\) such that, for all \(0<h\le h_0\), with
\(\Pb_a\)-probability one,
\[
c h^2
\le
\lambda_j\{\Sigma_h(Y^a)\}
\le
C h^2
\quad (j\le d_\star),
\quad
c
\le
\lambda_j\{\Sigma_h(Y^a)\}
\le
C
\quad (j>d_\star).
\]
It follows that
\[
\det\{\Sigma_h(Y^a)\}^{1/2}
\asymp
h^{d_\star},
\]
and hence, since \(G_h(Y^a)=\Sigma_h(Y^a)^{-1}\),
\[
\det\{G_h(Y^a)\}^{1/2}
=
\det\{\Sigma_h(Y^a)\}^{-1/2}
\asymp
h^{-d_\star}
\]
\(\Pb_a\)-almost surely, uniformly over sufficiently small \(h\). Thus
\[
\E[\det\{G_h(Y^a)\}^{1/2}]
\lesssim
h^{-d_\star}.
\]
Applying the first part with \(r(h)=d_\star\) gives
\[
H_h(\mathcal Y_0)
\le
H_h(\mathbb R^d)
\lesssim
h^{-d_\star}.
\]

Finally, suppose additionally that
\[
\int_{\mathbb R^d}
\kappa_{h,\theta_0}(u;Y^a)^2du
\gtrsim
\det\{G_h(Y^a)\}^{1/2}
\]
with \(\Pb_a\)-probability one uniformly over sufficiently small \(h\). Then,
using the lower eigenvalue consequence above,
\[
\begin{aligned}
H_h(\mathbb R^d)
&=
\E\left[
\int_{\mathbb R^d}
\kappa_{h,\theta_0}(u;Y^a)^2du
\right] \\
&\gtrsim
\E\left[
\det\{G_h(Y^a)\}^{1/2}
\right] \\
&\gtrsim
h^{-d_\star}.
\end{aligned}
\]
Together with the upper bound \(H_h(\mathbb R^d)\lesssim h^{-d_\star}\), this
gives
\[
H_h(\mathbb R^d)\asymp h^{-d_\star}.
\]
This completes the proof.
\end{proof}

\subsection{Section \ref{sec:dss}}

\subsubsection{Proof of Theorem \ref{thm:dss-expansion}}

\begin{proof}
The logic parallels the proof of Theorem~\ref{thm:onestep-expansion}.
Fix \(h>0\) and \(y\in\mathbb R^d\), and write
\(\eta^s_0=\eta^s_{\theta_0}\). Throughout the proof, we use the foldwise cross-fitting convention stated in the beginning.

We introduce several shorthands for notational simplicity. Define
\[
p_0=P_{a,h,\theta_0}(y),
\quad
g_0=G_{a,h,\theta_0}(y),
\quad
s_0=s^{\mathrm{geo}}_{a,h}(y)=g_0/p_0,
\]
and the corresponding learned-geometry quantities
\[
p_{\widehat\theta}=P_{a,h,\widehat\theta}(y),
\quad
g_{\widehat\theta}=G_{a,h,\widehat\theta}(y),
\quad
s_{\widehat\theta}=s_{a,h,\widehat\theta}(y)
=
g_{\widehat\theta}/p_{\widehat\theta}.
\]
Thus \(B_{s,\theta,h}(y)=s_{\widehat\theta}-s_0\). Also let
\[
\varphi_{P,0}(Z)=\varphi_P(Z;y,\eta^s_0),
\quad
\varphi_{G,0}(Z)=\varphi_G(Z;y,\eta^s_0),
\]
and
\[
\varphi_{P,\widehat\theta}(Z)
=
\varphi_P(Z;y,\eta^s_{\widehat\theta}),
\quad
\varphi_{G,\widehat\theta}(Z)
=
\varphi_G(Z;y,\eta^s_{\widehat\theta}),
\]
where $\eta^s_{\widehat\theta}
=
\{\pi_a,\mu_{a,h,\widehat\theta},\nu_{a,h,\widehat\theta},
\kappa_{h,\widehat\theta},\dot\kappa_{h,\widehat\theta}\}.$

Similarly, write
\[
\widehat\varphi_{P,\widehat\theta}(Z)
=
\varphi_P(Z;y,\widehat\eta^s_{\widehat\theta}),
\quad
\widehat\varphi_{G,\widehat\theta}(Z)
=
\varphi_G(Z;y,\widehat\eta^s_{\widehat\theta}).
\]
For brevity, set
\[
\widehat p=\widehat P_{a,h}(y)
=
\Pn\widehat\varphi_{P,\widehat\theta},
\quad
\widehat g=\widehat G_{a,h}(y)
=
\Pn\widehat\varphi_{G,\widehat\theta}.
\]
Then
\[
\widehat s^{\mathrm{geo}}_{a,h}(y)=\widehat g/\widehat p,
\]
and, by the definitions of the localized regressions,
\[
\Pb\{\varphi_{P,\widehat\theta}\}=p_{\widehat\theta},
\quad
\Pb\{\varphi_{G,\widehat\theta}\}=g_{\widehat\theta}.
\]

We first decompose the density and gradient numerator components as
\[
\widehat p-p_{\widehat\theta}
=
(\Pn-\Pb)\varphi_{P,\widehat\theta}
+
\Pb\{\widehat\varphi_{P,\widehat\theta}-\varphi_{P,\widehat\theta}\}
+
(\Pn-\Pb)\{\widehat\varphi_{P,\widehat\theta}-\varphi_{P,\widehat\theta}\},
\]
and
\[
\widehat g-g_{\widehat\theta}
=
(\Pn-\Pb)\varphi_{G,\widehat\theta}
+
\Pb\{\widehat\varphi_{G,\widehat\theta}-\varphi_{G,\widehat\theta}\}
+
(\Pn-\Pb)\{\widehat\varphi_{G,\widehat\theta}-\varphi_{G,\widehat\theta}\},
\]
respectively.

Since, within fold \(k\),
$\E\{\kappa_{h,\widehat\theta}(y;Y)\mid X,A=a,\mathcal F_{-k}\}
=
\mu_{a,h,\widehat\theta}(X;y)$, iterated expectation gives
\[
\begin{aligned}
\Pb\{\widehat\varphi_{P,\widehat\theta}
-
\varphi_{P,\widehat\theta}\}
&=
\Pb\left[
\frac{\pi_a(X)}{\widehat\pi_a(X)}
\{\mu_{a,h,\widehat\theta}(X;y)
-
\widehat\mu_{a,h,\widehat\theta}(X;y)\}
+
\widehat\mu_{a,h,\widehat\theta}(X;y)
-
\mu_{a,h,\widehat\theta}(X;y)
\right] \\
&=
\Pb\left[
\left\{\frac{\pi_a(X)}{\widehat\pi_a(X)}-1\right\}
\{\mu_{a,h,\widehat\theta}(X;y)
-
\widehat\mu_{a,h,\widehat\theta}(X;y)\}
\right].
\end{aligned}
\]
On the event where \(\widehat\pi_a\) is bounded below,
\[
\left|
\frac{\pi_a(X)}{\widehat\pi_a(X)}-1
\right|
\lesssim
|\widehat\pi_a(X)-\pi_a(X)|.
\]
Therefore, by the Cauchy--Schwarz inequality,
\[
\left|
\Pb\{\widehat\varphi_{P,\widehat\theta}
-
\varphi_{P,\widehat\theta}\}
\right|
\lesssim
\delta_\pi\delta_\mu^{\widehat\theta}(y).
\]
Similarly, since
\[
\E\{\dot\kappa_{h,\widehat\theta}(y;Y)
\mid X,A=a,\mathcal F_{-k}\}
=
\nu_{a,h,\widehat\theta}(X;y),
\]
we have
\[
\left\|
\Pb\{\widehat\varphi_{G,\widehat\theta}
-
\varphi_{G,\widehat\theta}\}
\right\|_2
\lesssim
\delta_\pi\delta_\nu^{\widehat\theta}(y).
\]
Define
\[
R_P^{\widehat\theta}(y)
=
\Pb\{\widehat\varphi_{P,\widehat\theta}
-
\varphi_{P,\widehat\theta}\},
\quad
R_G^{\widehat\theta}(y)
=
\Pb\{\widehat\varphi_{G,\widehat\theta}
-
\varphi_{G,\widehat\theta}\}.
\]
Then
\[
|R_P^{\widehat\theta}(y)|
\lesssim
\delta_\pi\delta_\mu^{\widehat\theta}(y),
\quad
\|R_G^{\widehat\theta}(y)\|_2
\lesssim
\delta_\pi\delta_\nu^{\widehat\theta}(y).
\]

We next control the empirical process terms involving estimated nuisance
functions. Positivity, the nuisance consistency conditions, and the conditional
moment bounds in Assumption~\ref{assumption:score-regularity} imply
\[
\begin{aligned}
&\|\widehat\varphi_{P,\widehat\theta}
-
\varphi_{P,\widehat\theta}\|_{L_2(\Pb)}
+
\|\widehat\varphi_{G,\widehat\theta}
-
\varphi_{G,\widehat\theta}\|_{L_2(\Pb)}
=
o_\Pb(1).
\end{aligned}
\]
By the foldwise cross-fitting convention and the sample-splitting
argument, the preceding \(L_2(\Pb)\) convergence implies
\[
(\Pn-\Pb)(\widehat\varphi_{P,\widehat\theta}
-
\varphi_{P,\widehat\theta})
=
o_\Pb(n^{-1/2}),
\]
and
\[
(\Pn-\Pb)(\widehat\varphi_{G,\widehat\theta}
-
\varphi_{G,\widehat\theta})
=
o_\Pb(n^{-1/2}).
\]
Under the Donsker alternative in
Assumption~\ref{assumption:sample-splitting}, the same conclusions follow from
stochastic equicontinuity
\citep[Chapter~2]{van1996weak}.

We now replace the learned-geometry empirical terms by their
population-geometry counterparts. By definition,
\[
\begin{aligned}
&\mu_{a,h,\widehat\theta}(X;y)-\mu_{a,h,\theta_0}(X;y) \\
&\quad=
\E\{
\kappa_{h,\widehat\theta}(y;Y)
-
\kappa_{h,\theta_0}(y;Y)
\mid X,A=a,\mathcal F_{-k}
\},
\end{aligned}
\]
and
\[
\begin{aligned}
&\nu_{a,h,\widehat\theta}(X;y)-\nu_{a,h,\theta_0}(X;y) \\
&\quad=
\E\{
\dot\kappa_{h,\widehat\theta}(y;Y)
-
\dot\kappa_{h,\theta_0}(y;Y)
\mid X,A=a,\mathcal F_{-k}
\}.
\end{aligned}
\]
Jensen's inequality and positivity therefore give
\[
\begin{aligned}
&\|\varphi_{P,\widehat\theta}-\varphi_{P,0}\|_{L_2(\Pb)}
+
\|\varphi_{G,\widehat\theta}-\varphi_{G,0}\|_{L_2(\Pb)} \\
&\quad\lesssim
\|\kappa_{h,\widehat\theta}(y;Y^a)
-
\kappa_{h,\theta_0}(y;Y^a)\|_{L_2(\Pb_a)} \\
&\quad\quad+
\|\dot\kappa_{h,\widehat\theta}(y;Y^a)
-
\dot\kappa_{h,\theta_0}(y;Y^a)\|_{L_2(\Pb_a)} \\
&\quad=
\delta_{\kappa,s}^{\widehat\theta}(y).
\end{aligned}
\]
Since
\(\delta_{\kappa,s}^{\widehat\theta}(y)=o_\Pb(1)\), the same
sample-splitting argument gives
\[
(\Pn-\Pb)(\varphi_{P,\widehat\theta}-\varphi_{P,0})
=
o_\Pb(n^{-1/2}),
\quad
(\Pn-\Pb)(\varphi_{G,\widehat\theta}-\varphi_{G,0})
=
o_\Pb(n^{-1/2}).
\]

By construction,
\[
\phi_P(Z;y,\eta^s_0)=\varphi_{P,0}(Z)-p_0,
\quad
\phi_G(Z;y,\eta^s_0)=\varphi_{G,0}(Z)-g_0.
\]
Consequently,
\[
\widehat p-p_{\widehat\theta}
=
(\Pn-\Pb)\{\phi_P(Z;y,\eta^s_0)\}
+
R_P^{\widehat\theta}(y)
+
o_\Pb(n^{-1/2}),
\]
and
\[
\widehat g-g_{\widehat\theta}
=
(\Pn-\Pb)\{\phi_G(Z;y,\eta^s_0)\}
+
R_G^{\widehat\theta}(y)
+
o_\Pb(n^{-1/2}).
\]

The geometry consistency condition also gives
\[
|p_{\widehat\theta}-p_0|
\le
\|\kappa_{h,\widehat\theta}(y;Y^a)
-
\kappa_{h,\theta_0}(y;Y^a)\|_{L_2(\Pb_a)}
=
o_\Pb(1),
\]
and
\[
\|g_{\widehat\theta}-g_0\|_2
\le
\|\dot\kappa_{h,\widehat\theta}(y;Y^a)
-
\dot\kappa_{h,\theta_0}(y;Y^a)\|_{L_2(\Pb_a)}
=
o_\Pb(1).
\]
Since \(p_0\) and \(p_{\widehat\theta}\) are bounded away from zero,
\[
s_{\widehat\theta}-s_0
=
\frac{g_{\widehat\theta}-g_0}{p_{\widehat\theta}}
-
s_0
\frac{p_{\widehat\theta}-p_0}{p_{\widehat\theta}}
=
o_\Pb(1).
\]

We now pass to the ratio. Let
\[
\Delta_p=\widehat p-p_{\widehat\theta},
\quad
\Delta_g=\widehat g-g_{\widehat\theta},
\]
and set
\[
r_n(y)
=
\delta_\pi
\{\delta_\mu^{\widehat\theta}(y)
+
\delta_\nu^{\widehat\theta}(y)\}.
\]
On the event where \(p_{\widehat\theta}\) and \(\widehat p\) are bounded below,
\[
\widehat s^{\mathrm{geo}}_{a,h}(y)-s_{\widehat\theta}
=
\frac{\Delta_g-s_{\widehat\theta}\Delta_p}
{p_{\widehat\theta}}
+
\mathcal Q_n(y),
\]
where
\[
\mathcal Q_n(y)
=
\left(
\frac{1}{\widehat p}
-
\frac{1}{p_{\widehat\theta}}
\right)
\{\Delta_g-s_{\widehat\theta}\Delta_p\}.
\]
Square integrability of the oracle influence functions and the component
expansions above imply
\[
|\Delta_p|
=
O_\Pb(n^{-1/2}+r_n(y)),
\quad
\|\Delta_g\|_2
=
O_\Pb(n^{-1/2}+r_n(y)).
\]
Therefore,
\[
\begin{aligned}
\|\mathcal Q_n(y)\|_2
&\lesssim
|\Delta_p|
\{\|\Delta_g\|_2
+
\|s_{\widehat\theta}\|_2|\Delta_p|\} \\
&=
o_\Pb(n^{-1/2})
+
O_\Pb\{r_n(y)^2\}.
\end{aligned}
\]
Since \(r_n(y)=o_\Pb(1)\), the quadratic term is absorbed into the final
\(O_\Pb\{r_n(y)\}\) remainder.

Substituting the component expansions into the linear part of the ratio gives
\[
\begin{aligned}
\widehat s^{\mathrm{geo}}_{a,h}(y)-s_{\widehat\theta}
&=
\frac{1}{p_{\widehat\theta}}
(\Pn-\Pb)
\left\{
\phi_G(Z;y,\eta^s_0)
-
s_{\widehat\theta}\phi_P(Z;y,\eta^s_0)
\right\} \\
&\quad+
\frac{
R_G^{\widehat\theta}(y)
-
s_{\widehat\theta}R_P^{\widehat\theta}(y)
}
{p_{\widehat\theta}}
+
\mathcal Q_n(y)
+
o_\Pb(n^{-1/2}).
\end{aligned}
\]
By square integrability,
\[
(\Pn-\Pb)\phi_P(Z;y,\eta^s_0)=O_\Pb(n^{-1/2}),
\quad
(\Pn-\Pb)\phi_G(Z;y,\eta^s_0)=O_\Pb(n^{-1/2}).
\]
Since
\(p_{\widehat\theta}-p_0=o_\Pb(1)\) and
\(s_{\widehat\theta}-s_0=o_\Pb(1)\), replacing
\(p_{\widehat\theta}\) by \(p_0\) and \(s_{\widehat\theta}\) by \(s_0\) in the
empirical term changes it by \(o_\Pb(n^{-1/2})\). Hence
\[
\frac{1}{p_{\widehat\theta}}
(\Pn-\Pb)
\left\{
\phi_G(Z;y,\eta^s_0)
-
s_{\widehat\theta}\phi_P(Z;y,\eta^s_0)
\right\}
=
(\Pn-\Pb)\{\phi_s(Z;y,\eta^s_0)\}
+
o_\Pb(n^{-1/2}).
\]

Collecting the causal product remainders and the
\(O_\Pb\{r_n(y)^2\}\) part of the quadratic ratio remainder, define
\(R_{s,\pi,\mu,\nu}^{\widehat\theta}(y)\) so that
\[
\begin{aligned}
R_{s,\pi,\mu,\nu}^{\widehat\theta}(y)
&=
\frac{
R_G^{\widehat\theta}(y)
-
s_{\widehat\theta}R_P^{\widehat\theta}(y)
}
{p_{\widehat\theta}}
+
O_\Pb\{r_n(y)^2\}.
\end{aligned}
\]
By Assumption~\ref{assumption:score-regularity},
\(p_{\widehat\theta}\) is bounded away from zero and
\(\|s_{\widehat\theta}\|_2=O_\Pb(1)\). Therefore,
\[
\begin{aligned}
\|R_{s,\pi,\mu,\nu}^{\widehat\theta}(y)\|_2
&=
O_\Pb\left[
\delta_\pi
\{\delta_\mu^{\widehat\theta}(y)
+
\delta_\nu^{\widehat\theta}(y)\}
\right].
\end{aligned}
\]

Finally,
\[
\widehat s^{\mathrm{geo}}_{a,h}(y)
-
s^{\mathrm{geo}}_{a,h}(y)
=
\{\widehat s^{\mathrm{geo}}_{a,h}(y)-s_{\widehat\theta}\}
+
\{s_{\widehat\theta}-s_0\},
\]
and \(s_{\widehat\theta}-s_0=B_{s,\theta,h}(y)\). Therefore,
\[
\widehat s^{\mathrm{geo}}_{a,h}(y)
-
s^{\mathrm{geo}}_{a,h}(y)
=
(\Pn-\Pb)\{\phi_s(Z;y,\eta^s_0)\}
+
B_{s,\theta,h}(y)
+
R_{s,\pi,\mu,\nu}^{\widehat\theta}(y)
+
o_\Pb(n^{-1/2}),
\]
which completes the proof.
\end{proof}

\subsubsection{Proof of Theorem \ref{thm:dss-risk-global}}

\begin{proof}
Throughout the proof, constants may depend on \(\pi_{\min}\),
\(p'_{\min}\), \(\mathrm{Vol}(\mathcal Y_0)\), and the score regularity
constants, but not on \(n\). Write
\[
P_{\widehat\theta}(y)=P_{a,h,\widehat\theta}(y),\quad
G_{\widehat\theta}(y)=G_{a,h,\widehat\theta}(y),\quad
s_{\widehat\theta}(y)=s_{a,h,\widehat\theta}(y),\quad
s_0(y)=s^{\mathrm{geo}}_{a,h}(y),
\]
and let \(\eta_0^s=\eta_{\theta_0}^s\). Also write
\[
\varphi_{P,0}(Z;y)=\varphi_P(Z;y,\eta_0^s),\quad
\varphi_{G,0}(Z;y)=\varphi_G(Z;y,\eta_0^s),
\]
\[
\varphi_{P,\widehat\theta}(Z;y)=\varphi_P(Z;y,\eta^s_{\widehat\theta}),\quad
\varphi_{G,\widehat\theta}(Z;y)=\varphi_G(Z;y,\eta^s_{\widehat\theta}),
\]
and
\[
\widehat\varphi_{P,\widehat\theta}(Z;y)=\varphi_P(Z;y,\widehat\eta^s_{\widehat\theta}),\quad
\widehat\varphi_{G,\widehat\theta}(Z;y)=\varphi_G(Z;y,\widehat\eta^s_{\widehat\theta}).
\]

Define
\[
R_P^{\widehat\theta}(y)=\Pb\{\widehat\varphi_{P,\widehat\theta}(Z;y)-\varphi_{P,\widehat\theta}(Z;y)\},\quad
R_G^{\widehat\theta}(y)=\Pb\{\widehat\varphi_{G,\widehat\theta}(Z;y)-\varphi_{G,\widehat\theta}(Z;y)\}.
\]
As in the proof of Theorem~\ref{thm:dss-expansion}, iterated expectation gives
\[
R_P^{\widehat\theta}(y)=
\Pb\left[
\left\{\frac{\pi_a(X)}{\widehat\pi_a(X)}-1\right\}
\{\mu_{a,h,\widehat\theta}(X;y)-\widehat\mu_{a,h,\widehat\theta}(X;y)\}
\right],
\]
with the analogous vector-valued expression for
\(R_G^{\widehat\theta}(y)\). Hence, by the propensity lower bound and the
Cauchy--Schwarz inequality,
\[
|R_P^{\widehat\theta}(y)|\lesssim
\delta_\pi\delta_\mu^{\widehat\theta}(y),\quad
\|R_G^{\widehat\theta}(y)\|_2\lesssim
\delta_\pi\delta_\nu^{\widehat\theta}(y).
\]

Adding and subtracting the learned- and population-geometry estimating
functions gives
\[
\widehat P_{a,h}(y)-P_{\widehat\theta}(y)
=
(\Pn-\Pb)\{\phi_P(Z;y,\eta_0^s)\}
+
R_P^{\widehat\theta}(y)+e_P(y),
\]
\[
\widehat G_{a,h}(y)-G_{\widehat\theta}(y)
=
(\Pn-\Pb)\{\phi_G(Z;y,\eta_0^s)\}
+
R_G^{\widehat\theta}(y)+e_G(y),
\]
where
\[
e_P(y)=
(\Pn-\Pb)\{\widehat\varphi_{P,\widehat\theta}(\cdot;y)-\varphi_{P,\widehat\theta}(\cdot;y)\}
+
(\Pn-\Pb)\{\varphi_{P,\widehat\theta}(\cdot;y)-\varphi_{P,0}(\cdot;y)\},
\]
\[
e_G(y)=
(\Pn-\Pb)\{\widehat\varphi_{G,\widehat\theta}(\cdot;y)-\varphi_{G,\widehat\theta}(\cdot;y)\}
+
(\Pn-\Pb)\{\varphi_{G,\widehat\theta}(\cdot;y)-\varphi_{G,0}(\cdot;y)\}.
\]
Here
\[
(\Pn-\Pb)\varphi_{P,0}=(\Pn-\Pb)\{\phi_P(Z;y,\eta_0^s)\},\quad
(\Pn-\Pb)\varphi_{G,0}=(\Pn-\Pb)\{\phi_G(Z;y,\eta_0^s)\},
\]
because centering only subtracts the corresponding population means.

The integrated nuisance conditions, propensity condition, positivity, and
Jensen's inequality imply
\[
\begin{aligned}
\int_{\mathcal Y_0}\big[
&\|\widehat\varphi_{P,\widehat\theta}-\varphi_{P,\widehat\theta}\|_{L_2(\Pb)}^2
+
\|\widehat\varphi_{G,\widehat\theta}-\varphi_{G,\widehat\theta}\|_{L_2(\Pb)}^2\\
&+
\|\varphi_{P,\widehat\theta}-\varphi_{P,0}\|_{L_2(\Pb)}^2
+
\|\varphi_{G,\widehat\theta}-\varphi_{G,0}\|_{L_2(\Pb)}^2
\big]dy
=
o_\Pb(1).
\end{aligned}
\]
Indeed, the first two terms are controlled by the regression consistency and
propensity-weighted moment assumptions, while the last two are bounded by the
integrated kernel and kernel-gradient perturbations. Applying the
conditional-variance argument separately within each fold and then
aggregating over the fixed number of folds gives
\[
\int_{\mathcal Y_0}e_P(y)^2dy=o_\Pb(n^{-1}),\quad
\int_{\mathcal Y_0}\|e_G(y)\|_2^2dy=o_\Pb(n^{-1}).
\]
Under the Donsker alternative in
Assumption~\ref{assumption:sample-splitting}, the same conclusion follows from
stochastic equicontinuity
\citep[Chapter~2]{van1996weak}.

On the event where \(\widehat P_{a,h}\) is bounded away from zero,
\[
\widehat s^{\mathrm{geo}}_{a,h}(y)-s_{\widehat\theta}(y)
=
\frac{\{\widehat G_{a,h}(y)-G_{\widehat\theta}(y)\}
-s_{\widehat\theta}(y)\{\widehat P_{a,h}(y)-P_{\widehat\theta}(y)\}}
{\widehat P_{a,h}(y)}.
\]
Substitution of the component expansions yields
\[
\widehat s^{\mathrm{geo}}_{a,h}(y)-s_{\widehat\theta}(y)
=
E_n^s(y)+R_{s,\pi,\mu,\nu}^{\widehat\theta}(y)+e_s(y),
\]
where
\[
E_n^s(y)=
\frac{(\Pn-\Pb)\{\phi_G(Z;y,\eta_0^s)\}
-s_{\widehat\theta}(y)(\Pn-\Pb)\{\phi_P(Z;y,\eta_0^s)\}}
{\widehat P_{a,h}(y)},
\]
\[
R_{s,\pi,\mu,\nu}^{\widehat\theta}(y)=
\frac{R_G^{\widehat\theta}(y)-s_{\widehat\theta}(y)R_P^{\widehat\theta}(y)}
{\widehat P_{a,h}(y)},\quad
e_s(y)=
\frac{e_G(y)-s_{\widehat\theta}(y)e_P(y)}
{\widehat P_{a,h}(y)}.
\]
By Assumption~\ref{assumption:score-regularity},
\(\widehat P_{a,h}\) is uniformly bounded below with probability tending to
one and
\(\operatorname*{ess\,sup}_{y\in\mathcal Y_0}
\|s_{\widehat\theta}(y)\|_2=O_\Pb(1)\). Consequently,
\[
\int_{\mathcal Y_0}\|e_s(y)\|_2^2dy=o_\Pb(n^{-1}),
\]
and
\[
\|R_{s,\pi,\mu,\nu}^{\widehat\theta}(y)\|_2
=
O_\Pb\!\left[
\delta_\pi\{\delta_\mu^{\widehat\theta}(y)+\delta_\nu^{\widehat\theta}(y)\}
\right].
\]

The same denominator and score bounds give
\[
\int_{\mathcal Y_0}\|E_n^s(y)\|_2^2dy
=
O_\Pb\!\left[
\int_{\mathcal Y_0}
\|(\Pn-\Pb)\{\phi_G(Z;y,\eta_0^s)\}\|_2^2dy
+
\int_{\mathcal Y_0}
\{(\Pn-\Pb)\{\phi_P(Z;y,\eta_0^s)\}\}^2dy
\right].
\]
Jensen's inequality yields,
\[
\E\{\phi_P(Z;y,\eta_0^s)^2\}\lesssim
\E\{\kappa_{h,\theta_0}(y;Y^a)^2\},\quad
\E\{\|\phi_G(Z;y,\eta_0^s)\|_2^2\}\lesssim
\E\{\|\dot\kappa_{h,\theta_0}(y;Y^a)\|_2^2\}.
\]
Since \(\phi_P\) and \(\phi_G\) are centered and the observations are
independent, for each \(y\in\mathcal Y_0\),
\[
\E\!\left[
\{(\Pn-\Pb)\phi_P(Z;y,\eta_0^s)\}^2
\right]
=
\frac{1}{n}\var\{\phi_P(Z;y,\eta_0^s)\}
\le
\frac{1}{n}\E\{\phi_P(Z;y,\eta_0^s)^2\},
\]
and similarly,
\[
\E\!\left[
\|(\Pn-\Pb)\phi_G(Z;y,\eta_0^s)\|_2^2
\right]
\le
\frac{1}{n}\E\{\|\phi_G(Z;y,\eta_0^s)\|_2^2\}.
\]
Therefore, by Fubini's theorem and the definition of \(H_h^s(\mathcal Y_0)\),
\[
\begin{aligned}
&\E\left[
\int_{\mathcal Y_0}
\|(\Pn-\Pb)\{\phi_G(Z;y,\eta_0^s)\}\|_2^2dy
+
\int_{\mathcal Y_0}
\{(\Pn-\Pb)\{\phi_P(Z;y,\eta_0^s)\}\}^2dy
\right] \\
&\quad\le
\frac{1}{n}
\int_{\mathcal Y_0}
\left[
\E\{\|\phi_G(Z;y,\eta_0^s)\|_2^2\}
+
\E\{\phi_P(Z;y,\eta_0^s)^2\}
\right]dy \\
&\quad\lesssim
\frac{1}{n}
\E\left[
\int_{\mathcal Y_0}
\left\{
\|\dot\kappa_{h,\theta_0}(y;Y^a)\|_2^2
+
\kappa_{h,\theta_0}(y;Y^a)^2
\right\}dy
\right]
=
\frac{H_h^s(\mathcal Y_0)}{n},
\end{aligned}
\]
and therefore
\[
\int_{\mathcal Y_0}\|E_n^s(y)\|_2^2dy
=
O_\Pb\!\left\{\frac{H_h^s(\mathcal Y_0)}{n}\right\}.
\]

Combining these bounds and using the finite Lebesgue measure of
\(\mathcal Y_0\) gives
\[
\int_{\mathcal Y_0}
\|\widehat s^{\mathrm{geo}}_{a,h}(y)-s_{\widehat\theta}(y)\|_2^2dy
=
O_\Pb\!\left[
\frac{H_h^s(\mathcal Y_0)}{n}
+
\operatorname*{ess\,sup}_{y\in\mathcal Y_0}
\|R_{s,\pi,\mu,\nu}^{\widehat\theta}(y)\|_2^2
\right]
+
o_\Pb(n^{-1}).
\]
Finally,
\[
\widehat s^{\mathrm{geo}}_{a,h}(y)-s_0(y)
=
\{\widehat s^{\mathrm{geo}}_{a,h}(y)-s_{\widehat\theta}(y)\}
+
B_{s,\theta,h}(y).
\]
Hence
\[
\begin{aligned}
\int_{\mathcal Y_0}
\|\widehat s^{\mathrm{geo}}_{a,h}(y)-s_0(y)\|_2^2dy
=
O_\Pb\!\left[
\frac{H_h^s(\mathcal Y_0)}{n}
+
\operatorname*{ess\,sup}_{y\in\mathcal Y_0}
\|R_{s,\pi,\mu,\nu}^{\widehat\theta}(y)\|_2^2
+
\int_{\mathcal Y_0}\|B_{s,\theta,h}(y)\|_2^2dy
\right]
+
o_\Pb(n^{-1}).
\end{aligned}
\]
Since \(s_0=s^{\mathrm{geo}}_{a,h}\), the result follows.
\end{proof}

\subsubsection{Proof of Theorem \ref{thm:dss-covariance-concentration}}

\begin{proof}
By Lemma~\ref{lem:kappa-basic}, Assumption~\ref{assump:subgauss-kappa}
implies that, for some \(C<\infty\),
\[
\int_{\mathbb R^d}\kappa_{h,\theta_0}(u;Y)^2du
\le
C\det\{G_h(Y)\}^{1/2}.
\]
We next bound the kernel-gradient term. For fixed \(Y\), write
\(G=G_h(Y)\) and \(m=m_h(Y)\). The assumed gradient envelope gives
\[
\int_{\mathbb R^d}\|\dot\kappa_{h,\theta_0}(u;Y)\|_2^2du
\le
(C'_0)^2\det(G)
\int_{\mathbb R^d}
\|G(u-m)\|_2^2
\exp\{-2c'_0\|u-m\|_G^2\}du.
\]
Under the change of variables \(z=G^{1/2}(u-m)\),
\(du=\det(G)^{-1/2}dz\) and
\(\|G(u-m)\|_2^2=z^\top Gz\). Therefore,
\[
\begin{aligned}
\int_{\mathbb R^d}\|\dot\kappa_{h,\theta_0}(u;Y)\|_2^2du
&\le
(C'_0)^2\det(G)^{1/2}
\int_{\mathbb R^d}
z^\top Gz\,e^{-2c'_0\|z\|_2^2}dz \\
&\lesssim
\det(G)^{1/2}\operatorname{tr}(G),
\end{aligned}
\]
where the last inequality follows from rotational symmetry, since
\[
\int_{\mathbb R^d}zz^\top e^{-2c'_0\|z\|_2^2}dz
=
C_{c'_0,d}I_d.
\]

Combining the kernel and kernel-gradient bounds yields
\[
\int_{\mathbb R^d}
\left\{
\kappa_{h,\theta_0}(u;Y)^2
+
\|\dot\kappa_{h,\theta_0}(u;Y)\|_2^2
\right\}du
\lesssim
\det\{G_h(Y)\}^{1/2}
\{1+\operatorname{tr}G_h(Y)\}.
\]
Since the integrand is nonnegative, for every measurable
\(\mathcal Y_0\subseteq\mathbb R^d\),
\[
\begin{aligned}
H_h^s(\mathcal Y_0)
&\le
H_h^s(\mathbb R^d) \\
&\lesssim
\E\left[
\det\{G_h(Y^a)\}^{1/2}
\{1+\operatorname{tr}G_h(Y^a)\}
\right].
\end{aligned}
\]

For the eigenvalue consequence, the stated assumption means that, with
\(\Pb_a\)-probability one, for all sufficiently small \(h\), the eigenvalues of
\(\Sigma_h(Y^a)\) can be ordered so that
\[
\lambda_j\{\Sigma_h(Y^a)\}\asymp h^2
\quad (j\le d_\star),
\quad
\lambda_j\{\Sigma_h(Y^a)\}\asymp1
\quad (j>d_\star),
\]
with constants uniform in \(h\). Since
\(G_h(Y^a)=\Sigma_h(Y^a)^{-1}\), it follows that
\[
\det\{G_h(Y^a)\}^{1/2}\lesssim h^{-d_\star},
\quad
\operatorname{tr}G_h(Y^a)\lesssim h^{-2}
\]
with \(\Pb_a\)-probability one. Hence, for \(h\le1\),
\[
\det\{G_h(Y^a)\}^{1/2}
\{1+\operatorname{tr}G_h(Y^a)\}
\lesssim
h^{-(d_\star+2)}.
\]
Taking expectations and applying the first part gives
\[
H_h^s(\mathcal Y_0)
\le
H_h^s(\mathbb R^d)
\lesssim
h^{-(d_\star+2)}.
\]
\end{proof}

\subsection{Section~\ref{sec:inference}}

\subsubsection{Proof of Theorem~\ref{thm:dis-band}}

\begin{proof}
When \(\theta^\dagger\) is learned from an auxiliary sample independent of
the analysis sample, all statements below are understood conditionally on
that external geometry-training information, in probability. Let
\[
T_n
=
\sup_{y\in\mathcal Y_0}
\frac{
\sqrt n
\left|
\widehat p^{\mathrm{geo}}_{a,h,\theta^\dagger}(y)
-
p^{\mathrm{geo}}_{a,h,\theta^\dagger}(y)
\right|
}
{\widehat\sigma_{\theta^\dagger}(y)}.
\]
The stated band covers simultaneously if and only if
\(T_n\le\widehat c_{1-\alpha}\).

By the assumed uniform asymptotic linearity,
\[
\sup_{y\in\mathcal Y_0}
\left|
\sqrt n
\{\widehat p^{\mathrm{geo}}_{a,h,\theta^\dagger}(y)
-
p^{\mathrm{geo}}_{a,h,\theta^\dagger}(y)\}
-
\mathbb G_n(y)
\right|
=
o_\Pb(1),
\]
where
\[
\mathbb G_n(y)
=
\sqrt n(\Pn-\Pb)
\phi^{\mathrm{geo}}_h(Z;y,\eta_{\theta^\dagger}).
\]
Let
\[
S_n
=
\sup_{y\in\mathcal Y_0}
\frac{|\mathbb G_n(y)|}{\sigma_{\theta^\dagger}(y)}.
\]
On the event
\[
\inf_{y\in\mathcal Y_0}\sigma_{\theta^\dagger}(y)\ge c_\sigma,
\quad
\sup_{y\in\mathcal Y_0}
|\widehat\sigma_{\theta^\dagger}(y)
-
\sigma_{\theta^\dagger}(y)|
\le c_\sigma/2,
\]
both denominators are uniformly bounded away from zero. Hence
\[
\begin{aligned}
|T_n-S_n|
&\lesssim
\sup_{y\in\mathcal Y_0}
\left|
\sqrt n
\{\widehat p^{\mathrm{geo}}_{a,h,\theta^\dagger}(y)
-
p^{\mathrm{geo}}_{a,h,\theta^\dagger}(y)\}
-
\mathbb G_n(y)
\right| \\
&\quad+
\sup_{y\in\mathcal Y_0}|\mathbb G_n(y)|
\sup_{y\in\mathcal Y_0}
|\widehat\sigma_{\theta^\dagger}(y)
-
\sigma_{\theta^\dagger}(y)|.
\end{aligned}
\]
The assumed process approximation entails asymptotic tightness in
\(\ell^\infty(\mathcal Y_0)\), and hence
\(\sup_{y\in\mathcal Y_0}|\mathbb G_n(y)|=O_\Pb(1)\)
\citep{van1996weak}. Together with uniform variance consistency and
nondegeneracy, this yields
\[
T_n=S_n+o_\Pb(1).
\]

Define the studentized multiplier statistic
\[
\widehat T_n
=
\sup_{y\in\mathcal Y_0}
\frac{|\widehat{\mathbb Z}_n(y)|}
{\widehat\sigma_{\theta^\dagger}(y)}.
\]
By the assumed multiplier approximation, the conditional law of
\(\widehat T_n\) consistently approximates that of \(S_n\). Standard
Gaussian and multiplier-bootstrap approximation results, together with
continuity at the \((1-\alpha)\)-quantile, therefore imply
\citep{chernozhukov2013gaussian,chernozhukov2016empirical}
\[
\Pb\{T_n\le\widehat c_{1-\alpha}\mid\theta^\dagger\}
\xrightarrow{\Pb}
1-\alpha
\]
when \(\theta^\dagger\) is random. If \(\theta^\dagger\) is fixed, the same
argument gives
\[
\Pb\{T_n\le\widehat c_{1-\alpha}\}
\to
1-\alpha.
\]
Thus the stated band has asymptotic simultaneous coverage \(1-\alpha\) for
\(\{p^{\mathrm{geo}}_{a,h,\theta^\dagger}(y):y\in\mathcal Y_0\}\).
Conditional convergence also implies unconditional coverage because the
conditional coverage probabilities are bounded.
\end{proof}

\subsubsection{Proof of Corollary~\ref{cor:population-geometry-band}}

\begin{proof}
Let
\[
E_n
=
\left\{
\sup_{y\in\mathcal Y_0}
|p^{\mathrm{geo}}_{a,h,\widehat\theta}(y)
-
p^{\mathrm{geo}}_{a,h,\theta_0}(y)|
\le
\Delta_{\rm geom,n}
\right\}.
\]
By assumption, \(\Pb(E_n)\to1\). On \(E_n\), the triangle inequality implies
that uniform coverage of \(p^{\mathrm{geo}}_{a,h,\widehat\theta}\) by the
uninflated band entails uniform coverage of
\(p^{\mathrm{geo}}_{a,h,\theta_0}\) by the inflated band. Hence
Theorem~\ref{thm:dis-band} gives asymptotic coverage at least \(1-\alpha\).

If \(\Delta_{\rm geom,n}=o_\Pb(n^{-1/2})\), variance nondegeneracy and uniform
variance consistency give
\[
\sup_{y\in\mathcal Y_0}
\frac{
\sqrt n
|p^{\mathrm{geo}}_{a,h,\widehat\theta}(y)
-
p^{\mathrm{geo}}_{a,h,\theta_0}(y)|
}
{\widehat\sigma_{\widehat\theta}(y)}
=
o_\Pb(1).
\]
Thus the studentized suprema for the fitted- and population-geometry targets
differ by \(o_\Pb(1)\). Multiplier validity and quantile continuity in
Theorem~\ref{thm:dis-band} therefore imply validity of the uninflated band for
the population-geometry target.
\end{proof}

\subsubsection{Proof of Theorem~\ref{thm:dss-stein-band}}

\begin{proof}
The argument is identical to that of Theorem~\ref{thm:dis-band}, with the index
set \(\mathcal Y_0\) replaced by \(\mathcal G\) and the DIS influence function
replaced by \(\phi_{\Psi,\theta^\dagger}\). When \(\theta^\dagger\) is learned
from an auxiliary sample independent of the analysis sample, all statements
below are understood conditionally on that external geometry-training
information, in probability.

Let
\[
T_n^\Psi
=
\sup_{g\in\mathcal G}
\frac{
\sqrt n\,
|\widehat\Psi_{a,h,\theta^\dagger}(g)
-
\Psi_{a,h,\theta^\dagger}(g)|
}
{\widehat\sigma_{\theta^\dagger}(g)}.
\]
Also let
\[
\mathbb G_n^\Psi(g)
=
\sqrt n(\Pn-\Pb)
\phi_{\Psi,\theta^\dagger}
(Z;g,\eta^s_{\theta^\dagger}).
\]
On the event that
\(\inf_{g\in\mathcal G}\sigma_{\theta^\dagger}(g)\ge c_\sigma\) and
\(\sup_{g\in\mathcal G}
|\widehat\sigma_{\theta^\dagger}(g)
-\sigma_{\theta^\dagger}(g)|
\le c_\sigma/2\),
\[
\begin{aligned}
\left|
T_n^\Psi-
\sup_{g\in\mathcal G}
\frac{|\mathbb G_n^\Psi(g)|}
{\sigma_{\theta^\dagger}(g)}
\right|
&\lesssim
\sup_{g\in\mathcal G}
\left|
\sqrt n
\{\widehat\Psi_{a,h,\theta^\dagger}(g)
-
\Psi_{a,h,\theta^\dagger}(g)\}
-
\mathbb G_n^\Psi(g)
\right| \\
&\quad+
\sup_{g\in\mathcal G}|\mathbb G_n^\Psi(g)|
\sup_{g\in\mathcal G}
|\widehat\sigma_{\theta^\dagger}(g)
-
\sigma_{\theta^\dagger}(g)|.
\end{aligned}
\]
The assumed process approximation entails
\(\sup_{g\in\mathcal G}|\mathbb G_n^\Psi(g)|=O_\Pb(1)\)
\citep{van1996weak}. Uniform asymptotic linearity, variance consistency, and
nondegeneracy therefore yield
\[
T_n^\Psi
=
\sup_{g\in\mathcal G}
\frac{|\mathbb G_n^\Psi(g)|}
{\sigma_{\theta^\dagger}(g)}
+
o_\Pb(1).
\]
The assumed multiplier approximation and quantile continuity, together with
standard multiplier-bootstrap results
\citep{chernozhukov2013gaussian,chernozhukov2016empirical}, therefore give
\[
\Pb\{T_n^\Psi\le\widehat c'_{1-\alpha}\mid\theta^\dagger\}
\xrightarrow{\Pb}
1-\alpha
\]
when \(\theta^\dagger\) is random, and
\(\Pb\{T_n^\Psi\le\widehat c'_{1-\alpha}\}\to1-\alpha\) when it is fixed. Hence the
stated band has asymptotic simultaneous coverage \(1-\alpha\) over
\(\mathcal G\). Conditional convergence also implies unconditional coverage.
\end{proof}

\subsubsection{Proof of Corollary~\ref{cor:population-geometry-stein-band}}

\begin{proof}
The argument parallels that of
Corollary~\ref{cor:population-geometry-band}. On the event that the
fitted-geometry band covers uniformly and the stated drift envelope holds, the
triangle inequality yields the inflated band for
\(\Psi_{a,h,\theta_0}\). Hence Theorem~\ref{thm:dss-stein-band} gives
asymptotic simultaneous coverage at least \(1-\alpha\).

If \(\Delta_{\Psi,n}=o_\Pb(n^{-1/2})\), variance nondegeneracy and uniform
variance consistency imply that the studentized suprema for the fitted- and
population-geometry functionals differ by \(o_\Pb(1)\). The multiplier and
quantile-continuity conditions therefore establish validity of the uninflated
band.
\end{proof}

\end{document}